\documentclass[acmsmall,screen]{acmart}

\usepackage{lineno}
\usepackage{graphicx}
\usepackage{xcolor}
\usepackage{listings}

\usepackage{comment}
\usepackage[inline]{enumitem}
\usepackage{wrapfig}
\usepackage{subcaption}
\usepackage{todonotes}
\usepackage[normalem]{ulem}
\usepackage{thm-restate}
\usepackage{tikz}
\usetikzlibrary {arrows,arrows.meta,automata,positioning}
\usepackage{algorithm}
\usepackage{algpseudocode}
\usepackage{bbm}
\usepackage{xparse}

\newenvironment{myparagraph}[1]{\vspace*{0.05in}\noindent{\bf #1.}}{}

\theoremstyle{definition}
\newtheorem{theorem}{Theorem}[section]
\newtheorem{definition}[theorem]{Definition}

\newtheorem{lemma}[theorem]{Lemma}
\newtheorem{proposition}[theorem]{Proposition}
\newtheorem{claim}[theorem]{Claim}

\newtheorem{example}[theorem]{Example}

\newcommand{\figlabel}[1]{\label{fig:#1}}
\newcommand{\figref}[1]{Fig.~\ref{fig:#1}}
\newcommand{\seclabel}[1]{\label{sec:#1}}
\newcommand{\secref}[1]{Section~\ref{sec:#1}}

\newcommand{\thmlabel}[1]{\label{thm:#1}}
\newcommand{\thmref}[1]{Theorem~\ref{thm:#1}}
\newcommand{\proplabel}[1]{\label{prop:#1}}
\newcommand{\propref}[1]{Proposition~\ref{prop:#1}}
\newcommand{\lemlabel}[1]{\label{lem:#1}}
\newcommand{\lemref}[1]{Lemma~\ref{lem:#1}}

\newcommand{\algolabel}[1]{\label{algo:#1}}
\newcommand{\algoref}[1]{Algorithm~\ref{algo:#1}}

\newcommand{\myBox}[2]{{\setlength{\fboxsep}{0pt}\colorbox{#1}{$\strut#2$}}}

\algtext*{EndWhile}  
\algtext*{EndIf}     
\algtext*{EndFor}    

\newcommand{\ucomment}[1]{}

\newcommand{\calQ}{\mathcal{Q}}
\newcommand{\calA}{\mathcal{A}}

\newcommand{\calT}{\mathcal{T}}

\newcommand{\calH}{\mathcal{H}} 
 
\renewcommand{\emptyset}{\varnothing}
\newcommand{\set}[1]{\{#1\}}
\newcommand{\setpred}[2]{\{#1 \,|\, #2\}}

\newcommand{\nats}{\mathbb{N}}
\newcommand{\natsp}{\nats_{>0}}

\newcommand{\estimateAndSample}{\ensuremath{\mathsf{estimateAndSample}}}

\newcommand{\reduce}{\ensuremath{\mathsf{reduce}}}
\newcommand{\round}{\ensuremath{\mathsf{roundUp}}}
\newcommand{\union}{\ensuremath{\mathsf{union}}}
\newcommand{\median}{\mathsf{median}}
\newcommand{\mom}{\mathsf{median\text{-}of\text{-}means}}
\newcommand{\preds}{\mathsf{Pred}}
\newcommand{\ancestors}{\mathsf{Ancest}}
\newcommand{\run}{\mathit{run}}
\newcommand{\can}{\mathit{can}}
\newcommand{\depclass}{\mathit{dep}}
\newcommand{\word}{\mathit{word}}

\newcommand{\vY}[4]{\mathfrak{S}^{#1}_{#2,#3}(#4)}
\newcommand{\vZ}[3]{\hat{\mathfrak{S}}^{#1}_{#2}(#3)}
\newcommand{\vW}[3]{\bar{\mathfrak{S}}^{#1}_{#2}(#3)}
\newcommand{\vM}[3]{\mathfrak{M}_{#1,#2}(#3)}
\newcommand{\Ex}{\mathsf{E}}
\newcommand{\Var}{\mathsf{Var}}
\newcommand{\ns}{\beta}
\newcommand{\nt}{\gamma}
\newcommand{\nv}{\xi}
\newcommand{\nsnt}{\alpha}
\newcommand{\threshold}{\theta}

\newcommand{\treq}{\sim_\indep}
\newcommand{\eqcl}[1]{[#1]_{\indep}}
\newcommand{\quot}[2]{{#2}/{#1}}

\newcommand{\nf}[1]{\mathsf{NF}_{#1}}

\newcommand{\memp}[1]{\operatorname{mem}_{#1}}

\newcommand{\alphabet}{\Sigma}
\newcommand{\indep}{\mathbb{I}}
\newcommand{\dep}{\mathbb{D}}
\newcommand{\cwdth}{\omega}

\newcommand{\aut}{\calA}
\newcommand{\unroll}{\textsf{u}}
\newcommand{\autunroll}{\aut^\unroll}
\newcommand{\states}{\calQ}
\newcommand{\statesunroll}{\states^\unroll}
\newcommand{\trans}{\calT}
\newcommand{\transunroll}{\trans^\unroll}
\newcommand{\qI}{q_{\mathsf{I}}}
\newcommand{\qIunroll}{\qI^\unroll}
\newcommand{\accept}{F}
\newcommand{\acceptunroll}{\accept^\unroll}
\newcommand{\langT}{L}
\newcommand{\lang}[1]{\langT(#1)}
\newcommand{\rej}{{\sf rej}}

\newcommand{\poly}[1]{\textsf{poly}(#1)}

\newcommand{\emptyword}{\lambda}  

\newcommand{\TraceMC}{\textsc{TraceMC}}
\newcommand{\TraceMCCore}{\textsc{TraceMCcore}}

\newcommand{\solutions}[1]{\textsf{solutions}(#1)}

\newcommand{\trpo}[1]{\preceq^{#1}_{\indep}}             
\newcommand{\lexord}{\prec_{\sf lex}}  

\newcommand{\extn}{\textsf{extn}}
\newcommand{\upsets}[1]{\mathcal{U}(#1)}

\newcommand{\trm}{{\sf term}}
\newcommand{\asgn}{{\sf asgn}}
\newcommand{\dmrct}{\$}
\newcommand{\asgnind}{\sf{ASGNIND}}
\newcommand{\asgns}{\sf{ASGN}}

\newcommand{\dtv}{d_{\textsf{TV}}}

\setcopyright{cc}
\setcctype{by}
\acmDOI{10.1145/3776723}
\acmYear{2026}
\acmJournal{PACMPL}
\acmVolume{10}
\acmNumber{POPL}
\acmArticle{81}
\acmMonth{1}
\received{2025-07-10}
\received[accepted]{2025-11-06}

\begin{document}

\title{Counting and Sampling Traces in Regular Languages}

\titlenote{All authors contributed equally to this research.
The symbol \textcircled{r} denotes random author order. The publicly
verifiable record of the randomization is available at
\url{https://www.aeaweb.org/journals/policies/random-author-order/search?RandomAuthorsSearch[search]=DGBMk80WL8yf}
}

\author{Alexis de Colnet}
\orcid{0000-0002-7517-6735}
\affiliation{%
  \institution{TU Wien}
  \city{Vienna}
  \country{Austria}
}
\email{decolnet@ac.tuwien.ac.at}

\author{Kuldeep S. Meel}
\orcid{0000-0001-9423-5270}
\affiliation{%
  \institution{University of Toronto}
  \city{Toronto}
  \country{Canada}
}
\email{meel@cs.toronto.edu}

\author{Umang Mathur}
\orcid{0000-0002-7610-0660}
\affiliation{%
  \institution{National University of Singapore}
  \city{Singapore}
  \country{Singapore}
}
\email{umathur@comp.nus.edu.sg}


\begin{abstract}
In this work, we study the fundamental problems of counting and sampling traces that a regular language touches. Formally, one fixes the alphabet $\Sigma$ and an independence relation $\mathbb{I} \subseteq \Sigma \times \Sigma$. The computational problems we address take as input a regular language $L$ over $\Sigma$, presented as a finite automaton with $m$ states, together with a natural number $n$ (presented in unary). For the counting problem, the output is the number of Mazurkiewicz \emph{traces} (induced by $\mathbb{I}$) that intersect the $n^\text{th}$ slice $L_n = L \cap \Sigma^n$ of $L$, i.e., traces that have at least one linearization in $L_n$. For the sampling problem, the output is a trace drawn from a distribution that is approximately uniform over all such traces. These problems are motivated by applications such as bounded model checking based on partial-order reduction, where an \emph{a priori} estimate of the size of the state space can significantly improve usability, as well as testing approaches for concurrent programs that use partial-order-aware random sampling, where uniform exploration is desirable for effective bug detection.

We first show that the counting problem is \#P-hard even when the automaton accepting the language $L$ is deterministic, which is in sharp contrast to the corresponding problem for counting the words of a DFA, which is solvable in polynomial time. We then show that the counting problem remains in the class \#P for both NFAs and DFAs, independent of whether $L$ is trace-closed. Finally, our main contributions are a \emph{fully polynomial-time randomized approximation scheme} (FPRAS) that, with high probability, estimates the desired count within a specified accuracy parameter, and a \emph{fully polynomial-time almost uniform sampler} (FPAUS) that generates traces while ensuring that the distribution induced on them is approximately uniform with high probability.
\end{abstract}



\maketitle


\section{Introduction}
\label{sec:intro}

The motivation behind this work stems from problems in algorithmic 
analysis of programs using techniques such as 
model checking that attempt to exhaustively explore the behaviors 
of these programs in order to isolate bugs.
Inherently, model checking is an expensive and resource-intensive task. 
Software reliability and quality assurance teams 
employing such exhaustive exploration techniques repeatedly suffer 
from the dilemma— \emph{how much resources must be allocated to a model checking task?} 
The usability of model checkers can thus be enhanced 
if we can quickly get a good estimate of the number of 
behaviors that the model checker will explore~\cite{usableGencMC2024}. 
The recent trend of an increased adoption of such otherwise expensive model checking
techniques in model software development practices~\cite{lightweightFM2021,vsync2021} 
make this question even more important and timely.

Another seemingly different but closely related analysis technique 
that has inspired the central question of this paper is randomized testing~\cite{RFF2024,PCT2010,SURW2025,sen2007effective},
where one draws executions from the program under test 
randomly and checks for the presence of bugs in the execution. 
For randomized testing to be effective, it is desirable that the 
distribution of behaviors induced by the testing technique be uniform, 
since otherwise the exploration may get biased. 
\emph{Can we come up with an ideal sampler?}

\subsection{The Ubiquity of Automata-Based Counting}

The answer to both challenges lies in efficient algorithms 
that estimate the number of behaviors (or execution paths) of a program. 
A solution to this counting problem also naturally lends itself 
to a solution to the uniform sampling problem: one can pick the next event 
to execute with probability proportional to the number of behaviors 
induced by each possible choice. 
In turn, both questions can elegantly be viewed through a 
common computational lens—\emph{counting words accepted by finite automata}. 
Program behaviors can be naturally represented as 
strings over an alphabet of actions, while control and data dependencies 
can be captured by transitions of an automaton. 
This automata-theoretic perspective provides a unifying framework 
for developing algorithmic solutions to several estimation and sampling 
problems that arise in program analysis, two of which we discuss next.

\begin{figure}[t]
	\centering
	\begin{subfigure}[b]{0.4\textwidth}
		\begin{center}
			\begin{lstlisting}[language=C,numbers=left,breaklines=true,xleftmargin=30pt,basicstyle=\ttfamily\footnotesize,frame=none]
x = 1;
if(x > 100) {
 while(x > 90){
  x = x-1;
 }
} else {
 y = y + 10;
}
x = y + 1;
			\end{lstlisting}
		\end{center}
		\caption{Simple imperative program}
		\figlabel{prog}
	\end{subfigure}
	\begin{subfigure}[b]{0.5\textwidth}
		\centering
		\begin{tikzpicture}[->,>={Stealth[round]},shorten >=1pt,%
			auto,node distance=2cm,on grid,semithick,
			inner sep=2pt,bend angle=40, minimum size=1mm, initial text=]
			\scalebox{0.65}{          
				\node[state, initial]   (q_0)                       {$\ell_\texttt{1}$};
				\node[state, accepting] (q_1) [right=of q_0]        {$\ell_\texttt{2}$};
				\node[state, accepting] (q_2) [below right=of q_1]  {$\ell_\texttt{3}$};
				\node[state, accepting] (q_3) [below=of q_2]        {$\ell_\texttt{4}$};
				\node[state, accepting] (q_4) [above right=of q_1]  {$\ell_\texttt{7}$};
				\node[state, accepting] (q_5) [below right=of q_4]  {$\ell_\texttt{9}$};
				\node[state, accepting] (q_6) at (7.5,0)        {$\ell_\texttt{exit}$};
				\path[->] (q_0) edge node {\texttt{x := 1}} (q_1)
				(q_1) edge node {\texttt{x <= 100}} (q_4)
				(q_1) edge node[left] {\texttt{x > 100}} (q_2)
				(q_2) edge [bend right] node[left] {\texttt{x > 90}} (q_3)
				(q_3) edge [bend right] node[right] {\texttt{x := x-1}} (q_2)
				(q_4) edge node {\texttt{y := y+10}} (q_5)
				(q_2) edge node {\texttt{x <= 90}} (q_5)
				(q_5) edge node {\texttt{x := y+1}} (q_6);
			}
		\end{tikzpicture}
		\vspace{-0.5in}
		\caption{Control Flow Automaton}
		\figlabel{cfg}
	\end{subfigure}
	\vspace{-0.1in}
	\caption{Modeling a program using its control flow automaton}
	\figlabel{prog-cfg}
	\vspace{-0.2in}
\end{figure}
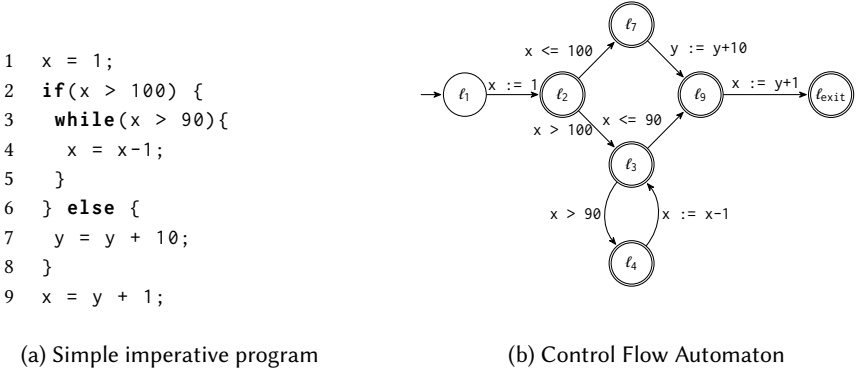

\myparagraph{Randomized testing and fuzzing}
Many modern testing frameworks generate executions of a program 
by sampling from a space of possible behaviors subject to syntactic 
or semantic constraints. 
For instance, techniques such as \emph{coverage-guided fuzzing} 
or \emph{search-based testing} can be modeled as a process that explores 
 the \emph{control flow automaton} of the automaton. 
Consider, for example, the simple imperative program and its corresponding control 
flow automaton shown in~\figref{prog-cfg}.
Here, states such as $\ell_1$, $\ell_7$, and $\ell_{\mathsf{exit}}$ 
correspond to program locations, while transitions 
are labeled with program statements (assignments or conditions). 
Each path through the automaton corresponds to a distinct test execution, 
and the ability to uniformly sample such paths 
directly influences the quality and coverage of the generated tests.
Further, recent work on statistical estimation and coverage prediction 
in testing and analysis~\cite{statisticalReachabilityLee2023,pReach2022} has shown that even coarse-grained estimates of the exploration 
space can meaningfully improve the usability and scalability of 
automated analysis by guiding resource allocation and stopping decisions.

\myparagraph{Model checking}
Likewise, in model checking, the program under analysis and the property to verify 
are jointly represented as a finite-state system whose reachable behaviors 
form a regular language over program actions. 
Each execution path through this automaton corresponds to a 
distinct program behavior. 
Estimating the number of such paths provides valuable information 
for resource planning and coverage estimation. 
Recent tools have shown that even coarse-grained estimates 
of the exploration space can substantially improve the usability 
and scalability of model checking by helping allocate computational budgets~\cite{usableGencMC2024}.


\subsection{Discounting for Equivalence}

In the context of model checking or other exploration-based testing
of concurrent programs, redundancy poses a new challenge for counting.
Several executions are essentially \emph{equivalent} to each other 
in that they exercise the same 
relative order of conflicting events 
and 
follow identical control-flow paths. 
To cater for this, modern model checkers employ 
\emph{partial order reduction}~\cite{Flanagan2005}
to prune redundant executions and aim to explore only one representative per equivalence class.
This is especially the case for recent tools
that target \emph{optimality}~\cite{Kokologiannakis2022,Abdulla2019,genmc2021,Abdulla2014},
where each equivalence class of executions is explored exactly once.

Estimating the number of distinct equivalence classes that 
intersect with the behaviors of a program is therefore a fundamental question.  
Such estimates can help anticipate the size of the exploration space, 
guide randomized schedulers toward unbiased coverage, 
and provide progress metrics for partial-order-based exploration.  
However, the current landscape of solutions remains unsatisfactory:
existing approaches either rely on Monte Carlo estimators
with high variance~\cite{usableGencMC2024},
restrict attention to specialized language families~\cite{tapct},
or employ heuristic counting schemes~\cite{sen2007effective,SURW2025}
that provide no guarantees on accuracy or bias.

The primary conceptual contribution of this work is a
systematic study of the problems of counting and sampling 
\emph{modulo} Mazurkiewicz trace equivalence~\cite{Mazurkiewicz87}—arguably
the most common notion of equivalence exploited in model checking and concurrency testing.
Two words over an alphabet $\alphabet$ are said to be equivalent 
if one can be transformed into the other by repeatedly commuting 
neighboring independent actions, as determined by a fixed 
\emph{independence relation} $\indep \subseteq \alphabet \times \alphabet$.
Each equivalence class under this relation is called a \emph{trace}.
Equipped with this natural notion of equivalence, 
the fundamental algorithmic questions we study can be phrased as follows.

\begin{description}
	\item[\quad Counting Modulo Trace Equivalence]
	Fix an alphabet $\alphabet$ and an independence relation $\indep \subseteq \alphabet \times \alphabet$.
	Given an automaton $\aut$ with $m$ states and a length $n$ (presented in unary),
	determine the number of distinct traces (under $\treq$)
	that intersect with the set of words of length $n$ accepted by $\aut$.
\end{description}

The corresponding \emph{uniform sampling} question asks for a procedure that samples 
executions of an automaton so that every equivalence class 
is chosen with equal probability mass. 
This requirement is fundamental to unbiased randomized testing.

\begin{description}
	\item[\quad Uniform Sampling Modulo Trace Equivalence]
	Fix an alphabet $\alphabet$ and an independence relation $\indep \subseteq \alphabet \times \alphabet$.
	Given an automaton $\aut$ with $m$ states and a length $n$ (presented in unary),
	sample strings of length $n$ accepted by $\aut$ so that for any
	two distinct traces $t_1, t_2$ intersecting $L_n(\aut)$, $\Pr[\text{a string from } t_1 \text{ is sampled}]
	= \Pr[\text{a string from } t_2 \text{ is sampled}]$,
	and all other traces receive probability mass~0.
\end{description}

\subsection{From \#P-hardness to efficient approximation}

\myparagraph{Complexity of exact counting}
Observe that the trace-counting problem generalizes the well-studied
\#NFA (resp.\ \#DFA) problem of counting the number of words of a fixed length~$n$ 
accepted by a given non-deterministic (resp.\ deterministic) finite automaton.
In particular, when the automaton represents a \emph{trace-closed} language—
that is, when $\lang{\aut} = \eqcl{\lang{\aut}}$ where 
$\eqcl{L} = \setpred{w \in \alphabet^*}{\exists w' \in L,\, w \treq w'}$—
the two problems coincide;
this can happen when, say, $\indep = \emptyset$.
In such cases, the counting problem modulo trace equivalence
can be parsimoniously, in polynomial time, reduced to the corresponding
\#NFA or \#DFA problem.
Indeed, one can intersect the language of~$\aut$ with the regular language
$L_{\nf{}}$ of (lexicographical) normal forms of~$\alphabet^*$,
which has a compact representation and contains exactly one word from each
trace class that intersects $\lang{\aut}$~\cite{diekert1995book}.

While programs in common multi-threaded imperative languages such as C/C++ are written
in a manner that their semantics induce trace-closed languages,
the setting of non trace-closed languages
nevertheless naturally arises in practical exploration based testing,
for instance, when executions are restricted by preemption bounds
or by syntactic depth constraints,
where only a subset of the representatives of each trace class is included.
\emph{Preemption-bounded model checking} restricts executions to those
containing at most~$k$ context switches when exploring behaviors of a concurrent program.
Recent algorithms~\cite{Marmanis2023Reconciling,contextbounding2007} on combining preemption bounding with partial
order reduction effectively attempt to explore one representative 
per equivalence class under some preemption bound $k$.
The resulting set of schedules~$B_k$ induced by this restriction is not trace-closed.
For example, when $k=0$, the schedule $a_1{\cdot}a_2{\cdot}b$ may 
satisfy the bound while the equivalent schedule $a_1{\cdot}b{\cdot}a_2$ does not
(here the actions $a_1, a_2$ are performed by some thread that is different from the thread of $b$).
Similarly, in \emph{randomized testing} frameworks such as PCT~\cite{PCT2010},
executions are sampled under bounds on the number or depth of ordering constraints
(and thus induce languages that are regular, but not trace-closed);
recent algorithms have attempted to make them trace-aware~\cite{tapct}.
Finally, specifications in the style of \emph{choreographic programming}~\cite{Honda2016MST,montesi2023introduction,Li2023Complete},
often only specify a totally ordered sequence of message exchanges among multiple
participants (e.g., $\mathtt{A}\!\to\!\mathtt{B} {:} \mathtt{ping};\ \mathtt{C}\!\to\!\mathtt{D}{:}\mathtt{tok}$),
while the projected system executes asynchronously.
Such specifications tend to correspond to regular languages which are not
trace-closed, and testing or analysis techniques for a model specified in the form of a choreography
must therefore work with a non–trace-closed language.

\myparagraph{Complexity of exact counting}
How hard is trace counting in general (i.e., when the language is not trace-closed)?
Is it strictly harder than counting words?
We answer this question in the affirmative.
We show that trace counting is \#P-hard even when the automaton is deterministic
(\thmref{DFA-hard}),
in sharp contrast to the classical \#DFA problem, which is solvable in polynomial time.
We further show that it is also in \#P even when the automaton is non-determinisitic
(\thmref{NFA-sharp-P-membership}).

\myparagraph{Approximately counting and sampling traces}
Given that exact counting is hard, it is natural to 
ask whether we can instead approximately solve the 
trace-counting problem efficiently.
Recent breakthroughs have shown that the classical \#NFA problem 
admits a \emph{fully polynomial-time randomized approximation scheme} (FPRAS)%
~\cite{MCMNFA2024,ACJR21,ACJR19,MeeldCPODS2025}.
It is tempting to ask whether the approximate trace-counting problem 
can be similarly reduced to \#NFA.
Unfortunately, this is not the case.
Unlike the trace-closed setting, an arbitrary regular language
may fail to contain a normal-form representative for every class,
so a naive intersection with $L_{\nf{}}$ can yield a language
that covers only a subset of the original trace classes.
Closing the automaton under trace equivalence (so as to safely intersect with $L_{\nf{}}$)
is also not a viable approach:
the trace closure of a regular language need not be regular
—and may not even be context-free.
For instance, the trace closure of $(abc)^*$ under full independence 
(i.e., $\indep = \setpred{(\sigma_1,\sigma_2)}{\sigma_1, \sigma_2\in \alphabet, \sigma_1 \neq \sigma_2}$)
is the set of all words with equal numbers of $a$'s, $b$'s, and $c$'s.

From a technical standpoint, the core difficulty of approximate counting
stems from the fact that trace classes can vary drastically in size.
For example, the language $L = \set{abcabc, aabbcc, aaaaaa}$ 
contains only two trace classes, but the first class
$\eqcl{abcabc}$, contributes two words while the second,
$\eqcl{aaaaaa}$, contributes only one.
Consequently, counting techniques based on uniform sampling of individual words,
as employed in recent algorithms for \#NFA,
become ineffective in our setting.
For the same reason, the approximate variant of uniform sampling—%
that is, \emph{almost-uniform generation of traces}—%
is also technically challenging to achieve efficiently.

\myparagraph{Main Contributions}
Despite these challenges, we show that approximate counting and sampling
modulo trace equivalence are tractable problems.
The primary contribution of our approach is a fully polynomial-time
randomized approximation scheme (FPRAS) for counting trace 
equivalence classes in regular languages (\thmref{main_result}). 
Our algorithm provides an $(\varepsilon,\delta)$-approximation in time polynomial in 
$m$, $k$, $n^{\cwdth}$, $1/\varepsilon$, and $\log(1/\delta)$, where 
$m$ is the number of states in the automaton, $n$ is the desired length,
$k$ is the alphabet size, 
and $\cwdth$ is a constant determined by the independence relation~$\indep$
and is bounded above by $k$.
This result extends recent advances in approximate counting for
\#NFA~\cite{MCMNFA2024,ACJR21,ACJR19,MeeldCPODS2025}
to the more intricate setting of counting modulo equivalence.

Our algorithmic approach builds upon the dependence-analysis framework
introduced by de~Colnet and Meel for \#nFBDD~\cite{MdC24a} and
\#NFA~\cite{MeeldCPODS2025} counting.
Their key insight was to track dependencies among partial solutions through
structured derivations, enabling efficient probabilistic estimation.
However, a direct adaptation of their method fails in our setting:
the independence relation $\indep$ introduces \emph{non-local effects},
where commuting symbols at one position can enable or disable reorderings at distant
positions in the word.
Moreover, the sizes of equivalence classes vary widely, meaning that
uniform sampling over words does not correspond to uniform sampling over traces.
Overcoming these challenges requires a new probabilistic analysis and
a refined notion of canonical derivation runs, which form the core
technical contributions of this paper.

Finally, we leveraging inter-reducibility of approximate counting
and sampling~\cite{JVV1986}, and also obtain a \emph{fully polynomial-time almost-uniform sampler} (FPAUS)
for generating traces from regular languages modulo trace equivalence
(\thmref{main-reduction}).

\myparagraph{Organization}
The rest of the paper is organized as follows.
\secref{prelim} 
establishes notation and fundamental concepts including trace 
equivalence and predictive membership. 
In \secref{exact}, we study the exact counting problem and establish upper
and lower bounds for the trace counting problem for both
NFAs and DFAs.
Before we present our FPRAS, we first present, in~\secref{sampling}, 
our algorithm for approximately uniformly sampling traces from an NFA,
by reducing it to our FPRAS.
We then move on to presenting our FPRAS in~\secref{fpras},
and present the notion of canonical runs in \secref{canrun}, 
which plays a central role in our technical analysis (\secref{analysis}). 
Finally, we conclude in~\secref{conclusion}.

\section{Preliminaries}
\seclabel{prelim}

In the following, we fix a finite alphabet $\alphabet$.
The set of all strings over $\alphabet$ is denoted using 
the regular expression $\alphabet^*$. 
Throughout this work, we will assume that $|\alphabet|\in O(1)$.
We use $\emptyword$ to denote the empty string.
Let $\aut = (\states,\alphabet,\trans,\qI,\accept)$ 
be a non-deterministic finite automaton (NFA)
working over alphabet $\alphabet$, with states $\states$, 
a transition relation $\trans \subseteq \states \times \alphabet \times \states$,
an initial state $\qI \in \states$ and a set of final states $F \subseteq \states$. 
For $q \in \states$, we denote by $L(q)$ the set of words that reach $q$ in $\aut$. 
We say that a state $q'$ is a predecessor of state $q$
if  there is a transition $(q',\alpha,q) \in \trans$ for some $\alpha$. 
We will use $\preds(q)$ to denote the set of predecessors of state $q$.
Finally, $\aut$ is deterministic, if for each $q\in \states, a\in\alphabet$,
there is at most one $q'$ such that $(q, a, q') \in \trans$.

\myparagraph{Unrolled automaton}{
Many of our algorithms work on the unrolled version of the input NFA.
Formally, let $\aut = (\states,\alphabet,\trans, \qI,\accept)$ 
be an NFA and let $n$ be a positive integer. 
Then, its \emph{unrolled automaton} $\autunroll_n$ 
accepts words of length $n$, and contains
states of the form  $q^i$ , where  $q \in \states$ and $i \in \set{0, 1, \ldots, n}$,
while ensuring that $L(q^i) = L(q) \cap \Sigma^i$ 
for every $q \in \states$ and every $0 \leq i \leq n$.
More formally, let $\states^i = \setpred{q^i}{q \in \states,\, L(q) \cap \alphabet^i \neq \emptyset}$.
and let $\statesunroll_n = \states^0 \cup \states^1 \cup \states^2 \cup \dots \cup \states^n$;
we will assume w.l.o.g that $\aut^n$ is not empty. 
Then, the unrolled automaton 
is the tuple $\autunroll_n = (\statesunroll_n,\alphabet,\transunroll_n, \qIunroll,\acceptunroll_n)$
whose initial state is $\qIunroll = \qI^0$, final states
are $\acceptunroll_n = \setpred{q^n}{q\in F}$, 
and transition relation is 
$\transunroll_n = \setpred{(q^i,\alpha,(q')^{i+1})}{(q,\alpha,q') \in \trans, q^i \in \states^i}$. 
In most cases, the subscript $n$ will be clear from context
and we will drop it, writing $\autunroll = (\statesunroll, \alphabet, \transunroll, \qIunroll, \acceptunroll)$.
We refer to the set $\states^i$ as the $i^{\text{th}}$ 
\emph{level of the automaton} and use the notation
$\states^{\geq i} = \bigcup_{i \leq j \leq n} \states^j$;
the notations for $\states^{\leq i}$, $\states^{< i}$ and $\states^{> i}$ are
defined analogously. 
For simplicity of notation, we also often omit the superscript $i$ 
in the name of states (so we write $q \in \statesunroll$ instead of the more
explicit $q^i \in \statesunroll_i$; 
when we need to emphasize that $q$ is in the 
$i^\text{th}$ level, we will write $q \in \calQ^i$). 
Unrolled automata can be represented graphically 
as directed acyclic graphs, see for instance \figref{example} 
for an unrolled automaton over $\alphabet = \set{a,b,c}$ 
and $n = 4$.
}

\begin{proposition}
\proplabel{unrolling}
Given an NFA $\aut = (\states,\alphabet,\trans,\qI,\accept)$ 
and a positive integer $n$, the unrolled automaton 
$\autunroll = (\statesunroll,\alphabet,\transunroll,\qIunroll,\acceptunroll)$ 
can be constructed in time $O(n|\trans|)$. 
\end{proposition}

\myparagraph{Trace equivalence}{
Trace theory, popularized by Antoni Mazurkiewicz~\cite{Mazurkiewicz87}, 
has emerged as a popular framework for formalizing behaviors of
concurrent programs as sequences or strings
of actions, together with an equivalence obtained by
repeated commutations of neighboring independent actions.
Formally, an \emph{independence relation} $\indep \subseteq \alphabet\times \alphabet$
is a symmetric and irreflexive binary relation on $\alphabet$;
its complement is often referred to as \emph{dependence relation}
$\dep = \alphabet\times \alphabet \setminus \indep$.
Together, $(\alphabet, \indep)$ is referred to as a \emph{concurrent alphabet}.
The trace equivalence induced by $(\alphabet, \indep)$, denoted $\treq$
is the smallest equivalence on $\alphabet^*$ such that
for every $w_1, w_2 \in \alphabet^*$ and for every $(a_1, a_2) \in \indep$,
we have $w_1 \cdot a_1 \cdot a_2 \cdot w_2 \treq w_1 \cdot a_1 \cdot a_1 \cdot w_2$.
Each class of this equivalence relation is often called a \emph{trace}.
We will use $\eqcl{w}$ to denote the equivalence class (or trace)
of a string $w \in \alphabet^*$, i.e.,
$\eqcl{w} = \setpred{w' \in \alphabet^*}{w \treq w'}$.
We will overload the notation and denote by $\eqcl{L}$ the trace closure of the language $L$,
i.e., $\eqcl{L} = \bigcup\limits_{w \in L} \eqcl{w}$.
The quotient $\quot{\treq}{L} = \setpred{\eqcl{w}}{w \in L}$ 
of a language $L$ with respect to the
trace equivalence $\treq$ is the set of classes contained in $L$.
An important parameter in our work is
the \emph{width of the concurrent alphabet} $(\alphabet, \indep)$, denoted by $\cwdth$, 
 and is the size of the largest set $\alphabet' \subseteq \alphabet$
such that for all $a \neq a' \in \alphabet'$, we have $(a, a') \in \indep$.
For two traces $t_1, t_2$, we use $t_1 \cdot t_2$ to denote the 
unique class $\setpred{w}{\exists w_1 \in t_1, w_2 \in t_2, w_1 \cdot w_2 \treq w}$.
Please refer to~\cite{diekert1995book} for a more thorough survey of trace theory.
}

\myparagraph{Counting modulo equivalence}{
The key question we ask in this work is the problem of counting
the number of traces that the $n^\text{th}$ slice of a given word language touches.
Formally, for a language $L \subseteq \alphabet^*$ and
$n \in \nats$, we use $L_n = L \cap \alphabet^n$ to denote its $n^\text{th}$ slice,
i.e., the set of those words in $L$ whose size is $n$.
The problem of counting the traces (induced by $\indep$) of a given language $L$ modulo $\treq$
with respect to a given slice $n$ asks to return the size of the set
$\quot{\treq}{L_n}$ of equivalence classes of $\treq$ that intersect with $L_n$.
}

\myparagraph{Approximate counting and sampling}{
In general, a counting problem $P$ for a binary relation
$R \subseteq \alphabet^* \times \alphabet^*$, 
 asks to determine, on input $x$,
the size $N(x)$ of the set $\solutions{x} = \setpred{y \in \alphabet^*}{(x, y) \in R}$;
in our setting, the relation $R$ is the set of pairs 
$((\aut, \indep, n), t)$ where $\aut$ is an NFA over $\alphabet$, $\indep$
is an independence relation over $\alphabet$, $n$ is a number in unary,
and $t$ is a trace (i.e., equivalence class of $\treq$) such that $t \cap \alphabet^n \cap L(\aut) \neq \emptyset$.
A \emph{fully polynomial-time randomized approximation scheme} (FPRAS) for
the problem $P$ is a randomized algorithm that, given $x$ and two real parameters 
$\varepsilon > 0$ and $\delta > 0$, runs in time polynomial in 
$|x|$ (the size of $x$), $\varepsilon^{-1}$ and $\log(\delta^{-1})$ 
and produces a random variable $\tilde{N}$ with the guarantee that
$$
\Pr[(1 - \varepsilon)N(x) \leq \tilde{N} \leq (1 + \varepsilon)N(x)] \geq 1 - \delta
$$
The uniform sampling problem for a binary relation $R$ asks to build a generator
that, on input $x$ (with $\solutions{x} \neq \emptyset$), outputs an element
of $\solutions{x}$ uniformly at random.
A \emph{fully polynomial-time almost uniform sampler} (FPAUS) for relation
$R$ is an algorithm that takes an input $x$ (with $\solutions{x} \neq \emptyset$)
and a number $\delta \in (0, 1)$
as input, runs in time polynomial in $|x|$ and $\log(\delta^{-1})$ 
and either returns a $y \in \solutions{x}$ or $\bot$ (meaning that the algorithm failed), 
and is such that the total variation distance between the uniform distribution $U_x$
on $\solutions{x}$ and the distribution $\Pi_x$ induced by the sampler
is  at most $\delta$. 
Formally, let $\Pi_x$ be the distribution induced by the sampler, and
let $U_x$ be the uniform distribution, i.e., $U_x(y) = 1/N(x)$ for every $y \in \solutions{x}$ and $U_x(\bot) = 0$.
Then the FPAUS satisfies the following guarantee (here $\dtv$ denotes total variation distance):
$$
\dtv(\Pi_x,U_x) = \frac{1}{2}\sum_{y \in \solutions{x} \cup \set{\bot}} \left|\,\Pi_x(y) - U_x(y)\, \right| \leq \delta
$$

\myparagraph{Median of Means} 
We will often use standard statistical estimation techniques
such as the median-of-means technique.
Let $\beta,\gamma \in \nats$ and let $\alpha = \beta\cdot \gamma$.
Let $X_1,\dots,X_{\alpha}$ be real-valued independent random variables 
with the same (unknown) mean $\mu = \Ex[X_i]$ (for every $i$) 
and variance $\sigma^2 = \Var[X_i]$ (for every $i$). 
The medians-of-means estimator of $\mu$ is a well-known statistical estimator
that tightens as the variance decreases. 
The estimator, parametrized by $\beta$ and $\gamma$, 
groups the random variables into $\gamma$ disjoint batches of size $\beta$, 
computes the average value of each batch, and finally returns the median of the averages
as an estimate to the actual mean $\mu$. 
Formally, for $1 \leq i \leq \gamma$ let $Y_i = \frac{1}{\beta}\sum_{j = 1}^\beta X_{j+(i-1)\cdot \gamma}$, 
then the median-of-means estimator for $\mu$ is the random variable $Z = \median(Y_1,\dots,Y_\gamma)$.

\begin{restatable}{lemma}{MoM}
\lemlabel{medians_of_means}
Let $Z$ be the median-of-means estimator as defined above,
and let $\epsilon > 0$.
We have, $\Pr\left[\left|Z - \mu\right| > \epsilon\right] \leq \exp\left(-\gamma\left(1 - \frac{2\sigma^2}{\epsilon^2\beta}\right)\right)$.
\end{restatable}
\noindent For completeness, the proof of~\lemref{medians_of_means} appears in \ucomment{our companion technical report}.

\myparagraph{Normal Forms}{
Our algorithms routinely leverage succinct representations
of equivalence classes, of which many have emerged.
A common invariant~\cite{blass1984equivalence} often studied
is the \emph{trace partial order}:
for a string $w$ over the concurrent alphabet $(\alphabet, \indep)$, 
the trace partial order $\trpo{w}$ of $w$ is the smallest partial order
over the set of labeled indices 
${\sf LABIND}_w = \setpred{(a, i) \in \alphabet \times \natsp}{1 \leq i \leq \#_{a}(w)}$,
such that $(a, i) \trpo{w} (b, j)$ iff $(a, b) \not\in \indep$
and the $i^\text{th}$ occurrence of $a$ in $w$ appears before the $j^\text{th}$ occurrence of $b$ in $w$. 
Here, $\#_{a}(w)$ represents the number of occurrences of the letter $a$ in the string $w$.
We remark that $\trpo{w} = \trpo{w'}$ iff $w \treq w'$, and further,
the set of linearizations of $\trpo{w}$ is exactly the set $\eqcl{w}$.
Alternatively, one can represent traces using normal forms~\cite{blass1984equivalence} 
or representative elements from them.
Let us fix $\lexord \subseteq \alphabet \times \alphabet$
to be a strict total order on $\alphabet$, and 
let us abuse the same notation to denote the induced lexicographic ordering on the set of all strings.
The lexicographic normal form of a trace $t$, denoted $\nf{t}$
is the smallest string (according to $\lexord$) 
in the class $t$.
We will denote the set of all lexicographic normal
forms with $L_{\nf{}} = \setpred{\nf{\eqcl{w}}}{w \in \alphabet^*}$.
Other normal forms such as the Foata normal form exist, but we will skip discussing them here.
We will also often abuse notation and 
use $\nf{w}$ to denote $\nf{\eqcl{w}}$.
It is known that the set of all lexicographic normal forms with respect to a given concurrent alphabet 
form a regular language:
\begin{proposition}[\cite{diekert1995book}]
\proplabel{lexnormalform-regular}
For a concurrent alphabet $(\alphabet, \indep)$ and a lexicographic ordering $\lexord$
on $\alphabet$, the induced set $L_{\nf{}}$ of normal forms is a regular language.
Further, $L_{\nf{}}$ is  prefix-closed,
i.e., for every $u, v, w \in \alphabet^*$ if $w = u \cdot v$
and if $w \in L_{\nf{}}$, then $u \in L_{\nf{}}$.
\ucomment{State size of automaton.}
\end{proposition}
}

\myparagraph{Checking equivalence and membership}{
Another key ingredient in our proposed algorithms will 
be an algorithm for checking equivalence of strings (`is $w_1 \treq w_2$?').
This can be done by constructing the transitive reduction of the trace partial orders
of the two strings and checked for equality, giving us the following:
\begin{proposition}
Given a concurrent alphabet $(\alphabet, \indep)$, and strings $w_1, w_2 \in \alphabet^*$
as input, the problem of checking if $w_1 \treq w_2$ can be solved
in time $O(|\alphabet|\cdot |\indep| \cdot \max \set{|w_1|, |w_2|})$.
\end{proposition}
Yet another, and a perhaps more important ingredient we use is an algorithm for the 
\emph{predictive membership problem} or \emph{membership against a trace language},
which asks if, for a given word $w$, we have $w \in \eqcl{L}$. 
This was shown to be solvable in polynomial time~\cite{Bertoni1989} 
when the regular language is considered to be fixed 
(i.e., not considered part of the input), however it is easy to extrapolate the results
to show that in fact, the problem remains polynomial-time solvable
as long as the alphabet is fixed, while the NFA or DFA representation of the input
is counted as part of the input:
\begin{theorem}[\cite{Bertoni1989}]
\thmlabel{predictive-membership-complexity}
Given as input a concurrent alphabet $(\alphabet, \indep)$ of width $\cwdth$, an NFA $\aut$ with $m$ states
and string $w \in \alphabet^*$ of length $n$, 
the problem of checking if $w \in \eqcl{\lang{\aut}}$ can be solved
in time $O(\cwdth \cdot m \cdot n^{\cwdth})$.
\end{theorem}
The dependence of $n^{\cwdth}$ above bound was shown to be optimal, 
conditional on the widely believed
Strong Exponential Time Hypothesis (SETH)~\cite{ang2024predictive}. 
In our algorithms, we will often refer to
an oracle $\memp{\treq}(\aut,q,w)$ that, 
given an automaton $\aut$, a state $q$ of $\aut$, 
and a word $w \in \alphabet^*$, 
decides the membership query `is $w \in \eqcl{L(q)}$?'.
}


\section{Complexity of Exact Counting}
\seclabel{exact}

We now study the computational complexity of counting trace 
equivalence classes for some slice of a given regular language, presented as an automaton.
First, we show that the problem remains in the \#P complexity class, which is also the complexity in which one can determine the vanilla problem of \#NFA.

\begin{theorem}
	\thmlabel{NFA-sharp-P-membership}
	Fix a concurrent alphabet $(\alphabet, \indep)$.
	Given a non-deterministic finite automaton $\aut$ accepting 
	language $L$, an independence relation $\indep$, and a positive 
	integer $n$ presented in unary, the problem of determining
	the size $|\quot{\treq}{L_n}|$ is in \#P.
\end{theorem}
\begin{proof}
	We construct a non-deterministic Turing 
	machine $M$ that operates as follows. 
	Given as input the NFA $\aut$ with $m$ states and the 
	length $n$ in unary, the machine $M$ non-deterministically guesses
	a word $w$ of length $n$ and accepts if both the following are true:
	\begin{enumerate*}[label=(\roman*)]
	\item $w$ is lexicographically the smallest element in $\eqcl{w}$, and 
	\item $w \in \eqcl{L}$.
	\end{enumerate*}
	The first check can be performed in time $O(|w|)$ since the set of all lexicographic normal forms
	over $\alphabet^*$ is a regular language (\propref{lexnormalform-regular}), 
	whose automaton has constant size, and checking membership in such an automaton takes linear time.
	The second check can be performed in time $O(\cwdth \cdot m \cdot |w|^{\cwdth})$ (\thmref{predictive-membership-complexity}).
	Thus, $M$ is a non-deterministic \emph{polynomial-time} Turing machine.
	Next, observe that the number of accepting paths of $M$ on input $\aut$ and $n$
	is equal to $|\quot{\treq}{L_n}|$, because each accepting path has a one-to-one correspondence with $\quot{\treq}{L_n}$.
	This establishes membership in \#P.
\end{proof}

We now ask how hard is the problem. 
Of course, in the case when $\indep=\emptyset$, 
the problem trivially becomes the \#NFA problem which is already \#P-hard \ucomment{cite}. 
Interestingly, we show that the trace counting 
problem remains hard even when the language is presented as a DFA.
This is in stark contrast to the more traditional \#DFA problem 
which can be solved in polynomial time,
using a simple dynamic programming algorithm.

\begin{theorem}
\thmlabel{DFA-hard}
	Fix a concurrent alphabet $(\alphabet, \indep)$.
	Given a DFA $\aut$ accepting 
	language $L$, an independence relation $\indep$, and a positive 
	integer $n$ presented in unary, the problem of determining $|\quot{\treq}{L_n}|$ is \#P-hard.
\end{theorem}

\begin{proof}
We show that counting traces over DFA modulo equivalence induced by $\indep$
is \#P-hard via a parsimonious reduction from \#DNF, which is known to be \#P-hard \ucomment{cite}. \\

\noindent
\underline{\em Reduction.}
Let $\phi = T_1 \lor T_2 \lor \cdots \lor T_k$ be a DNF formula over propositions 
$X = \set{x_1, \ldots, x_n}$ with terms $T_1, T_2, \ldots, T_k$.
We construct a DFA $\aut_\phi$ over a fixed 
concurrent alphabet $(\alphabet, \indep)$ such that all
strings accepted by $\aut$ have length $n$, and more importantly,
$|\quot{\treq}{\lang{\aut_\phi}}|$ equals the number of 
satisfying assignments of $\phi$.
Let use fix the alphabet to be $\alphabet = \set{a,b,0,1,\dmrct}$
and the following independence relation:
\[\indep = \set{(a,b), (b,a)}\]
Observe that the above independence relation marks 
the letters  $0$ and $1$ to be dependent with each other and with both of $a$ and $b$.

Our construction ensures that each word $w$ accepted by $\aut_\phi$ 
is of the form $w = \pi_1 \cdot \dmrct \cdot \pi_2$,
where the prefix $\pi_1$ has length $k$ and contains exactly $k-1$ occurences of $a$, 
 and exactly one occurence of $b$,
 while the suffix $\pi_2$ is a word of length $n$ over the alphabet $\set{0, 1}$.
The position $i$ of $b$ in the prefix $\pi_1$ intuitively represents 
the term $T_i$ that the assignment corresponding to $w$ sets to true. 

The DFA $\aut_\phi= (\states,\alphabet,\trans,\qI,\accept)$ can now be formalized as follows.
The states are partitioned as: $\states = \states_{\trm} \uplus \states_{\asgn} \uplus \set{q_\rej}$.
The states in $\states_{\trm}$ read the prefix of length $k$ (containing letters $a$ and $b$),
the states in $\states_{\asgn}$ read the remaining suffix, 
while the state $q_\rej$ is a dedicated reject state and is an absorbing state.
Each state in $\states_{\trm}$ tracks:
\begin{enumerate*}[label=(\alph*)]
	\item the number of $a$'s read so far (from 0 to $k-1$),
	\item whether the letter $b$ has been read yet, and
	\item if the letter $b$ was read, the position at which it appeared ($1$ to $k$).
\end{enumerate*}
Formally, 
\[
\states_{\trm} = \setpred{q_i}{0 \leq i \leq k-1} \uplus \setpred{r_{i,\beta}}{0 \leq i \leq k-1, 1 \leq \beta \leq k}
\]
where:
\begin{enumerate*}[label=(\alph*)]
	\item the state $q_i$ represents the fact that we have read $i$ $a$'s but no $b$ yet, and
	\item the state $r_{i,\beta}$ represents the fact that we have read $i$ $a$'s and also read $b$ at position 
	$\beta$.
\end{enumerate*}
The states in $\states_{\asgn}$ represent the assignment of interest and track how far
we have finished reading the assignment. Formally,
\[
\states_{\asgn} = \setpred{p_{T,\ell}}{T \text{ is a term in } \phi, 0 \leq \ell \leq n}
\]
where $p_{T,\ell}$ represents that the term of choice is $T$
and that the first $\ell$ bits of the assignment have been read so far.
We now describe the transition function for the DFA $\aut_\phi$.
The transitions within $\states_\trm$ track the number of $a$'s and the position of $b$'s
as follows:
\begin{itemize}
	\item $\trans(q_{i}, a) = q_{i+1}$ for $i < k-1$ (reading the $(i+1)^\text{th}$ $a$ before reading any $b$)
	\item $\trans(q_{i}, b) = r_{i,i+1}$ for $i \leq k-1$ (reading the first $b$ at position $i+1$)
	\item $\trans(r_{i,\beta}, a) = r_{i+1,\beta}$ for $i \leq k-1$ (reading remaining $a$'s after the first $b$, preserving position $\beta$)
\end{itemize}
After having read $k-1$ $a$'s and one $b$, the automaton transitions to a state from $\states_\asgn$:
\[
\trans(r_{k-1,\beta}, \dmrct) = p_{T_\beta,0}, \text{for every } 1 \leq \beta \leq k
\]
For each term $T$ in $\phi$ and $1 \leq \ell \leq n$:
\begin{itemize}
	\item If $x_\ell$ appears positively in $T$: 
	$\trans(p_{T,\ell-1}, 1) = p_{T,\ell}$
	\item If $x_\ell$ appears negatively in $T$: 
	$\trans(p_{T,\ell-1}, 0) = p_{T,\ell}$ 
	\item If $x_\ell$ does not appear in $T$:
	$\trans(p_{T,\ell-1}, b) = p_{T,\ell}$ for $b \in \set{0,1}$
\end{itemize}
All transitions not described above go to state $q_\rej$.
The initial state is $q_{0} \in \states_\asgn$ and the accepting states are:
\[
F = \setpred{p_{T,n}}{T \text{ is a term in } \phi}
\]

\noindent
\underline{\em Parsimony.}
In the following, we use $\asgns \subseteq \set{0,1}^n$ to be the set of
all satisfying assignments of $\phi$,
and use $\asgnind$ to be the set of all pairs $(\alpha, i)$
such that $\alpha \in \set{0,1}^n$ is a satisfying assignment of $\phi$ and the term
$T_i$ in $\phi$ evaluates to true under $\alpha$.

\begin{claim}
There is a bijection between $\quot{\treq}{\lang{\aut_\phi}}$ and $\asgns$.
\end{claim}

\begin{proof}
First observe that:
\[
	\lang{\aut_\phi} = \setpred{a^i\cdot b\cdot a^{k-i-1}\cdot \dmrct \cdot \alpha}{(\alpha, i) \in \asgnind}
\]
Now, we remark that for any two strings $w = a^i\cdot b\cdot a^{k-i-1}\cdot \dmrct \cdot \alpha, w' = a^j\cdot b\cdot a^{k-j-1}\cdot \dmrct \cdot \alpha$ (where $i\neq j$), we have that $w \treq w'$ since $(a, b) \in \indep$.
Further, for any two strings $w = a^i\cdot b\cdot a^{k-i-1}\cdot \dmrct \cdot \alpha, w' = a^j\cdot b\cdot a^{k-j-1}\cdot \dmrct \cdot \beta$, if $\alpha \neq \beta$, then we must have
$(w, w') \notin \treq$.
Now, the desired bijection $f : \quot{\treq}{\lang{\aut_\phi}} \to \asgns$ is easy to see ---
$f(\eqcl{a^{k-1}\cdot b\cdot \dmrct \cdot \alpha}) = \alpha$.
Let us first argue why $f$ is injective.
Let $t_1 = \eqcl{a^{k-1}\cdot b\cdot \dmrct \cdot \alpha_1}$ and 
$t_2 = \eqcl{a^{k-1}\cdot b\cdot \dmrct \cdot \alpha_2}$ 
such that $t_1 \neq t_2$. 
In this case we must have $\alpha_1 \neq \alpha_2$ and thus $f(t_1) \neq f(t_2)$.
Let us now argue why $f$ is surjective.
Let $\alpha \in \asgns$ and let $T_i$ be some term that evaluates to true under $\alpha$.
We know that $(\alpha, i) \in \asgnind$ and thus $w = a^i\cdot b\cdot a^{k-i-1}\cdot \dmrct \cdot \alpha \in \lang{\aut_\phi}$. In this case we have $f(\eqcl{w}) = \alpha$.
\end{proof}

\noindent
\underline{\em Time complexity of the reduction.}
The reduction runs in polynomial time:
\begin{enumerate*}[label=(\alph*)]
	\item the prefix component has $O(k^2)$ states tracking 
	combinations of $a$-count and $b$-position,
	\item the assignment component has $O(kn)$ states for checking 
	each term, and
	\item the alphabet size is constant.
\end{enumerate*}
\end{proof}

\section{From Approximate Trace Counting to Almost Uniform Trace Sampling}
\seclabel{sampling}

Before we proceed to describing (in \secref{fpras}) our more intricate FPRAS, 
in this section, we describe how we can use
it to build an almost uniform sampler.
At a high level, our sampler generates traces (equivalence classes on strings) 
that touch $L_n(\aut)$ by (equivalently) generating their normal forms (strings).
Each normal form we generate is iteratively constructed,
starting from the shortest prefix $\emptyword$, extending it
by one letter in each iteration.
At iteration $i$, a letter $a \in \alphabet$ is picked to extend an already
generated prefix $u$ with probability proportional to the
number of traces that can extend the trace of $u \cdot a$ to a trace 
that intersects $L_n(\aut)$. Let us see this in more detail.

Let $u \in \alphabet^*$ be the prefix (of the desired normal form) constructed so far;
recall that $L_{\nf{}}$ is prefix-closed (\propref{lexnormalform-regular})
so $u \in L_{\nf{}}$.
Let us use $C(u)$ to denote the number of equivalence classes that contain words of length $n$
accepted by $\aut$, and whose normal form contains $u$ as a prefix:
\begin{align*}
C(u) = |\setpred{t \in \quot{\treq}{L_n(\aut)}}{\exists v, \nf{t} = u \cdot v}|
\end{align*}
Let  $\extn(u)$
be the possible choices of letters that $u$ can be extended
with, so that there is a further extension to a normal form word of length $n$
which is equivalent to something in $L_n(\aut)$:
\begin{align*}
\extn(u) = \setpred{a \in \alphabet}{\exists v, \eqcl{u \cdot a \cdot v} \cap L_n(\aut) \neq \emptyset}
\end{align*}
That is, $a \in \extn(u)$ iff $C(u\cdot a) > 0$.
Then, after having sampled prefix $u$, the next letter
 $a \in \extn(u)$ to append to $u$ should be picked with probability
\[
\frac{C(u \cdot a)}{C(u)}
\]
Observe that, for the empty string $\emptyword$, we have $C(\emptyword) = |\quot{\treq}{L_n(\aut)}|$.
Further, for a trace $t \in \quot{\treq}{L_n(\aut)}$ of length
in the language of $\aut$, we have $C({\nf{t}}) = 1$.
Suppose $\nf{t} = a_1a_2\dots a_n$. 
When the letters are sampled independently from each other, 
the probability to sample $\nf{t}$ is thus as desired:
$$
\frac{C(a_1)}{C(\lambda)}\cdot \frac{C(a_1 a_2)}{C(a_1)}  \dots  \frac{C(a_1\cdots a_n)}{C(a_1\cdots a_{n-1})} = \frac{C(a_1\dots a_n)}{C(\lambda)} = \frac{1}{|\quot{\treq}{L_n(\calA)}|}
$$ 

As such, this style of sampling bears resemblance to the standard approach for
sampling that iteratively produces 
a sequence of prefix-extension choices~\cite{JVV1986,ACJR19,ACJR21,MCMNFA2024,MeeldCPODS2025}.
However, this simple approach requires significant
work to adapt to our setting, because the sampled word may not 
be accepted by the automaton: even when a trace $t$ touches
$L_n(\aut)$ (i.e., $t \cap L_n(\aut) \neq \emptyset$), 
its normal form $\nf{t}$ need not be accepted by $\aut$ (i.e., $\nf{t} \not\in L_n(\aut)$).
Further, the quantity $C(u)$ is hard to compute exactly (\thmref{DFA-hard}).
We circumvent these two issues as follows.
We first show that the set of
words whose normal form starts with a fixed string $u$ as prefix, is regular
and accepted by a DFA of size $O(\poly{|u|})$.
We then leverage our FPRAS (\secref{fpras})
to produce approximate estimates for $\set{C(w)}_{w \in L_{\nf{}}}$
so that our sampler can approximately sample normal forms.

In~\secref{prefix-validator}, we show the construction of
our \emph{prefix validator automaton} $\aut_u$, which, for a fixed normal form $u$,
accepts the set of all words that have $u$ as a prefix in their normal form.
In~\secref{fpaus}, we present the final almost uniform sampler
together with analysis of its correctness.


\subsection{Prefix-Validator Automaton} 
\seclabel{prefix-validator}

Here, we describe an essential component of our FPAUS --- a
 succinct computational model that, for a given word $u \in L_{\nf{}}$
(i.e., $u$ is known to be lexicographically minimal from its class),
accepts the set of all words $w$ whose normal form starts with $u$.
We show that in fact, this model is a DFA which we call the \emph{prefix validator
automaton} for $u$ and denote using $\aut_u$, has polynomially many states
and can be constructed in polynomial time (see \lemref{soundness_prefix_validator}).

Given $u \in \Sigma^*$, we view the trace partial order $\trpo{u}$ as a DAG $G(\trpo{u})$ 
where the vertices are the $(a,i) \in {\sf LABIND}_u$ 
and there is an edge from $(a,i)$ to $(b,j)$ iff if $(a,i) \trpo{u} (b,j)$. 
By convention $\trpo{\lambda}$ is seen as the DAG with an empty set of vertices. 
The lemmas of this section are proved in appendix \ucomment{our companion technical report}.

\begin{definition}[DAG-prefix]
Let $u,u' \in \Sigma^*$, $G(\trpo{u'})$ is a \emph{DAG-prefix} of $G(\trpo{u})$ when $G(\trpo{u'})$ is a sub-DAG of $G(\trpo{u'})$ whose sources are also sources of $G(\trpo{u})$ and that is upward closed, that is, if $(a,i)$ is a vertex of $G(\trpo{u'})$ and if $(b,j)$ is an parent of $(a,i)$ in $G(\trpo{u})$ then $(b,j)$ is a parent of $(a,i)$ in $G(\trpo{u'})$.
\end{definition}

Importantly, $G(\trpo{u'})$ can be a DAG-prefix of $G(\trpo{u})$ while $u'$ is not a prefix of $u$. Note that if $G(\trpo{u'})$ is a DAG-prefix of $G(\trpo{u})$ and $u' = u''a$, with $a \in \Sigma$, then $G(\trpo{u''})$ is also a DAG-prefix of $G(\trpo{u})$.

\begin{definition}[Border]
The \emph{border} of a DAG-prefix $G(\trpo{u'})$ of $G(\trpo{u})$ is the set $L \subseteq {\mathsf LABIND}_u$ such that $(b,j) \in L$ if and only if $G(\trpo{u' b})$ is a DAG-prefix of $G(\trpo{u})$. $b$ is called a letter of the border. It is clear that there cannot be $j \neq k$ such that both $(b,j)$ and $(b,k)$ are in $L$. 
\end{definition}

Determining whether $G(\trpo{u'})$ is a DAG-prefix of $G(\trpo{u})$ takes polynomial time. It follows that finding the border of a DAG-prefix of $G(\trpo{u})$ also takes polynomial time.

\begin{definition}[$u$-Prefix and $u$-Residual]
Let $u \in \Sigma^*$ and $w \in \Sigma^*$. The \emph{$u$-prefix} of $w$ is the longest prefix $u'$ of $w$ such that $G(\trpo{u'})$ is a DAG-prefix of $G(\trpo{u})$. This is well-defined since $G(\trpo{\lambda})$ is a DAG-prefix of $G(\trpo{u})$. Writing $w = u'v$, $v$ is the \emph{$u$-residual} of $w$.
\end{definition}


States in the prefix-validator automaton for $u$ are tuples of the form $(u',b,L)$ where $u' \in \Sigma^*$, $b \in \Sigma \cup \{\lambda\}$ and $L \subseteq \Sigma$. A word $w$ reaches $(u',b,L)$ if and only if $u'$ is the $u$-prefix of $\nf{w}$ and $b$ is the first letter the $u$-residual of $\nf{w}$ and $L$ is the set of letters of the $u$-residual of $\nf{w}$. By convention, $b = \lambda$ if the $u$-residual of $\nf{w}$ is $\lambda$. The initial state is $(\lambda,\lambda,\emptyset)$ and the accepting states are all states $(u',b,L)$ where $u' = u$. To understand the transitions it is essential to get a good grasp on how, given $w \in \Sigma^*$ and $a \in \Sigma$, one finds the $u$-prefix of $\nf{wa}$ knowing only the $u$-prefix of $\nf{w}$ plus the first letter and the set of letters of the $u$-residual of $\nf{w}$.

Suppose we have
$$
\nf{w} = \underbrace{\myBox{red!40}{\ell_1\dots\ell_k}}_{u'}\underbrace{\myBox{blue!40}{\ell_{k+1} \dots \ell_l}}_{v}
$$
where the word in red is the $u$-prefix $u'$ of $\nf{w}$ and the word in blue is the $u$-residual $v$ of $\nf{w}$. We claim that $\nf{wa}$ is obtained by inserting $a$ at a certain position in $\nf{w}$ without permuting the remaining letters.

\begin{restatable}{lemma}{insertA}\label{lemma:insert_a}
Let $\nf{w} = \ell_1 \dots \ell_l$ and $a \in \Sigma$. There is an integer $i \in [l + 1]$, such that (with $\ell_{l+1} = \lambda$) $$
\nf{w \cdot a} =  \ell_1\dots \ell_{i-1} \cdot a \cdot \ell_i\dots \ell_l$$
\end{restatable}

We typically assume that $i$ in Lemma~\ref{lemma:insert_a} is as large as possible, so that either $\ell_i = \lambda$ or $a \neq \ell_i$ and $a \lexord \ell_i$. Suppose first that $(a,c) \in \dep$ for some letter $c$ of $v$ (this requires $v \neq \lambda$), then $a$ is ``blocked'' by $v$ from accessing $u'$ and $\nf{wa}$ equals 
$$
\myBox{red!40}{\ell_1\dots\ell_k}\underbrace{\myBox{blue!40}{\ell_{k+1} \dots \ell_i}}_{\neq \lambda}a \, \myBox{blue!40}{\ell_{i+1} \dots \ell_l}
$$
We can then show that the $u$-prefix of $\nf{wa}$ is still the red word, i.e., $u'$. This is because $G(\trpo{u'})$ is a  DAG-prefix of $G(\trpo{u})$ while $G(\trpo{u'\ell_{k+1}})$ is not (otherwise $u'\ell_{k+1}$ would be the $u$-prefix of $\nf{w}$). 

\begin{restatable}{lemma}{uPrefixOne}\label{lemma:uPrefixOne}
Let $a \in \Sigma$ and $w \in \Sigma^*$. Let $u'$ be the $u$-prefix of $\nf{w}$ and $v$ be the $u$-residual of $\nf{w}$. If there is a letter $c$ in $v$ such that $(a,c) \in \dep$, then there are $v_1 \in \Sigma^* \setminus \{\lambda\}$ and $v_2 \in \Sigma^*$ such that $v = v_1v_2$, $\nf{wa} = u'v_1av_2$ and the $u$-prefix of $\nf{w a}$ is $u'$.
\end{restatable}

When $a$ is not ``blocked'' by any letter of $v$ then, by Lemma~\ref{lemma:insert_a}, $\nf{wa}$ is takes of the following three forms. 
$$
\text{(I)} : \myBox{red!40}{\ell_1\dots\ell_i}\,a \underbrace{\myBox{red!40}{\ell_{i+1}\dots\ell_k}}_{\neq \lambda}\myBox{blue!40}{\ell_{k+1} \dots \ell_l}
\qquad 
\text{(II)} : \myBox{red!40}{\ell_1\dots\ell_k}\,a \, \myBox{blue!40}{\ell_{k+1} \dots \ell_l}
\qquad
\text{(III)} : \myBox{red!40}{\ell_1\dots\ell_k}\underbrace{\myBox{blue!40}{\ell_{k+1} \dots \ell_i}}_{\neq \lambda}a \, \myBox{blue!40}{\ell_{i+1} \dots \ell_l}
$$
We can show that case (I) occurs if only if $\nf{u'a} \neq u'a$. In fact, in case (I), $\nf{u'a} = \ell_1\dots\ell_i a \ell_{i+1}\dots \ell_k$ for the $i$ of Lemma~\ref{lemma:insert_a}. Then we can show that the $u$-prefix of $\nf{wa}$ is the $u$-prefix of $\nf{u'a}$ which, importantly, may not be $\nf{u'a}$ in its entirety.

Cases (II) and (III) occurs when $\nf{u'a} = u'a$. Here, perhaps $a$ cannot ``move inside'' $u'$ (because of its dependence with letters of $u'$) or perhaps $a$ can ``move inside'' $u'$ but leaving it outside gives a word that is smaller lexicographically. In any cases, $\nf{wa}$ starts with $u'$ and the question is whether $a$ stays between $u'$ and $v$ (the second case) or whether a lexicographically smaller word is obtained by moving $a$ inside $v$ (the third case). To distinguish between these two cases, it suffices to know whether $\ell_{k+1} \lexord a$ or $a \lexord \ell_{k+1}$. If $\ell_{k+1} \lexord a$ then the smallest lexicographic word is obtained by moving $a$ inside the $u$-residual; this is case (III). If $\ell_{k+1} \lexord a$ then the smallest lexicographic word is obtained by moving $a$ between the $u'$ and $v$; this is case (II). Depending on the case, the $u$-prefix of $\nf{wa}$ can be shown to be $u'$ or $u'a$ because because $G(\trpo{u'})$ is a DAG-prefix of $G(\trpo{u})$ while neither $G(\trpo{u'\ell_{k+1}})$ nor $G(\trpo{u'a\ell_{k+1}})$ is. Note that we still have to check whether $G(\trpo{u'a})$ is a DAG-prefix of $G(\trpo{u})$ in case (II).
  
\begin{restatable}{lemma}{uPrefixTwo}\label{lemma:uPrefixTwo}
Let $a \in \Sigma$ and $w \in \Sigma^*$. Let $u'$ be the $u$-prefix of $\nf{w}$ and $v$ be the $u$-residual of $\nf{w}$. Suppose $(a,c) \in \indep$ for every letter $c$ of $v$. If $\nf{u' a} \neq u'a$ then there is $u'_1,u'_2 \in \Sigma^*$ such that $u' = u'_1u'_2$ and $\nf{u'a} = u'_1au'_2$ and $\nf{wa} = u'_1au'_2v$ and the $u$-prefix of $\nf{w a}$ is the $u$-prefix of $\nf{u' a}$. 
\end{restatable}

\begin{restatable}{lemma}{uPrefixThree}\label{lemma:uPrefixThree}
Let $a \in \Sigma$ and $w \in \Sigma^*$. Let $u'$ be the $u$-prefix of $\nf{w}$ and $v$ be the $u$-residual of $\nf{w}$. Suppose $(a,c) \in \indep$ for every letter $c$ of $v$ and let $b$ be the first letter of $v$. Suppose $\nf{u' a} = u'a$. If $v = \lambda$ or $a \lexord b$ then $\nf{wa} = u'av$ and the $u$-prefix of $\nf{w a}$ is the $u$-prefix of $u'a$. Otherwise, there are $v_1 \in \Sigma^* \setminus \{\lambda\}$ and $v_2 \in \Sigma^*$ such that $v = v_1v_2$, $\nf{wa} = u'v_1av_2$ and the $u$-prefix of $\nf{w a}$ is $u'$.
\end{restatable}

So, to derive the $u$-prefix of $\nf{wa}$, we first want to know the letters of the $u$-residual to determine whether $a$ is ``blocked'' from accessing the $u$-prefix. If $a$ is not ``blocked'' then we want to know if we are in situation (I), (II) or (III) which requires computing $\nf{u'a}$ and its $u$-prefix, with $u'$ the $u$-prefix of $\nf{w}$, and we potentially need to know the first letter of the $u$-residual (to distinguish cases (II) and (III)). We know how to compute $\nf{u'a}$ and its $u$-prefix in polynomial-time. Thus, with only the $u$-prefix of $\nf{w}$ and the set of letters and the first letter of the $u$-residual of $\nf{w}$, we can determine the $u$-prefix of $\nf{wa}$ in polynomial-time.

\begin{definition}\label{def:prefix-aut}
	Let $(\alphabet,\indep)$ be a concurrent alphabet.
	Let $u=u_1\cdots u_k \in \alphabet^k$.
	The \emph{prefix validator automaton} for $u$ is the DFA 
	$\aut_u = (\calQ_u,\alphabet,\calT_u,q_{I, u},F_u)$ where
	\begin{itemize}
		\item \textbf{States}: Each state $p \in \calQ_u$ is of the form $(u',b,L)$ where $u' \in \Sigma^*$ such that $G(\trpo{u'})$ is a DAG-prefix of $G(\trpo{u})$, $b \in \Sigma \cup \{\lambda\}$, and $L$ is a set over $\Sigma$. Informally, a word $w$ reaches $(u',b,L)$ if and only if $u'$ is the $u$-prefix of $\nf{w}$ and $\nf{w} = u'v$ where $v$ starts with $b$ (with $b = \lambda$ if $v = \lambda$) and $L$ is the set of letters of $v$.
		\item \textbf{Transitions}: for every $p_1 = (u'_1,b_1,L_1) \in \calQ_u$ and every $a \in \Sigma$ we have a transition $(p_1,a,(u'_2,b_2,L_2))$. If there is $c \in L_1$ such that $(a,c) \in \dep$ then $a$ is ``blocked'' and $u'_2 = u'_1$, $b_2 = b_1$ and $L_2 = L_1 \cup \{a\}$. Otherwise, we compute $\nf{u'_1a} = u''v'$, with $u''$ its $u$-prefix. 
		\begin{itemize}
		\item[•] If $\nf{u'_1a} \neq u'_1a$, then we are in case (I): $u'_2 = u''$, $b_2 = \mathit{firstLetter}(v'\cdot b_1)$ and $L_2 = L_1 \cup \mathit{letters}(v')$. 
		\item[•] If $\nf{u'_1a} = u'_1a$ and if $b_1 = \lambda$ or $a \lexord b_1$ then we are in case (II): $u'_2 = u''$ and $b_2 = \mathit{firstLetter}(v'\cdot b_1)$ and $L_2 = L_1 \cup \mathit{letters}(v')$
		\item[•] $\nf{u'_1a} = u'_1a$ and if $b_1 \lexord a$ then we are in case (III): $u'_2 = u'_1$ and $b_2 = b_1$ and $L_2 = L_1 \cup \{a\}$.
		\end{itemize}
		\item \textbf{Initial State}: The initial state is $q_{I, u} = (\lambda,\lambda,\emptyset)$.
		\item \textbf{Acceptance}: $(u',b,L) \in F_u$ if and only if $u' = u$

	\end{itemize}
\end{definition}

\if{False}

We now describe the workings of the prefix validator automaton and state its correctness.
In the following, for a word $w \in \alphabet^*$ we use 
$\upsets{w}$ to be the collection of subsets $U \subseteq {\sf LABIND}_w$
that are upward closed with respect to $\trpo{w}$,
i.e., $\big((a, i) \in U \land (a, i) \trpo{w} (b, j) \implies (b, j) \in U\big)$.\todo{change to prefix?}

\begin{definition}\label{def:prefix-aut}
	Let $(\alphabet,\indep)$ be a concurrent alphabet.
	Let $\aut=(\calQ,\alphabet,\calT,q_I,F)$ be an NFA.
	Let $u=u_1\cdots u_k \in \alphabet^k$.
	The \emph{prefix validator automaton} for $u$ is an NFA 
	$\aut_u = (\calQ_u,\alphabet,\calT_u,q_{I, u},F_u)$ where
	\begin{itemize}
		\item \textbf{States}: Each state $p \in \calQ_u$ is either 
		of the form $p = (q, U, \sigma, S, P)$,
		where $q \in \calQ$, $U \in \upsets{u}$, $\sigma \in \alphabet \uplus \set{\bot}$
		and $S, P \subseteq \alphabet$, or is the absorbing reject state $p_\rej$.
		Informally, the component $q$ tracks the state in which the
		original automaton $\aut$ may land after reading
		a string $w$.
		The component $U$ tracks the residual suffix of the partial order $\trpo{u}$
		that is not yet observed in the word $w$ read so far.
		The letter $\sigma$ denotes the last letter read in $w$ that belongs to $u$,
		and, is $\bot$ if none exists.
		The component $S$ denotes the set ${\sf Swap2Right}(U) = \setpred{b\in \alphabet}{\forall (a, i) \in U, \text{we have, } (b,a) \in \indep \land a \lexord b}$. Finally, the set $P$ tracks the set of letters seen in the word $w$ read so far and are not part of the prefix of $u$ read so far.


		\item \textbf{Transitions}: Let $p = (q,U,\sigma,S,P), p' = (q',U',\sigma',S',P') \in \calQ_u$
		and let $a \in \alphabet$.
		Then, $(p, a, p') \in \calT_u$ if and only if $(q,a,q') \in \calT$ and one of the following holds:
		\begin{enumerate}[label=(\alph*)]
			\item there is some $i$ for which $(a, i) \in \min_{\trpo{u}} U$,
			$U' = U \setminus \set{(a, i)}$, $\sigma' = a$, $P' = P$.
			Further, if $U' = \emptyset$, then $S' = {\sf Swap2Right}(U')$, and
			$S' = \Sigma$ otherwise.

			\item there is no $i$ for which $(a, i) \in \min_{\trpo{u}} U$,
			$a \in S$, $U' = U$, $\sigma' = \sigma$, $P' = P \cup \set{a}$,
			$S' = {\sf Swap2Right}(U')$, and either $\textcolor{red}{P \cup \set{\sigma}} = \emptyset$, or,
			$\exists c \in P \cup \set{\sigma},  ((a,c) \in \dep \vee c \lexord a)$;
			here, $P \cup \set{\bot} = P$.
		\end{enumerate}

		 \item \textbf{Initial State}: The initial state of $\aut_u$ is the tuple: 
		 \[q_{I, u} = (q_I, {\sf LABIND}_u, \bot, {\sf Swap2Right}({\sf LABIND}_u), \emptyset)\]

		\item \textbf{Acceptance}: $(q,U,\sigma,S,P) \in F_u$ iff $U = \emptyset$ and $q \in F$.
	
	\end{itemize}
\end{definition}
\fi

\begin{restatable}{lemma}{prefixValidatorSoundness}\lemlabel{soundness_prefix_validator}
	The language of the prefix-validator automaton $\aut_u$ is $\lang{\aut_u} = \setpred{w}{\exists v \in \alphabet^*, \nf{\eqcl{w}}  = u \cdot v}$. The size of state space $\states_u$ is bounded
	by $\cwdth \cdot |u|^{\cwdth} \cdot |\alphabet| \cdot 2^{|\alphabet|}$, where $\cwdth$ is the width of the concurrent alphabet $(\alphabet,\indep)$. The
	automaton construction can be performed in time $O(|\calT_u| \cdot |\calQ_u| \cdot (|\Sigma| + |u|^2))$.
\end{restatable}


\subsection{FPAUS for Trace Sampling}
\seclabel{fpaus}

The procedure for almost-uniform trace sampling is given in Algorithm~\ref{alg:sampler}, \textsc{TraceSample}. It takes in the concurrent alphabet $(\alphabet,\indep)$, an automaton $\aut$ over $\alphabet$, a word-length $n$ and a parameter $\delta \in (0,1)$ and returns either $\bot$ or an element of $\quot{\treq}{L_{n}(\aut)}$. \textsc{TraceSample} does $O(\log(\delta^{-1}))$ calls to the subroutine \textsc{TraceSampleCore}, which performs the prefix-extension-based sampling previously described. However, \textsc{TraceSampleCore} does not always return the constructed sample, instead it may reject the sample and return $\bot$ (Line~\ref{line:rejection}). If all $O(\log(\delta^{-1}))$ calls to \textsc{TraceSampleCore} returns $\bot$ then does so \textsc{TraceSample}, otherwise the first non-$\bot$ output is returned. \textsc{TraceSampleCore} knows the probability $\phi$ of the sample it constructs, which may be not be close enough to $|\quot{\treq}{L_{n}(\aut)}|^{-1}$. By rejecting the sample with probability proportional to $\phi^{-1}$, we ensure that if the output of \textsc{TraceSampleCore} is not $\bot$, then all traces are equally likely to returned. In other words, conditionned on not returning $\bot$, \textsc{TraceSampleCore} is better than an almost uniform sampler, it is an \emph{exact} uniform sampler. Still, the probability of \textsc{TraceSampleCore}  returning $\bot$ may be too large, hence the $O(\log(\delta^{-1}))$ independent calls to boost the probability of returning an actual trace.


\begin{algorithm}[t]
	\caption{\textsc{TraceSample}}
	\algolabel{fpaus}
	\begin{algorithmic}[1]
		\Require Concurrent alphabet $(\Sigma,\indep)$, NFA $\calA$ over $\Sigma$, length $n$, parameter $\delta \in (0,1)$
		\State $m \gets 3 \lceil \log(\delta^{-1}) \rceil$; $out \gets \bot$; $i \gets 1$
		\While{$i < m \text{ and } out = \bot$}
		\State $out \gets \textsc{TraceSampleCore}(\calA,(\Sigma,\indep),n,\delta)$
		\State $i \gets i + 1$
		\EndWhile
		\State \Return $out$
	\end{algorithmic}
	\label{alg:sampler}
\end{algorithm}

\begin{algorithm}[t]
	\caption{\textsc{TraceSampleCore}}
	\algolabel{samplerCore}
	\begin{algorithmic}[1]
		\Require Concurrent alphabet $(\Sigma,\indep)$, NFA $\calA$ over $\Sigma$, length $n$, parameter $\delta \in (0,1)$
		\State $u_0 \gets \lambda$; $i \gets 1$; $\phi \gets 1$; $out' \gets \bot$; $\varepsilon' \gets 1/(16n)$; $\delta' \gets \delta/(2^n\cdot 3 \lceil \log(\delta^{-1}) \rceil)$ 
		\State $\tilde{C} \gets \text{FPRAS}(\lang{\calA}, n, \varepsilon', \delta')$
		\While{$i < n$}
		\State $\extn(u_{i-1}) \gets \{u_{i-1} \cdot a \mid u_{i-1}\cdot a \in L_{\nf{}}\}$ 
		\For{each $w \in \extn(u_{i-1})$}
		\State Construct the prefix-validator $\calA_{w}$ per Definition~\ref{def:prefix-aut}~\label{line:prefix-aut}
		\State Construct $\calA'_w$, the intersection of $\calA_w$ with $\calA$
		\State $\tilde{C}(w) \gets \text{FPRAS}(\calA'_w, 
		n, \varepsilon', \delta')$~\label{line:sample}
		\EndFor
		\State Sample $u_i$ from $\extn(u_{i-1})$ with probability $\tilde{C}(u_i)/\sum_{w \in \extn(u_{i-1})} \tilde{C}(w)$
		\State $\phi \gets \phi \cdot \tilde{C}(u_i)/\sum_{w \in \extn(u_{i-1})} \tilde{C}(w)$
		\State $i \gets i+1$
		\EndWhile
		\If{$\phi \geq 1/(2\tilde{C})$}
		\State $out' \gets u_n$ with probability $1/(2 \phi \tilde{C})$\label{line:rejection}
		\EndIf		
		\State \Return $out'$
	\end{algorithmic}
	\label{alg:sampler_core}
\end{algorithm}

\if{False}
\begin{theorem}[FPRAS to FPAUS Reduction]\label{thm:main-reduction}
	Given a concurrent alphabet $(\Sigma, \indep)$ of width $\cwdth$ 
	and an NFA $\calA$ with $m$ states, if there exists an FPRAS \textcolor{red}{for 
	$|\quot{\treq}{L_n}|$} with time complexity $T(n,\varepsilon,\delta)$, 
	then Algorithm~\ref{alg:sampler} is an FPAUS with:
	
	\begin{itemize}
		\item Time complexity $O(n^{\cwdth+2}m|\Sigma|T(n,\frac{\varepsilon}
		{2n},\frac{\delta}{n}))$
		
		\item Output distribution $\mu$ satisfying for all 
		$\tau \in \quot{\treq}{L_n}$:
		\[
		\frac{1-\varepsilon}{|\quot{\treq}{L_n}|} \leq \mu(\tau) \leq 
		\frac{1+\varepsilon}{|\quot{\treq}{L_n}|}
		\]
		with probability at least $1-\delta$
	\end{itemize}
\end{theorem}
\begin{proof}
The proof relies on analyzing how the multiplicative approximation 
error accumulates across the sampling steps. Let $C_u$ denote the 
number of traces in $\quot{\treq}{L_n(\calA)}$ whose normal form starts 
with $u$. For the base case, when $u_0 = \lambda$, we have $C_\lambda 
= |\quot{\treq}{L_n}|$. By induction on prefix length, we can show 
the error in estimating extension probabilities compounds 
geometrically, yielding total error bounded by $\varepsilon$ after $n$ 
steps.
Let $E$ be the event that, at every step $i$, all $\tilde{C}(w)$ computed in the course of the algorithm are within $(1\pm \varepsilon')C(w)$. When $E$ holds we have 
$$
\frac{1 - \varepsilon'}{1 + \varepsilon'}\cdot \frac{C(u_i)}{C(u_{i-1})} \leq 
\frac{\tilde C(u_i)}{\sum_{w \in ext(u_{i-1})} \tilde C(w)} 
\leq 
\frac{1 + \varepsilon'}{1-\varepsilon'}\cdot \frac{C(u_i)}{C(u_{i-1}}
$$ 
Thus 
$
\Pr[u_n = a_1\dots a_n \mid E] =  \prod_{i \in [n]} \frac{\widetilde{C}_{u_{i-1}\cdot a_i}}{\sum_{b \in A_i} \widetilde{C}_{u_{i-1}\cdot b}} 
$ is bounded by
$$
\frac{1 - \varepsilon}{L}
\leq
\left(\frac{1 - \frac{\varepsilon}{2n}}{1 + \frac{\varepsilon}{2n}}\right)^n \frac{1}{|C_\lambda|}
\leq
\Pr[u_n = a_1\dots a_n \mid E]
\leq  
\left(\frac{1 + \frac{\varepsilon}{2n}}{1-\frac{\varepsilon}{2n}}\right)^n \frac{1}{|C_\lambda|} \leq \frac{1 + \varepsilon}{L}
$$
A union bound on $\Pr[\text{not } E]$ shows that it is at most $(|A_1| +\dots + |A_n|)\cdot \frac{\delta'}{n\cdot |\Sigma|} \leq \delta'$.
\end{proof}
\fi

\begin{restatable}{theorem}{mainReduction}[FPRAS to FPAUS Reduction]\label{thm:main-reduction}
	Suppose there exists an FPRAS with runtime $T(\calA,n,\varepsilon',\delta')$ for counting the traces of an NFA. Then, Algorithm~\ref{alg:sampler}, TraceSample, is an FPAUS for sampling traces from an NFA. Formally, given a concurrent alphabet $(\alphabet,\indep)$, an NFA $\aut$ over $\alphabet$, a positive integer $n$ and a parameter $\delta \in (0,1)$. TraceSample calls the FPRAS $O(n\log(\delta^{-1}))$ times with parameters $\varepsilon' = O(n^{-1})$ and $\delta' = O(\delta 2^{-n})$, runs in time $O(n\log(\delta^{-1}) (n^3\cdot (|\Sigma|+n^2) + T(\calA,n,\varepsilon',\delta')))$ and returns a value $out \in \{\bot\} \cup \quot{\treq}{L_{n}(\aut)}$. Furthermore, the total variation distance between the distribution of $out$ and the uniform distribution over $\quot{\treq}{L_{n}(\aut)}$ is at most $\delta$, i.e., 
	$$
	\frac{1}{2}\left(\Pr[out = \bot] + \sum_{t \in \quot{\treq}{L_{n}(\aut)}} \left|\Pr[out = t] - \frac{1}{|\quot{\treq}{L_{n}(\aut)|}|} \right| \right) \leq \delta
	$$
\end{restatable}


\section{An FPRAS for Trace Counting}
\seclabel{fpras}

We now turn our focus to presenting the FPRAS for trace counting. While the 
high-level structure of our approach follows the sampling-based approach of 
Meel and de Colnet~\cite{MeeldCPODS2025} in the context of counting words of NFA, there are 
significant technical barriers posed by trace counting, which are absent in 
the case of counting words of NFA. In order to discuss the technical barriers, 
we will present the high-level structure of the approach and highlight what 
differentiates trace counting from counting words of NFA. 

At a high level, the FPRAS works by computing an estimate $N(q)$ (of $|\quot{\treq}{L(q) \cap \alphabet^i}|$)
for each state $q$ of the given automaton $\aut$ 
for each length $0 \leq i\leq n$,
and achieves this by visiting the states of the unrolled automaton $\autunroll$
one layer at a time.
For each state $q \in \states^i$ (i.e., states in 
layer $i$ of $\autunroll$), we compute $N(q)$ by using 
the values the estimates $(N(q'))_{q' \in \preds(q)}$,
together with sets 
$(S^r(q'))_{r \in [\nsnt], q' \in \preds(q)}$ which are intuitively 
sets of traces sampled uniformly from $L(q)$. 
After this, the FPRAS also computes $(S^r(q))_{r \in [\nsnt]}$ for the new state $q$;
it does so by looking at sampled sets 
$(S^r(q'))_{r \in [\nsnt], q' \in \preds(q)}$ from the previous layer,
and extending it with a transition symbol. 

The core philosophy behind sampling-based approaches is that if we are trying 
to estimate the cardinality of a set $\mathcal{K}$, then we would like to 
sample every element $k \in \mathcal{K}$ with equal probability, i.e., no 
element of $\mathcal{K}$ should be preferred over another. In the case of 
counting problems over automata, we are working with unrolled automata. 
Accordingly, in the context of \#NFA, this translates to the desideratum of 
sampling every word $w \in L(q)$ for $q \in \calQ^i$ with equal probability. 
In the case of counting traces (i.e., counting \emph{modulo equivalence}), 
we want to ensure that our sample set does not 
consist of more than one word from the same equivalence class. We also want to 
ensure that for every trace $t$, we have a unique representative word 
$w \in t$ such that the set of samples would either contain $w$ or no 
other word from $t$. 
However, we do not a priori know the set of words 
in the equivalence class (or trace) $t$ for which there 
exists a path from the start state to 
$q \in \states^{\ell}$ and therefore, we cannot designate the representative 
word for a given class. The key observation is that we can still define a 
total ordering $\prec$ on all the words in $L(q)$, and designate the 
lexicographically smallest word in a class as the representative word for that 
class. It is worth remarking that the representative word for a class is now 
dependent on the state $q$, but it turns out that from technical analysis, all 
we need is that every class in $L(q)$ has a unique representative word.  

Now moving on to the real technical hurdle: First observe that in the case of 
counting words, for every pair of words $w, w' \in L(q)$, if there exists $v$ such that 
$w \cdot v \in L(q_F)$, then $w' \cdot v$ is also in $L(q_F)$. Equally 
crucially, if $w \neq w'$, then $w\cdot v \neq w' \cdot v$. Therefore, there 
is no reason for us to prefer $w$ over $w'$ (or vice-versa). 

Unfortunately, the aforementioned property breaks down in the case of trace 
counting. Let $q \in \calQ^i$, consider $w \in L(q)$ and $w'$ such that $w$ 
and $w'$ do not belong to the same class, i.e., $\eqcl{w} \neq \eqcl{w'}$ but there exists $v$ and $v'$ such that  $\eqcl{w\cdot v} \neq \eqcl{w' \cdot v'}$. Why does 
this pose a challenge? 
Well, it may be the case that we have 
$L(q) = \set{w_1, w_2, w_3}$ together with $\eqcl{w_1} \neq \eqcl{w_2}$ and 
$\eqcl{w_2} \neq \eqcl{w_3}$ and $\eqcl{w_1} \neq \eqcl{w_3}$, 
i.e., all the three words $w_1, w_2, w_3$ belong to different equivalence classes
but have a run in the unrolled automaton $\autunroll$ that leads to 
the same state $q$. 
Furthermore,  suppose we have the case that for $v_1$ and $v_2$
\begin{itemize}
	\item $\{w_i \cdot v_1, w_i \cdot v_2\} \in L(q_F)$ for all $i = 1,2,3$
	\item $\eqcl{w_1 \cdot v_1} = \eqcl{w_2 \cdot v_2}$ and $\eqcl{w_1 \cdot v_2} = \eqcl{w_2 \cdot v_1}$ 
	\item  
	$\eqcl{w_3 \cdot v_1} \neq \eqcl{w_1 \cdot v_2}$ and $\eqcl{w_3 \cdot v_2} \neq \eqcl{w_1 \cdot v_1}$
\end{itemize}

The aforementioned challenge is rather nontrivial to tackle, as when we are 
sampling \emph{representative} words for $q$, we do want to sample $w_1$, $w_2$ 
and $w_3$ with equal probability since all of them belong to different 
equivalence classes (and in our constructed example, each class is of size $1$). 
However, $\eqcl{w_1 \cdot v_1} = \eqcl{w_2 \cdot v_2}$ and $\eqcl{w_1 \cdot v_2} = \eqcl{w_2 \cdot v_1}$, we are suddenly preferring $\eqcl{w_1 v_1}$ over $\eqcl{w_3 v_1}$. Therefore, 
 if we want to ensure that we sample $\eqcl{w_3 v_1}$ with same probability as $\eqcl{w_1 v_1}$ , we 
would have to keep larger sample sets. How large should these sets be? That 
quantity would depend on how often multiple equivalence classes can shrink to 
one when they are extended by a word -- similar to how $\eqcl{w_1}$ and 
$\eqcl{w_2}$ when extended by $v_1$ (or $v_2$) shrink to the same equivalence class. We 
capture the aforementioned quantity in Lemma~\ref{lemma:main_result_core}, which in turn 
leads to a factor of $\cwdth\cdot n^\cwdth$ in the size of the sets of samples 
$S^r(q)$.

\subsection{Algorithm}

We present an efficient approximation algorithm for counting 
traces in regular languages.  Our algorithm constructs for each state $q \in \calQ^\mathsf{u}$ 
\begin{itemize}
\item[•] an estimate $N(q)$ of $|\quot{\treq}{L(q)}|$ and
\item[•] sets of samples $S^1(q),\dots,S^\nsnt(q) \subseteq L(q)$.
\end{itemize}

The main algorithm {\TraceMC} (Algorithm~\ref{alg:main}) takes as input an NFA $\calA$, a
 length $n$, and approximation parameters $\varepsilon$ and $\delta$. It first computes the unrolled NFA $\calA^{\mathsf{u}}$ of $\calA$ for length $n$ and returns $0$ if the language
 is empty. The algorithm then computes the width $\cwdth$ of the independence relation and
 sets parameters $\ns$, $\nt$, $\nsnt$, $\nv$, and $\threshold$ based on $\varepsilon$,
  $\delta$, $\cwdth$, $n$, and $|\calQ^{\mathsf{u}}|$. It runs {\TraceMCCore} independently $\nv$ times and returns the median of the results. 
  \begin{algorithm}[htb]
  	\caption{$\TraceMC(\calA,n,\varepsilon,\delta)$}
  	\label{alg:main}
  	\begin{algorithmic}[1]
  		\State compute unrolled NFA $\calA^{\mathsf{u}}$ of $\calA$ for $n$
  		\State \textbf{if} $L(\calA^{\mathsf{u}}) = \emptyset$ \textbf{then return} $0$		
  		\State compute $\cwdth$ the width of $(\Sigma,\indep)$
  		\State $\ns \gets \lceil 8\cdot\cwdth\cdot n^{\cwdth+1}\cdot(1+\varepsilon)\cdot\varepsilon^{-2} \rceil$; $\quad\nt \gets \lceil 2\cdot\ln(16\cdot|\calQ^{\mathsf{u}}|) \rceil$; $\quad\nsnt \gets \ns\cdot\nt$; $\quad\nv \gets \lceil 8\cdot\ln(\delta^{-1})\rceil$
  		\State $\threshold \gets 16\cdot\nsnt \cdot(1-\varepsilon)^{-1}\cdot|\calQ^{\mathsf{u}}|$
  		\State \textbf{for} $1 \leq j \leq \nv$ \textbf{do} $\mathsf{est}_j \gets \TraceMCCore(\calA^{\mathsf{u}},n,\nsnt,\ns,\nt,\threshold)$
  		\State \Return $\median(\mathsf{est}_1,\dots,\mathsf{est}_{\nv})$
  	\end{algorithmic}
  \end{algorithm}
  
  To prevent excessive resource consumption, {\TraceMCCore} maintains careful control of the number of samples through 
  the threshold parameter $\threshold$, terminating early with an 
  estimate of zero if this number grows too 
  large. This answer is erroneous but the analysis provides an upper bound on the probability to exit the algorithm this way. The 
  final estimate is computed in $\TraceMC$ as the median of multiple 
  independent runs, providing robust approximation guarantees.

{\TraceMCCore} (Algorithm~\ref{alg:core}) processes states of the unrolled automaton layer by layer. It initializes $N(q_I) \gets 1$ and $S^r(q_I) \gets \{\lambda\}$ for all $r \in [\nsnt]$. For each level $i$ from $1$ to $n$ and each state $q \in \calQ^i$, it calls {\estimateAndSample}$(q)$ to compute $N(q)$ and the sample sets $S^r(q)$. The algorithm tracks the total number of samples and terminates early, returning $0$, if this count exceeds $\threshold$. Otherwise, it returns $N(q_F)$.
\begin{algorithm}[ht]
	\caption{$\TraceMCCore(\calA^{\mathsf{u}},n,\nsnt,\ns,\nt,
		\threshold)$}
	\label{alg:core}
	\begin{algorithmic}[1]
		\State $N(q_I) \gets 1$
		\State \textbf{for} $1 \leq r \leq \nsnt$ \textbf{do} $S^r(q_I) \gets \{\lambda\}$
		\State $\text{numberSamples} \gets \nsnt$
		\For{$1 \leq i \leq n$}
		\For{$q \in \calQ^i$}
		\State $\estimateAndSample(q)$
		\State $\text{numberSamples} \gets \text{numberSamples} + \sum_{1 \leq r \leq \nsnt} |S^r(q)|$
		\State \textbf{if} $\text{numberSamples} \geq \threshold$ \textbf{then} \Return $0$\label{line:interrupt}
		\EndFor
		\EndFor
		\State \Return $N(q_F)$
	\end{algorithmic}
\end{algorithm}

{\estimateAndSample} (Algorithm~\ref{alg:estimate}) computes $N(q)$ and the sample sets $S^r(q)$ for a given state $q$. Let $\preds(q) = (q_1,\dots,q_k)$. The procedure first computes $N_{max}(q) = \max(N(q_1),\dots,N(q_k))$. For each trial $r \in [\nsnt]$ and each predecessor $q_i$, it constructs a normalized sample set $\bar S^r(q_i,q) \gets \reduce(S^r(q_i),N(q_i)/N_{max}(q))$. These normalized sets are then aggregated using the {\union} procedure to form $\hat{S}^r(q)$ for each $r$. The procedure computes $\hat N(q) \gets N_{max}(q)\cdot\mom(\ns,$ $|\hat S^1(q)|,\dots,|\hat S^\nsnt(q)|)$ and sets $N(q) \gets \round(q,\min(N_{max}(q),\hat{N}(q)))$ where $\round$ takes in $q$ and a value $v$ and returns the smallest value that is greater than or equal to $v$ and that is either an integer or of the form $(1+\varepsilon)\ell$ or of the form $(1-\varepsilon)\ell$ for some integer $\ell$. The rounding is here for technical reasons. Finally, for each $r \in [\nsnt]$, {\estimateAndSample} constructs $S^r(q) \gets \reduce(\hat{S}^r(q),N_{max}(q)/N(q))$.

\begin{algorithm}[ht]
	\caption{$\estimateAndSample(q)$ with $\preds(q) = (q_1,\dots,q_k)$}
	\label{alg:estimate}
	\begin{algorithmic}[1]
		\State $N_{max}(q)(q) \gets \max(N(q_1),\dots,N(q_k))$
		\State \textbf{for} $1 \leq r \leq \nsnt$ \textbf{and} $1 \leq i \leq k$ \textbf{do} $\bar S^r(q_i,q) \gets \reduce(S^r(q_i),\frac{N(q_i)}{N_{max}(q)})$\label{line:sbar}
		\State \textbf{for} $1 \leq r \leq \nsnt$ \textbf{do} $\hat{S}^r(q) \gets \union(q,\bar S^r(q_1,q),\dots,\bar S^r(q_k,q))$\label{line:shat}
		\State $\hat N(q) \gets N_{max}(q)\cdot\mom(\ns,|\hat S^1(q)|,\dots,|\hat S^\nsnt(q)|)$
		\State $N(q) \gets \round(q,\min(N_{max}(q),\hat{N}(q)))$
		\State \textbf{for}  $1 \leq r \leq \nsnt$ \textbf{do} $S^r(q) \gets \reduce(\hat{S}^r(q),\frac{N_{max}(q)}{N(q)})$\label{line:s}
	\end{algorithmic}
\end{algorithm}

The {\reduce} procedure (Algorithm~\ref{alg:reduce}) adjusts sampling probabilities. Given a random set $S \subseteq W$ and probability $p \in [0,1]$, it generates $S' \subseteq S$ such that $\Pr[w \in S'] = p\cdot \Pr[w \in S]$ for all $w \in W$. The procedure creates an empty set $S'$ and adds each element $w \in S$ to $S'$ independently with probability $p$.

\begin{algorithm}[htb]
	\caption{$\reduce(S,p)$ with $p \in [0,1]$}
	\label{alg:reduce}
	\begin{algorithmic}[1]
		\State $S' \gets \emptyset$
		\State \textbf{for} each $w \in S$\textbf{,} add $w$ to $S'$ with probability $p$
		\State \Return $S'$
	\end{algorithmic}
\end{algorithm}

We assume access to a membership procedure $\memp{\treq}$ that answers $true$ on input $(A,q,w)$ if a word in $\eqcl{w}$ is accepted by $q$ in automaton $A$. Now, in our case we have a word $w\cdot s$ that reaches $q$ through the transition $(q_i,s,q)$ and we have to determine whether there is an equivalent word that reaches $q$ through an earlier transition. We do this by removing $(q_i,s,q)$ and all transitions following it from $\calA^{\mathsf{u}}$ and calling $\memp{\treq}$ on the new automaton, $q$ and $w\cdot s$. 

The {\union} procedure (Algorithm~\ref{alg:union}) aggregates normalized sample sets from the predecessors of $q$. Let $\preds(q) = (q_1,\dots,q_k)$ and $S_1,\dots,S_k$ be sets with $S_i \subseteq L(q_i)$. The procedure initializes an empty set $S$ and computes $T$, the set of all transitions entering $q$. It processes transitions in $T$ according to a fixed total order $\prec$. For each transition $\tau = (q_i,s,q)$ and each word $w \in S_i$, the procedure constructs a modified automaton $\calA'$ by removing from $\calT^\mathsf{u}$ the transition $\tau$ and all transitions that follow $\tau$ in the ordering $\prec$. It then calls the membership oracle $\memp{\treq}(\calA',q,w\cdot s)$ to check whether a word equivalent to $w \cdot s$ reaches $q$ in $\calA'$. If the oracle returns false, the word $w \cdot s$ is added to $S$. The procedure returns $S$.
\begin{algorithm}[ht]
	\caption{$\union(q,S_1,\dots,S_k)$ with $\preds(q) = (q_1,\dots,q_k)$  and $S_i \subseteq L(q_i)$}
	\label{alg:union}
	\begin{algorithmic}[1]
		\State $S \gets \emptyset$
		\State $T \gets \calT^\mathsf{u} \cap (\calQ^\mathsf{u} \times \Sigma \times \{q\})$ 
		\For{$\tau = (q_i,s,q) \in T$}
		\For{$w \in S_i$}
		\State $\calA' \gets (\calQ^\mathsf{u},\Sigma, \calT^\mathsf{u} \setminus \{\tau' \in T \mid  \tau' \succcurlyeq \tau\}, q_I, q_F)$\label{line:new_automaton}
		\State \textbf{if} $\memp{\treq}(\calA',q,w\cdot s)$ answers \emph{false} \textbf{then}   add $w \cdot s$ to $S$ \label{line:membership_check}
		\EndFor
		\EndFor
		\State \Return $S$
	\end{algorithmic}
\end{algorithm}


\subsection{Canonical Runs and Canonical Words}\label{sec:canrun}

We seek to count the traces modulo $\indep$ that intersect $L(q)$, yet our algorithm samples words of $L(q)$. As explained before, words acts as representants of their trace. Thus, if $S(q)$ is a sample set for $q$ and that a word $w \in L(q)$ is put in $S(q)$ by the algorithm, then we have sampled $\eqcl{w}$. To avoid that the size of the intersection $t \cap L(q)$ of a trace $t$ with $L(q)$ does not influence the probability that $t$ is sampled for $q$, we make sure that every trace $t$ intersecting $L(q)$ is equally likely to be sampled through a canonical representant, denoted by $\can(t,q)$.

Words sampled for $q$ are built over words sampled for $q$'s predecessors. That is, the only way $w$ can be sampled for $q$ is if there is a transition $(q',a,q) \in \calT^\mathsf{u}$ such that $w = \omega\cdot a$ and that $\omega$ is sampled for $q'$. This means that sample words are constructed through runs in the automaton. Because the automaton is non-deterministic, two traces $t$ and $t'$ that both intersect $L(q)$ can have different number of runs reaching $q$ for $\can(t,q)$ and $\can(t',q)$, respectively. So to make sure that $t$ and $t'$ are sampled with equal probability, it is important that a canonical word is placed in a sample set only when it is constructed through a specific run. We call that run the canonical run and note it $\run(t,q)$.

Runs in the unrolled automaton are represented as sequences of transitions. For convenience we also use an arrow notation, for instance the run $((q_I,a,q_1),(q_1,b,q_2),(q_2,c,q_3))$ can be written
$$
q_I \xrightarrow{a} q_1 \xrightarrow{b} q_2 \xrightarrow{c} q_3
$$
For $R$ a run and $0 \leq k \leq |R|$, we denote by $\word(R)$ the word constructed by $R$. Note that for the empty run $\lambda$ we have $\word(\lambda) = \lambda$ ($\lambda$ being the symbol used for any empty sequence, be it a word or a run). For instance the word of the run $((q_I,a,q_1),(q_1,b,q_2),(q_2,c,q_3))$ is $abc$.

\begin{definition}[Canonical runs]\label{definition:canonical_run}
Let $q \in \calQ^\mathsf{u}$ and let $t \in \quot{\treq}{L(q)}$. The \emph{canonical run} $\run(t,q)$ of $t$ for $q$ is defined as follows
\begin{itemize}
\item[•] if $q$ is the initial state of $\calA^\mathsf{u}$, then $t = \eqcl{\lambda}$ and $\run(t,q) = \lambda$,
\item[•] otherwise let $(q',s,q) \in \calT^\mathsf{u}$ be the first transition with respect to $\prec$ such that $W = \{w \in L(q') \mid w \cdot s \in t\}$ is not empty. There is a unique $t' \in \quot{\treq}{L(q')}$ such that $W  \subseteq t'$. We set $\run(t,q) = \run(t',q') \cdot (q',s,q)$.
\end{itemize} 
\end{definition}

\begin{definition}[Canonical words]\label{definition:canonical_run}
Let $q \in \calQ^\mathsf{u}$ and let $t \in \quot{\treq}{L(q)}$. The \emph{canonical word} of $t$ for $q$ is $\can(t,q) = \word(\run(t,q))$. It is readily verified that $t = \eqcl{\can(t,q)}$.
\end{definition}

\begin{example}
In~\figref{example}, suppose the independence relation is $I = \{(a,b),(b,a),(b,c),(c,b)\}$. We have that $\eqcl{abcc} = \{abcc,bacc\}$. Both $abcc$ and $bacc$ are accepted by the automaton through the runs 
\begin{align*}
q_I \xrightarrow{b} q_1 \xrightarrow{a} q_6 \xrightarrow{c} q_9 \xrightarrow{c} q_F
\qquad
q_I \xrightarrow{a} q_2  \xrightarrow{b} q_5 \xrightarrow{c} q_9 \xrightarrow{c} q_F
\qquad
q_I  \xrightarrow{b} q_3 \xrightarrow{a} q_6 \xrightarrow{c} q_9 \xrightarrow{c} q_F
\end{align*}
Let the alphabet symbols be ordered as $a \prec b \prec c$, let the states be ordered as $q_F \prec q_1 \prec q_2 \prec q_3 \prec q_4 \prec q_5 \prec q_6 \prec q_7 \prec q_8 \prec q_9 \prec q_{10} \prec q_F$ and define $(q_i,s,q_j) \prec (q_k,s',q_l)$ when $q_i \prec q_k$, or when $q_i = q_k$ and $q_j \prec q_l$, or when $q_i = q_k$ and $q_j = q_l$ and $s \prec s'$. Since the three runs above go through $(q_9,c,q_F)$ we have that 
$$
\run(\eqcl{abcc},q_F) = \run(\eqcl{abc},q_9) \xrightarrow{c} q_F
$$
Now, $\eqcl{abc} \cap L(q_9) = \{abc,bac\}$. $abc$ is constructed by extending $ab \in L(q_5)$ with $c$ while $ba$ is constructed by extending $ba \in L(q_6)$ with $c$. Since $(q_5,c,q_9) \prec (q_6,c,q_9)$ we have that 
$$
\run(\eqcl{abcc},q_F) = \run(\eqcl{abc},q_9) \xrightarrow{c} q_F = \run(\eqcl{ab},q_5) \xrightarrow{c} q_9 \xrightarrow{c} q_F
$$
Following the construction, we find 
$$
\run(\eqcl{abcc},q_F) = \run(\eqcl{a},q_2) \xrightarrow{b} q_5 \xrightarrow{c} q_9 \xrightarrow{c} q_F 
= q_I \xrightarrow{a} q_2  \xrightarrow{b} q_5 \xrightarrow{c} q_9 \xrightarrow{c} q_F
$$
Thus, the canonical word of $\eqcl{abcc}$ for $q_F$ is $abcc$.
\end{example}

\begin{figure*}
\centering
\scalebox{0.8}{
\begin{tikzpicture}[xscale=1.95,yscale=1.5, every node/.style={scale=0.8
5}]
\def\f{\small}
\def\s{0.8cm}

\node[draw=white,,fill=white,circle,minimum size =\s] (qi) at (0,0) { };

\node[draw=white,,fill=white,circle,minimum size =\s] (q1) at (1,+1) { };
\node[draw=white,,fill=white,circle,minimum size =\s] (q2) at (1,+0.33) { };
\node[draw=white,,fill=white,circle,minimum size =\s] (q3) at (1,-0.33) { };
\node[draw=white,,fill=white,circle,minimum size =\s] (q4) at (1,-1) { };

\node[draw=white,,fill=white,circle,minimum size =\s] (q5) at (2,+1) { };
\node[draw=white,,fill=white,circle,minimum size =\s] (q6) at (2,0) { };
\node[draw=white,,fill=white,circle,minimum size =\s] (q7) at (2,-1) { };

\node[draw=white,,fill=white,circle,minimum size =\s] (q8) at (3,+1) { };
\node[draw=white,,fill=white,circle,minimum size =\s] (q9) at (3,0) { };
\node[draw=white,,fill=white,circle,minimum size =\s] (q10) at (3,-1) { };

\node[draw=white,fill=white,circle,minimum size =\s] (qf) at (4,0) { };

\draw[-latex] (qi) to  node[midway,fill=white,font=\f] {$b$} (q1);
\draw[-latex] (qi) to  node[midway,fill=white,font=\f] {$a$}  (q2);
\draw[-latex] (qi) to  node[midway,fill=white,font=\f] {$b$}  (q3);
\draw[-latex] (qi) to  node[midway,fill=white,font=\f] {$c$}  (q4);

\draw[-latex] (q1) to  node[midway,fill=white,font=\f]  {$b$}  (q5);
\draw[-latex] (q1) to  node[near end,fill=white,font=\f]  {$a$}  (q6);

\draw[-latex] (q2) to  node[near end,fill=white,font=\f]  {$b$}  (q5);
\draw[-latex] (q2) to  node[midway,fill=white,font=\f]  {$a$}  (q6);

\draw[-latex] (q3) to  node[midway,fill=white,font=\f]  {$a$}  (q6);
\draw[-latex] (q3) to  node[near end,fill=white,font=\f]  {$b$}  (q7);

\draw[-latex] (q4) to  node[near end,fill=white,font=\f]  {$c$}  (q6);
\draw[-latex] (q4) to  node[midway,fill=white,font=\f]  {$c$}  (q7);

\draw[-latex] (q5) to  node[midway,fill=white,font=\f]  {$b$} (q8);
\draw[-latex] (q5) to  node[near start,fill=white,font=\f]  {$c$} (q9);

\draw[-latex] (q6) to  node[near start,fill=white,font=\f]  {$a$} (q8);
\draw[-latex] (q6) to  node[midway,fill=white,font=\f]  {$c$} (q9);
\draw[-latex] (q6) to  node[midway,fill=white,font=\f]  {$c$} (q10);

\draw[-latex] (q7) to node[midway,fill=white,font=\f]  {$a$} (q10);

\draw[-latex] (q8) to  node[midway,fill=white,font=\f]  {$a$} (qf);
\draw[-latex] (q9) to  node[midway,fill=white,font=\f]  {$c$} (qf);
\draw[-latex] (q10) to [in=-120,out=20] node[midway,fill=white,font=\f]  {$b$} (qf);
\draw[-latex] (q10) to [in=-150,out=50] node[midway,fill=white,font=\f]  {$a$} (qf);


\node[draw,fill=white,circle,minimum size =\s] (qi) at (0,0) {$q_I$};

\node[draw,fill=white,circle,minimum size =\s] (q1) at (1,+1) {$q_1$};
\node[draw,fill=white,circle,minimum size =\s] (q2) at (1,+0.33) {$q_2$};
\node[draw,fill=white,circle,minimum size =\s] (q3) at (1,-0.33) {$q_3$};
\node[draw,fill=white,circle,minimum size =\s] (q4) at (1,-1) {$q_4$};

\node[draw,fill=white,circle,minimum size =\s] (q5) at (2,+1) {$q_5$};
\node[draw,fill=white,circle,minimum size =\s] (q6) at (2,0) {$q_6$};
\node[draw,fill=white,circle,minimum size =\s] (q7) at (2,-1) {$q_7$};

\node[draw,fill=white,circle,minimum size =\s] (q8) at (3,+1) {$q_8$};
\node[draw,fill=white,circle,minimum size =\s] (q9) at (3,0) {$q_9$};
\node[draw,fill=white,circle,minimum size =\s] (q10) at (3,-1) {$q_{10}$};

\node[draw,fill=white,circle,minimum size =\s] (qf) at (4,0) {$q_F$};
\end{tikzpicture}
}
\caption{Automaton for Example 2}\label{fig:example}
\end{figure*}
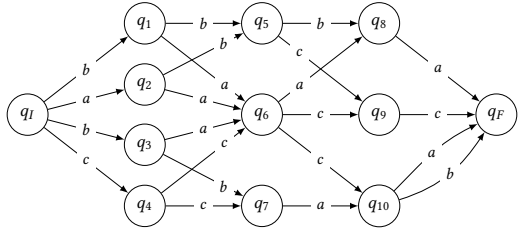

In Section~\ref{sec:fpras}, we claimed that $\TraceMCCore$ creates only sample that are canonical words. We prove this formally with the next lemmas.

\begin{lemma}\label{lemma:samples_are_canonical_words}
Let $q \in \calQ^\mathsf{u}$ and $t \in \quot{\treq}{L(q)}$. Let $r \in [\nsnt]$. If in an run of $\TraceMCCore(\calA^{\mathsf{u}},n,\nsnt,\ns,\nt,\theta)$ we have that $S^r(q) \cap t \neq \emptyset$ then $S^r(q) \cap t = \{\can(t,q)\}$.
\end{lemma}
\begin{proof}
Let $k \in [0,n]$ such that $q \in \calQ^k$. We proceed by induction on $k$. For $k = 0$ the lemma is immediate since $q = q_I$, $L(q_I) = \quot{\treq}{\lambda} = \{\lambda\}$, $\can(\{\lambda\},q_I) = \lambda$ and the sets $S^r(q_I) = \{\lambda\}$. Now, suppose that $k > 0$ and that the lemma holds for all states in $\calQ^{\leq k-1}$. Suppose $w \in S^r(q) \cap t$. Since $S^r(q) \subseteq \hat S^r(q)$ we have that $w \in \hat S^r(q) \cap t$. Let $q_1,\dots,q_l$ be the predecessors of $q$ and recall that $\hat S^r(q) = \union(q,\bar S^r(q_1,q),\dots,\bar S^r(q_l,q))$. Write $w = w' \cdot s$, with $s$ the last symbol of $w$. The $\union$ ensures that $w \in \hat S^r (q)$ if only if there is $i \in [\ell]$ such that the following holds
\begin{itemize}
\item[1.] $(q_i,s,q) \in \calT^\mathsf{u}$  
\item[2.] $w' \in S^r(q_i,q)$  
\item[3.] there is no transition $(q_j,s',q) \prec (q_i,s,q)$ such that $w'' \cdot s' \treq w' \cdot s$ for some $w'' \in L(q_j)$.
\end{itemize} 
Let $t' = \quot{\treq}{w'}$. Note that $S^r(q_i,q) \subseteq S^r(q_i)$ so, by induction, $w' = \can(t',q_i)$. By items 1. and 3., we have that $\run(t,q) = \run(t',q_i) \cdot (q_i,s,q)$, so $\can(t,q) = \can(t',q_i)\cdot s =  w' \cdot s = w$.
\end{proof}

\begin{lemma}\label{lemma:samples_are_canonical_runs}
Let $q \in \calQ^\mathsf{k}$, $t \in \quot{\treq}{L(q)}$ and $\run(t,q) = (q_I = q_0 \xrightarrow{a_1} q_1 \xrightarrow{a_2} q_2 \xrightarrow{a_3} \dots \xrightarrow{a_k} q_k = q)$. If in an run of $\TraceMCCore(\calA^{\mathsf{u}},n,\nsnt,\ns,\nt,\theta)$ we have that $a_1a_2\dots a_k \in S^r(q)$ then for all $0 \leq i \leq k$ we have that $a_1\dots a_i \in S^r(q_i)$.
\end{lemma}

It is now verified that the algorithm can sample only one word per trace and per state and that, assuming $\prec$ is fixed, this word is unique. In addition this word is sampled for $q$ only if it is constructed through its canonical run to $q$. With this we can now say that the algorithm essentially samples traces and we may sometimes replace an event ``$w \in S^r(q)$'' by the event ``$\eqcl{w} \in S^r(q)$''. 

\subsection{Divergence and Dependence} 

When two words $w$ and $w'$ can be sampled for $q$, we will be interested in the probability that both are sample at the same time, that is,  
$$
\Pr[w \in S^r(q) \text{ and } w' \in S^r(q)]
$$
In the next section we will see that the above probability is connected to the \emph{divergence state} of $\run(\eqcl{w},q)$ and $\run(\eqcl{w'},q)$. For $R$ a run in the automaton and $k \geq 0$ an integer at most the number of states in $R$, we denote by $R[k]$ the run $R$ restricted to its first $k$ transitions. For instance if $R = ((q,a,q'),(q',b,q''),(q'',c,q'''))$, then $R[0] = \lambda$, $R[1] = ((q,a,q'))$, $R[2] = ((q,a,q'),(q',b,q''))$ and $R[3] = R$. The \emph{end state} of $R$ is the end state of its last transition.

\begin{definition}
Let $R$ and $R'$ be two distinct runs in $\calA^\mathsf{u}$. The \emph{divergence state} of $R$ and $R'$ is the end state of $R[k]$ for the largest $k \leq \min(|R|,|R'|)$ such that $R[k] = R'[k]$. We also say that $R$ and $R'$ \emph{diverge at their $(k+1)^\text{th}$ state}.
\end{definition}

\begin{example}
In Figure~\ref{fig:divergence} (left), the two runs 
$$
q_I \xrightarrow{a} q_2 \xrightarrow{b} q_5 \xrightarrow{c} q_9 \xrightarrow{c} q_F \quad \text{ and } \quad q_I \xrightarrow{a} q_2 \xrightarrow{b} q_5 \xrightarrow{b} q_8 \xrightarrow{a} q_F
$$
diverge at their $3^\text{rd}$ node $q_5$. These two runs diverge because they read a next symbol after $q_5$ but in other cases runs can diverge at a state even though reading the same symbol, for instances the two runs 
$$
q_I \xrightarrow{a} q_2 \xrightarrow{a} q_6 \xrightarrow{c} q_{10} \xrightarrow{b} q_F \quad \text{ and }
\quad q_I \xrightarrow{a} q_2 \xrightarrow{a} q_6 \xrightarrow{c} q_9 \xrightarrow{c} q_F
$$ 
shown in ~\ref{fig:divergence} (right), diverge at $q_6$ despite both reading $c$ next.

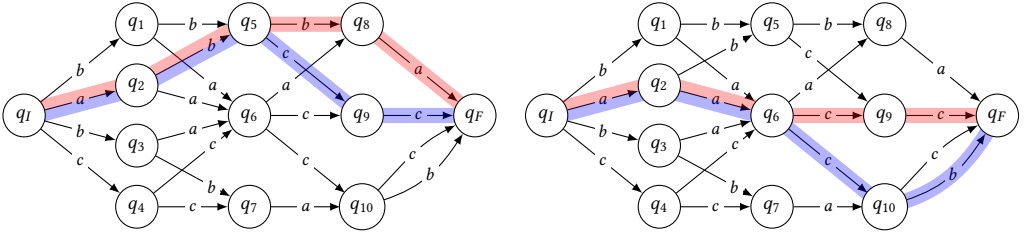
\begin{figure*}[h!]
\centering
\begin{subfigure}{0.5\textwidth}
\begin{tikzpicture}[xscale=1.5,yscale=1.2, every node/.style={scale=0.7
5}]
\def\f{\small}
\def\s{0.75cm}
\def\i{2}

\node[draw=white,,fill=white,circle,minimum size =\s] (qi) at (0,0) { };

\node[draw=white,,fill=white,circle,minimum size =\s] (q1) at (1,+1) { };
\node[draw=white,,fill=white,circle,minimum size =\s] (q2) at (1,+0.33) { };
\node[draw=white,,fill=white,circle,minimum size =\s] (q3) at (1,-0.33) { };
\node[draw=white,,fill=white,circle,minimum size =\s] (q4) at (1,-1) { };

\node[draw=white,,fill=white,circle,minimum size =\s] (q5) at (2,+1) { };
\node[draw=white,,fill=white,circle,minimum size =\s] (q6) at (2,0) { };
\node[draw=white,,fill=white,circle,minimum size =\s] (q7) at (2,-1) { };

\node[draw=white,,fill=white,circle,minimum size =\s] (q8) at (3,+1) { };
\node[draw=white,,fill=white,circle,minimum size =\s] (q9) at (3,0) { };
\node[draw=white,,fill=white,circle,minimum size =\s] (q10) at (3,-1) { };

\node[draw=white,fill=white,circle,minimum size =\s] (qf) at (4,0) { };

\draw[-latex] (qi) to  node[midway,fill=white,font=\f,inner sep=\i] {$b$} (q1);
\draw[-latex] (qi) to  node[midway,fill=white,font=\f,inner sep=\i] {$a$}  (q2);
\draw[-latex] (qi) to  node[midway,fill=white,font=\f,inner sep=\i] {$b$}  (q3);
\draw[-latex] (qi) to  node[midway,fill=white,font=\f,inner sep=\i] {$c$}  (q4);

\draw[-latex] (q1) to  node[midway,fill=white,font=\f,inner sep=\i]  {$b$}  (q5);
\draw[-latex] (q1) to  node[near end,fill=white,font=\f,inner sep=\i]  {$a$}  (q6);

\draw[-latex] (q2) to  node[near end,fill=white,font=\f,inner sep=\i]  {$b$}  (q5);
\draw[-latex] (q2) to  node[midway,fill=white,font=\f,inner sep=\i]  {$a$}  (q6);

\draw[-latex] (q3) to  node[midway,fill=white,font=\f,inner sep=\i]  {$a$}  (q6);
\draw[-latex] (q3) to  node[near end,fill=white,font=\f,inner sep=\i]  {$b$}  (q7);

\draw[-latex] (q4) to  node[near end,fill=white,font=\f,inner sep=\i]  {$c$}  (q6);
\draw[-latex] (q4) to  node[midway,fill=white,font=\f,inner sep=\i]  {$c$}  (q7);

\draw[-latex] (q5) to  node[midway,fill=white,font=\f,inner sep=\i]  {$b$} (q8);
\draw[-latex] (q5) to  node[near start,fill=white,font=\f,inner sep=\i]  {$c$} (q9);

\draw[-latex] (q6) to  node[near start,fill=white,font=\f,inner sep=\i]  {$a$} (q8);
\draw[-latex] (q6) to  node[midway,fill=white,font=\f,inner sep=\i]  {$c$} (q9);
\draw[-latex] (q6) to  node[midway,fill=white,font=\f,inner sep=\i]  {$c$} (q10);

\draw[-latex] (q7) to node[midway,fill=white,font=\f,inner sep=\i]  {$a$} (q10);

\draw[-latex] (q8) to  node[midway,fill=white,font=\f,inner sep=\i]  {$a$} (qf);
\draw[-latex] (q9) to  node[midway,fill=white,font=\f,inner sep=\i]  {$c$} (qf);
\draw[-latex] (q10) to [in=-120,out=20] node[midway,fill=white,font=\f,inner sep=\i]  {$b$} (qf);
\draw[-latex] (q10) to [in=-150,out=50] node[midway,fill=white,font=\f,inner sep=\i]  {$c$} (qf);

\draw[color=red,opacity=0.3,line width=1.45mm] (0,0.062) -- (1,+0.33+0.062) -- (2,1+0.07);
\draw[color=blue,opacity=0.3,line width=1.45mm] (0,-0.062) -- (1,+0.33-0.062) -- (2,1-0.07);
\draw[color=blue,opacity=0.3,line width=2mm] (2,1) -- (3,0) -- (4,0);
\draw[color=red,opacity=0.3,line width=2mm] (2,1) -- (3,1) -- (4,0);

\node[draw,fill=white,circle,minimum size =\s] (qi) at (0,0) {$q_I$};

\node[draw,fill=white,circle,minimum size =\s] (q1) at (1,+1) {$q_1$};
\node[draw,fill=white,circle,minimum size =\s] (q2) at (1,+0.33) {$q_2$};
\node[draw,fill=white,circle,minimum size =\s] (q3) at (1,-0.33) {$q_3$};
\node[draw,fill=white,circle,minimum size =\s] (q4) at (1,-1) {$q_4$};

\node[draw,fill=white,circle,minimum size =\s] (q5) at (2,+1) {$q_5$};
\node[draw,fill=white,circle,minimum size =\s] (q6) at (2,0) {$q_6$};
\node[draw,fill=white,circle,minimum size =\s] (q7) at (2,-1) {$q_7$};

\node[draw,fill=white,circle,minimum size =\s] (q8) at (3,+1) {$q_8$};
\node[draw,fill=white,circle,minimum size =\s] (q9) at (3,0) {$q_9$};
\node[draw,fill=white,circle,minimum size =\s] (q10) at (3,-1) {$q_{10}$};

\node[draw,fill=white,circle,minimum size =\s] (qf) at (4,0) {$q_F$};
\end{tikzpicture}
\end{subfigure}\begin{subfigure}{0.5\textwidth}
\begin{tikzpicture}[xscale=1.5,yscale=1.2, every node/.style={scale=0.7
5}]
\def\f{\small}
\def\s{0.75cm}
\def\i{2}

\node[draw=white,,fill=white,circle,minimum size =\s] (qi) at (0,0) { };

\node[draw=white,,fill=white,circle,minimum size =\s] (q1) at (1,+1) { };
\node[draw=white,,fill=white,circle,minimum size =\s] (q2) at (1,+0.33) { };
\node[draw=white,,fill=white,circle,minimum size =\s] (q3) at (1,-0.33) { };
\node[draw=white,,fill=white,circle,minimum size =\s] (q4) at (1,-1) { };

\node[draw=white,,fill=white,circle,minimum size =\s] (q5) at (2,+1) { };
\node[draw=white,,fill=white,circle,minimum size =\s] (q6) at (2,0) { };
\node[draw=white,,fill=white,circle,minimum size =\s] (q7) at (2,-1) { };

\node[draw=white,,fill=white,circle,minimum size =\s] (q8) at (3,+1) { };
\node[draw=white,,fill=white,circle,minimum size =\s] (q9) at (3,0) { };
\node[draw=white,,fill=white,circle,minimum size =\s] (q10) at (3,-1) { };

\node[draw=white,fill=white,circle,minimum size =\s] (qf) at (4,0) { };

\draw[-latex] (qi) to  node[midway,fill=white,font=\f,inner sep=\i] {$b$} (q1);
\draw[-latex] (qi) to  node[midway,fill=white,font=\f,inner sep=\i] {$a$}  (q2);
\draw[-latex] (qi) to  node[midway,fill=white,font=\f,inner sep=\i] {$b$}  (q3);
\draw[-latex] (qi) to  node[midway,fill=white,font=\f,inner sep=\i] {$c$}  (q4);

\draw[-latex] (q1) to  node[midway,fill=white,font=\f,inner sep=\i]  {$b$}  (q5);
\draw[-latex] (q1) to  node[near end,fill=white,font=\f,inner sep=\i]  {$a$}  (q6);

\draw[-latex] (q2) to  node[near end,fill=white,font=\f,inner sep=\i]  {$b$}  (q5);
\draw[-latex] (q2) to  node[midway,fill=white,font=\f,inner sep=\i]  {$a$}  (q6);

\draw[-latex] (q3) to  node[midway,fill=white,font=\f,inner sep=\i]  {$a$}  (q6);
\draw[-latex] (q3) to  node[near end,fill=white,font=\f,inner sep=\i]  {$b$}  (q7);

\draw[-latex] (q4) to  node[near end,fill=white,font=\f,inner sep=\i]  {$c$}  (q6);
\draw[-latex] (q4) to  node[midway,fill=white,font=\f,inner sep=\i]  {$c$}  (q7);

\draw[-latex] (q5) to  node[midway,fill=white,font=\f,inner sep=\i]  {$b$} (q8);
\draw[-latex] (q5) to  node[near start,fill=white,font=\f,inner sep=\i]  {$c$} (q9);

\draw[-latex] (q6) to  node[near start,fill=white,font=\f,inner sep=\i]  {$a$} (q8);
\draw[-latex] (q6) to  node[midway,fill=white,font=\f,inner sep=\i]  {$c$} (q9);
\draw[-latex] (q6) to  node[midway,fill=white,font=\f,inner sep=\i]  {$c$} (q10);

\draw[-latex] (q7) to node[midway,fill=white,font=\f,inner sep=\i]  {$a$} (q10);

\draw[-latex] (q8) to  node[midway,fill=white,font=\f,inner sep=\i]  {$a$} (qf);
\draw[-latex] (q9) to  node[midway,fill=white,font=\f,inner sep=\i]  {$c$} (qf);
\draw[-latex] (q10) to [in=-120,out=20] node[midway,fill=white,font=\f,inner sep=\i]  {$b$} (qf);
\draw[-latex] (q10) to [in=-150,out=50] node[midway,fill=white,font=\f,inner sep=\i]  {$c$} (qf);

\draw[color=red,opacity=0.3,line width=1.45mm] (0,0.062) -- (1,+0.33+0.062) -- (2,0+0.062);
\draw[color=blue,opacity=0.3,line width=1.45mm] (0,-0.062) -- (1,+0.33-0.062) -- (2,0-0.062);
\draw[color=blue,opacity=0.3,line width=2mm] (2,0) -- (3,-1) to [out=15,in=-115] (4,0);
\draw[color=red,opacity=0.3,line width=2mm] (2,0) -- (3,0) -- (4,0);

\node[draw,fill=white,circle,minimum size =\s] (qi) at (0,0) {$q_I$};

\node[draw,fill=white,circle,minimum size =\s] (q1) at (1,+1) {$q_1$};
\node[draw,fill=white,circle,minimum size =\s] (q2) at (1,+0.33) {$q_2$};
\node[draw,fill=white,circle,minimum size =\s] (q3) at (1,-0.33) {$q_3$};
\node[draw,fill=white,circle,minimum size =\s] (q4) at (1,-1) {$q_4$};

\node[draw,fill=white,circle,minimum size =\s] (q5) at (2,+1) {$q_5$};
\node[draw,fill=white,circle,minimum size =\s] (q6) at (2,0) {$q_6$};
\node[draw,fill=white,circle,minimum size =\s] (q7) at (2,-1) {$q_7$};

\node[draw,fill=white,circle,minimum size =\s] (q8) at (3,+1) {$q_8$};
\node[draw,fill=white,circle,minimum size =\s] (q9) at (3,0) {$q_9$};
\node[draw,fill=white,circle,minimum size =\s] (q10) at (3,-1) {$q_{10}$};

\node[draw,fill=white,circle,minimum size =\s] (qf) at (4,0) {$q_F$};
\end{tikzpicture}
\end{subfigure}
\caption{Divergence states}\label{fig:divergence}
\vspace{-0.2in}
\end{figure*}
\end{example}

We will later explain that if $\run(w,q)$ and $\run(w',q)$ diverge at their $k^\text{th}$ state $q^*$ and that  the common prefix of $w$ and $w'$ up to $q^*$ is $\word(\run(w,q)[k-1]) = w^*$, then $\Pr[w \in S^r(q) \text{ and } w' \in S^r(q)]$ is (in first approximation) bounded from above by
$$
\frac{\Pr[w \in S^r(q)]\cdot \Pr[w' \in S^r(q)]}{\Pr[w^* \in S^r(q^*)]}
$$
Intuitively, the farther $q^*$ is to $q$, the closer the events $w \in S^r(q)$ and $w' \in S^r(q)$ are to be independent. In the extreme case, the two runs diverge at the initial node $q_I$, then $\Pr[w^* \in S^r(q^*)] = \Pr[\lambda \in S^r(q_I)] = 1$ and the events are fully independent. 

Let $q \in \calQ^k$ and $t \in \quot{\treq}{L(q)}$. For $l < k$, we denote by 
$$
\depclass(t,l,q)
$$ the set of $t' \in \quot{\treq}{L(q)}$ such that $\run(t,q)$ and $\run(t',q)$ diverge at their $l^\text{th}$ state. Since we claim that $\Pr[t \in S^r(q) \text{ and } t' \in S^r(q)]$ depends on the divergence states of $\run(t,q)$ and $\run(t',q)$, in a sense the set $\depclass(t,l,q)$ gathers the $t'$ with the same $\Pr[t' \in S^r(q) \mid t \in S^r(q)]$. The sets $\depclass(t,l,q)$ partition $\quot{\treq}{L(q)}$.

$$
\quot{\treq}{L(q)} = \depclass(t,1,q) \cup \dots \cup \depclass(t,k+1,q)
$$

A crucial observation is that the size of $\depclass(t,l,q)$ decreases as $l$ approaches $k$. Informally, there are few pairs $(t,t')$ such that $t \in S^r(q)$ and $t' \in S^r(q)$ depend strongly on each other. 

\begin{lemma}\label{lemma:combinatorial_lemma_helper}
Let $S_1 \subseteq \quot{\treq}{\alphabet^{n_1}}$ and $S_2 \subseteq \quot{\treq}{\alphabet^{n_2}}$
be two sets containing traces of length $n_1 \in \nats$ and $n_2 \in \nats$.
Let $S = \setpred{t_1 \cdot t_2}{t_1 \in S_1, t_2 \in S_2}$.
We have that $|S| \geq \frac{1}{\cwdth\cdot n^{\cwdth}} |S_1|\cdot|S_2|$, where $n = n_1+n_2$ and $\cwdth$ is the width of the concurrent alphabet $(\alphabet, \indep)$.
\end{lemma}

\begin{proof}
First, let us count the number of distinct ways of partitioning a trace $t$ (of length $n$) into two traces $t_1$ and $t_2$ so that $t = t_1 \cdot t_2$. At a high level, this is the number of cuts in the partial order induced by $t$. The number of different cuts $N_{\sf cuts}$ in the partial order equals the number of possible subsets of the elements of $t$ that are pairwise independent. Since the width of the partial order is $\cwdth$, the number of cuts $N_{\sf cuts}$ is bounded above by $n^1 + n^2 + \ldots + n^{\cwdth} \leq \cwdth \cdot n^{\cwdth}$ since any set of pairwise independent elements from $t$ cannot have size more than $\cwdth$. This gives us $N_{\sf cuts} \leq \cwdth \cdot n^{\cwdth}$.
Now we remark that each pair $(t_1, t_2) \in S_1 \times S_2$ such that $t_1 \cdot t_2 = t$ is characterized by a unique cut in $t$. In other words, for each trace $t$, there are at most $N_{\sf cuts}$ pairs $(t_1, t_2) \in S_1 \times S_2$ for which $t = t_1 \cdot t_2$. 
This means $|S_1| \cdot |S_2| \leq |S| \cdot N_{\sf cuts} \leq |S| \cdot \cwdth \cdot n^{\cwdth}$. Thus we have
$|S| \geq  \frac{1}{\cwdth\cdot n^{\cwdth}} |S_1|\cdot|S_2|$ as desired.
\end{proof}

\begin{lemma}\label{lemma:combinatorial_lemma}
Let $q \in \calQ^\mathsf{u}$ and $t \in \quot{\treq}{L(q)}$, let $q'$ be the $k^\text{th}$ state reached by $\run(t,q)$, then 
$$
|\depclass(t,k,q)| \leq \cwdth\cdot n^\cwdth\frac{|\quot{\treq}{L(q)}|}{|\quot{\treq}{L(q')}|}
$$ 
where $\cwdth$ is the width of the concurrent alphabet $(\alphabet, \indep)$.
\end{lemma}
\begin{proof}
Suppose $q \in \calQ^h$. Let $N = |\depclass(t,k,q)|$ and $\depclass(t,k,q) = \{t_1,t_2,\dots,t_N\}$. For each trace $t_i$ let $w_i$ be the word obtained concatenating in order the symbols read in $\run(t_i,q)$. Let $w$ be the word obtained that way for $\run(t,q)$ and let $w = u\cdot v$ with $u \in \Sigma^{k-1}$ and $v \in \Sigma^{h - k + 1}$. Note that $u \in L(q')$. Since every run $\run(t_i,q)$ diverges from $\run(t,q)$ we have that, for every $i$, $w_i = u \cdot v_i$ with $v_i \in \Sigma^{h - k +1}$. It follows that for every $1 \leq i < j \leq N$, $v_i \neq v_j$ and $\eqcl{v_i} \neq \eqcl{v_j}$ (otherwise $t_i$ and $t_j$ would be equal). For every $u' \in L(q')$ and $i \in [N]$, $u'\cdot v_i$ is in $L(q)$. Let $S_1 = \quot{\treq}{L(q')}$, $S_2 = \quot{\treq}{\{v_1,\dots,v_N\}} = \{\eqcl{v_1},\dots,\eqcl{v_N}\}$ and $S = \setpred{t_1 \cdot t_2}{t_1 \in S_1, t_2 \in S_2}$, we have $S \subseteq \quot{\treq}{L(q)}$. So, by Lemma~\ref{lemma:combinatorial_lemma_helper}, $|\quot{\treq}{L(q)}| \geq |S| \geq \frac{1}{\cwdth\cdot h^{\cwdth}} |S_1|\cdot|S_2| =  \frac{1}{\cwdth\cdot h^{\cwdth}} |\quot{\treq}{L(q')}|\cdot|\depclass(t,k,q)| \geq \frac{1}{\cwdth\cdot n^{\cwdth}} |\quot{\treq}{L(q')}|\cdot|\depclass(t,k,q)|$.
\end{proof}


\subsection{FPRAS Analysis}\label{sec:analysis}

We use the notation $(1 \pm \varepsilon)N$, with $N \geq 0$, to refer to the interval $[(1-\varepsilon)N,\,(1+\varepsilon)N]$, $(1-\varepsilon)N$ and $(1+\varepsilon)N$ are the lower and upper limits of $(1 \pm \varepsilon)N$. The rest of the section is dedicated to proving~\thmref{main_result} and Lemma~\ref{lemma:main_result_core}.

\begin{theorem}
\thmlabel{main_result}
Let $\calA$ be an NFA over the fixed-size alphabet $\Sigma$ and let $\indep \subseteq \Sigma \times \Sigma$ be an independence relation. Let $\cwdth$ be the width of the concurrent alphabet $(\Sigma,\indep)$. $\TraceMC(\calA,n,\varepsilon,\delta)$ takes
$
O(\log(\delta^{-1})(\cwdth n^{\cwdth+2}\varepsilon^{-2}|\calQ|\log(n|\calQ|) +  n|\calQ|^2\log|\calQ|))
$ elementary operations and $O(\log(\delta^{-1})\cwdth n^{\cwdth+2}\varepsilon^{-2}$ $|\calQ|^2\log(n|\calQ|))$ calls to the membership oracle $\memp{\treq}$, and returns $\mathsf{est}$ with the guarantee
$$
		\Pr\left[\mathsf{est}  \in (1 \pm \varepsilon) |\quot{\treq}{L_n(\calA)}| \right] \leq 1 - \delta.
$$
\end{theorem}

\begin{lemma}\label{lemma:main_result_core}
Let $\calA^\mathsf{u}$ be the unrolled automaton of length $n$ such that $L(\calA^\mathsf{u}) \neq \emptyset$. Let $m = \max_i |\calQ^i|$, $\varepsilon > 0$, and define $\nsnt$, $\ns$, $\nt$ and $\threshold$ as in $\TraceMC$, then $\TraceMCCore(\calA^\mathsf{u},n,\ns,\nt,\threshold)$ terminates in $O(\ns\nt n m + m^2\log(m)n)$ elementary operations and $O(\threshold m) = O(\ns\nt n m^2)$ calls to the oracle $\memp{\treq}$ and returns $\mathsf{est}$ with the guarantee
$$
		\Pr\left[\mathsf{est}  \notin (1 \pm \varepsilon) |\quot{\treq}{L_n(\calA)}| \right] \leq \frac{1}{4}.
$$
\end{lemma}

\noindent Given Lemma~\ref{lemma:main_result_core}, proving of~\thmref{main_result} boils down to a standard median trick.

\begin{proof}[Proof of~\thmref{main_result}]
Let $X_i$ be the indicator variable taking value $1$ if the $i^\text{th}$ run returns $\mathsf{est}_i  \notin (1 \pm \varepsilon)|\quot{\treq}{L_n(\calA)}|$, and $0$ otherwise. Let $X = \sum_{i = 1}^{\nv} X_i$. By Lemma~\ref{lemma:main_result_core}, $\Ex[X] \leq \frac{\nv}{4}$. The $X_i$'s are independent, so we use Chernoff-Hoeffding inequality
\begin{align*}
\Pr\left[\underset{1 \leq j \leq \nv}{\median}(\mathsf{est}_i) \not\in (1 \pm \varepsilon)|\quot{\treq}{L_n(\calA)}|\right] \leq
\Pr\left[X \geq \frac{\nv}{2}\right] \leq \Pr\left[X - \Ex[X] \geq \frac{\nv}{4}\right]
 \leq \exp\left(-\frac{\nv}{8}\right) \leq \delta 
\end{align*}

The running time is $\nv$ times the running time of $\TraceMCCore$ plus the cost of unrolling the automaton and of computing the median. Unrolling takes time negligible compared to running $\TraceMCCore$. The fast median computation of $\nv$ values takes time $O(\nv)$.
\end{proof}

In the rest of the section, $\calA^\mathsf{u}$, $\nsnt$, $\ns$, $\nt$ and $\threshold$ are fixed as in the statement of Lemma~\ref{lemma:main_result_core}. We will show that $\TraceMCCore$ is unlikely to exit through line~\ref{line:interrupt}, which makes it fine to prove the tightness of the estimates in the variant of $\TraceMCCore$ that does not stop when the number of samples grows large. We call $\TraceMCCore^*$ the algorithm $\TraceMCCore$ without line~\ref{line:interrupt}.

\begin{lemma}\label{lemma:proba_p(q)}
$\TraceMCCore^*(\calA^\mathsf{u},n,\ns,\nt,\threshold)$ computes $N(q) \not\in (1 \pm \varepsilon)|\quot{\treq}{L(q)}|$ for some $q \in \calQ^\mathsf{u}$ with probability at most $1/16$.
\end{lemma}

\begin{lemma}\label{lemma:proba_S(q)}
$\TraceMCCore^*(\calA^\mathsf{u},n,\ns,\nt,\threshold)$ constructs sets $S^r(q)$ such that $\sum_{r \in [\ns\nt],q \in \calQ^\mathsf{u}} |S^r(q)| \geq \threshold$ with probability at most $1/8$.
\end{lemma}

\noindent We prove that Lemma~\ref{lemma:main_result_core} follows from the two lemmas above.
\begin{proof}[Proof of Lemma~\ref{lemma:main_result_core}]
Let $\Delta = (1 \pm \varepsilon)|\quot{\treq}{L(\calA^\mathsf{u})}|$.
Let $A$ be $\TraceMCCore(\calA^\mathsf{u},n,\ns,\nt,\threshold)$ and $A^*$ be $\TraceMCCore^*(\calA^\mathsf{u},n,\ns,\nt,\threshold)$. Let $\Sigma = \sum_{r \in [\ns\nt]}\sum_{q\in \calQ^{\mathsf{u}}} |S^r(q)|$. $A$ returns $N(q_F)$ if it exits normally and $0$ when it exits through line~\ref{line:interrupt}, which occurs if and only if $\Sigma \geq \threshold$. By assumption $|L(\calA^\mathsf{u})| > 0$, so $0 \not\in \Delta$. By Lemmas~\ref{lemma:proba_p(q)} and~\ref{lemma:proba_S(q)},
\begin{align*}
\Pr_{A}\left[\mathsf{est} \not\in \Delta\right] \leq \Pr_{A^*}\left[N(q_F) \not\in \Delta \text{ or } \Sigma \geq \threshold\right] \leq \Pr_{A^*}\left[N(q_F) \not\in \Delta\right] + \Pr_{A^*}\left[\Sigma \geq \threshold\right] \leq 1/4
\end{align*}

For the running time, let $m = \max_i |\calQ^i|$. The total number of samples is in $O(\threshold)$. Each sample goes through at most $O(m)$ {\reduce}'s so the cumulated cost of all {\reduce}'s is $O(\threshold m)$ operations. Arithmetic computations break down into (1) computing $\nt|\calQ^\mathsf{u}|$ means of $\ns$ integers, (2) computing $|\calQ^\mathsf{u}|$ minimums in subsets of $O(m)$ rational numbers, and computing $|\calQ^\mathsf{u}|$ medians of $\nt$ rational numbers; this can be done in $O((\nt\ns+\nt\log(\nt)+m)|\calQ^\mathsf{u}|)$ elementary operations. For the {\union}'s, the set $T$ can be computed and sorted w.r.t. $\prec$ ahead of time for every $q$ in time $O(m\log(m)|\calQ^\mathsf{u}|)$. The oracle $\memp{\treq}$ is called at most $m$ times for the same sample. Hence a total of $O(\threshold m)$ oracle calls and $O(\threshold m + m\log(m)|\calQ^\mathsf{u}| + \ns\nt|\calQ^\mathsf{u}|) = O(\ns\nt n m + m^2\log(m)n)$ operations.
\end{proof}

The dependence between $N(q)$, $S^r(q)$ and $\hat S^r(q)$ makes the analyze of $\TraceMCCore^*$ delicate. To deal with it, we introduce a framework for the analysis, a random process that simulates many possible runs of $\TraceMCCore^*$ at once.

\subsection{$\mathfrak{S}$-Process}

The set of ancestors of $q \in \calQ^{\mathsf{u}}$ is $\ancestors(q) = \bigcup_{q' \in \preds(q)} \{q'\} \cup \ancestors(q')$. Note that $q \not\in \ancestors(q)$. 

\begin{definition}
An \emph{history} for a state $q$ in $\calA^{\mathsf{u}}$ is a mapping from $\ancestors(q)$ to $\mathbb{R}$. An history $h$ for $q$ is \emph{realizable} when there is run of $\TraceMCCore^*(\calA^{\mathsf{u}},n,\ns,\nt,\threshold)$ that assigns $N(q')$ to $h(q')$ for all $q' \in \ancestors(q)$; any such run is called \emph{compatible} with $h$. Given two states $q_1$ and $q_2$ of $\calA^{\mathsf{u}}$, two histories for $q_1$ and $q_2$ are \emph{compatible} if they are consistent on $\ancestors(q_1) \cap \ancestors(q_2)$. The only possible history for the initial state is the empty history $\emptyset$ since it has no predecessor.
\end{definition}

The $\mathfrak{S}$-process simulates many possible runs of $\TraceMCCore^*$ all at once. In the $\mathfrak{S}$-process, the random sets $S^r(q)$ and $\hat S^r(q)$ are simulated by the random sets $\mathfrak{S}^r(q)$ and $\hat{\mathfrak{S}}^r(q)$. These random sets are called $\mathfrak{S}$-variables. There is one random set $\vZ{r}{h}{q}$ per realizable history $h$ for $q$, and one random set $\vY{r}{h}{v}{q}$ per realizable history $h$ for $q$ and per value $v$ that can be computed for $N(q)$. $\vZ{r}{h}{q}$ simulates $\hat S^r(q)$ in runs of $\TraceMCCore^*$ compatible with $h$, and $\vY{r}{h}{v}{q}$ simulates $S^r(q)$ for the same runs where, in addition, $\TraceMCCore^*$ sets $N(q)$ to $v$. 

\myparagraph{$\mathfrak{S}$-variables} Fix $r \in [\nsnt]$.
\begin{enumerate}[leftmargin=*]
\item[•] if $q$ is the initial state $q_I$ of $\calA^{\mathsf{u}}$, then the only possible history for $q$ is the empty history $\emptyset$, the only value for $N(q)$ is $1$, and the only value for $S^r(q)$ is $\{\eqcl{\lambda}\}$. So we define $\vY{r}{\emptyset}{1}{q} = \{\eqcl{\lambda}\}$.
\item[•] if $q\neq q_I$ then let $\preds(q) = \{q_1,\dots,q_k\}$. Fix a realizable history $h$ for $q$ and let $h_i$ be the restriction of $h$ to $\ancestors(q_i)$. $h_i$ is a realizable history for $q_i$. Let $v_i = h(q_i)$ and $m_h(q) = \max(v_1,\dots,v_k)$. The random sets $\vY{r}{h_1}{v_1}{q_1},\dots,\vY{r}{h_k}{v_k}{q_k}$ exists and we define, for all $i \in [k]$,
\begin{align*}
\vW{r}{h}{q_i,q} &= \reduce\left(\vY{r}{h_i}{v_i}{q_i},\frac{v_i}{m_h(q)}\right)
\\
\vZ{r}{h}{q} &= \union(q,\vW{r}{h}{q_1,q},\dots,\vW{r}{h}{q_k,q})
\end{align*}
Every value $v$ that can be given to $N(q)$ in a run of $\TraceMCCore^*$ compatible with $h$ verifies $v \geq N_{max}(q) = m_h(q)$. We define
$$
\vY{r}{h}{v}{q} = \reduce\left(\vZ{r}{h}{q},\frac{m_h(q)}{v}\right)
$$
\end{enumerate}

$\vW{r}{h}{q_i,q}$ and $\vZ{r}{h}{q}$ are the constructed like $\bar S^r(q_i,q)$ and $\hat S^r(q)$ in \estimateAndSample, lines~\ref{line:sbar} and~\ref{line:shat}, with $v_i$ replacing $N(q_i)$ and $\vY{r}{h_i}{v_i}{q_i}$ replacing $S^r(q_i)$. Similarly, $\vY{r}{h}{v}{q}$ is built as in line~\ref{line:s} with $v$ replacing $N(q)$ and $\vZ{r}{h}{q}$ replacing $\hat S^r(q)$. $m_h(q)$ is a constant that plays the role of the random variable $N_{max}(q)$.  $N_{max}(q)$ equals $m_h(q)$ in all runs of $\TraceMCCore^*$ compatible with $h$.

\begin{restatable}{lemma}{coupling}\label{lemma:coupling}
Let $H(q)$ be the random variable recording the history for $q$ in a run of $\TraceMCCore^*$. Let $h$ be a realizable history for $q$. Let $e(\hat S^1(q),\dots,\hat S^\nsnt(q))$ be an event that is function of $\hat S^1(q),\dots,\hat S^\nsnt(q)$ and let $e(\vZ{1}{h}{q},\dots,\vZ{\nsnt}{h}{q})$ be the same event where each $\hat S^r(q)$ is replaced by $\vZ{r}{h}{q}$. Then 
$$
\Pr[H(q) = h \text { and } e(\hat S^1(q),\dots,\hat S^\nsnt(q))] \leq \Pr[e(\vZ{1}{h}{q},\dots,\vZ{\nsnt}{h}{q})]
$$
Similarly, let $e(S^1(q),\dots,S^\nsnt(q))$ be an event that is function of $S^1(q),\dots,S^\nsnt(q)$ and let $e(\vY{1}{h}{v}{q},\dots,$ $\vY{\nsnt}{h}{v}{q})$ be the same event where each $S^r(q)$ is replaced by $\vY{r}{h}{v}{q}$. Then 
$$
\Pr[H(q) = h \text { and } N(q) = v \text { and } e(S^1(q),\dots,S^\nsnt(q))] \leq \Pr[e(\vY{1}{h}{v}{q},\dots,\vY{\nsnt}{h}{v}{q})]
$$
\end{restatable}
The first advantage of the $\mathfrak{S}$-process over $\TraceMCCore^*$ is that its $\nsnt$ subprocesses identified by the superscripts $r$ are totally disjoint and thus independent. The second advantage is that can find the exact probability that a trace is in a $\mathfrak{S}$-variable and that we can bound from above the probability that two distinct traces are in two $\mathfrak{S}$-variables. The proofs of these results are deferred to appendix.

\begin{restatable}{lemma}{probaFirstOrder}\label{lemma:proba_first_order}
For $q \in \calQ^\mathsf{u}$, every $r \in [\nsnt]$, every realizable history $h$ for $q$, every value $v$ such that $\vY{r}{h}{v}{q}$ is defined, and every $t \in \quot{\treq}{L(q)}$, we have 
$$
\Pr[t  \in \vY{r}{h}{v}{q}] = v^{-1} \quad \text{and} \quad \Pr[t \in \vZ{r}{h}{q}] = m_h(q)^{-1}
$$ 
\end{restatable}

\begin{restatable}{lemma}{probaSecondOrder}\label{lemma:proba_second_order}
For $k > 0$, let $q,q' \in \calQ^k$, $t \in \quot{\treq}{L(q)}$ and $t' \in \quot{\treq}{L(q')}$ such that $(t,q) \neq (t',q')$. Let $q^*$ be the divergence state of $\run(t,q)$ and $\run(t',q)$. Let $h$ and $h'$ be compatible histories for $q$ and $q'$, respectively, then 
$$
\Pr[t \in \vY{r}{h}{v}{q},\, t' \in \vY{r}{h'}{v'}{q'}] \leq \frac{h(q^*)}{v'v}
\quad \text{and} \quad \Pr[t \in \vZ{r}{h}{q},\, t' \in \vZ{r}{h'}{q'}] \leq \frac{h(q^*)}{m_h(q)m_{h'}(q')}
$$
\end{restatable}

\subsection{All Estimates are Tight -- Proof of Lemma~\ref{lemma:proba_p(q)}}

For convenience let $\Delta(q) = (1 \pm \varepsilon)|\quot{\treq}{L(q)}|$. 
Here, we prove that the event $\bigcup_{q \in \calQ^{\mathsf{u}}} N(q) \not\in \Delta(q)$ occurs with probability $P$ at most $1/16$ in $\TraceMCCore^*$. Order the states of $\calQ^\mathsf{u}$ so that every node comes after its predecessors. $q_I$ is the first in that ordering. It is impossible that $N(q_I) \not\in \Delta(q_I)$ since $N(q_I) = 1 = |\quot{\treq}{L(q)}|$. If there are states $q$ such that $N(q) \not\in \Delta(q)$, then there better be one that comes first in the state ordering and it has to be different from $q_I$. Thus,
\begin{align*}
P := \Pr\left[\bigcup_{q \in \calQ^\mathsf{u}} N(q) \not\in \Delta(q)\right]
= \Pr\left[\bigcup_{q \in \calQ^{> 0}} N(q) \not\in \Delta(q) \text{ and } \forall q' \in \ancestors(q),\, N(q') \in \Delta(q')\right]
\\
\leq \sum_{q \in \calQ^{> 0}} \Pr\left[N(q) \not\in \Delta(q) \text{ and } \forall q' \in \ancestors(q),\, N¨(q') \in \Delta(q')\right]
\end{align*}

Let $\calH_q$ be the set of realizable histories for $q$ that verify $N(q') \in \Delta(q')$ for all ancestors $q'$ of $q$ and $H(q)$ be the random variable that records the history for $q$. 
\begin{align*}
P \leq \sum_{q \in \calQ^{> 0}} \sum_{h \in \calH_q} \Pr[H(q) = h \text{ and } N(q) \not\in \Delta(q)]
\end{align*}

Recall that $N(q) = \round(q,\max(N_{max}(q),\hat N(q)))$. A real value $v$ is \emph{acceptable for $q \in Q$} if it an integer between $1$ and $2^n$ or if it is of the form
$
(1+\varepsilon)\cdot \ell
$ or
$
(1-\varepsilon)\cdot \ell 
$  
for some integer $\ell$ between $1$ and $2^n$. The function $\round(q,v)$ takes in $q$ and a value $v$ and returns the lowest acceptable value that this greater than or equal to $v$.  Rounding cannot make a value leave $\Delta(q)$.

\begin{lemma}\label{lemma:rounding_of_good_is_still_good}
If $v \in \Delta(q)$ then $\round(q,v) \in \Delta(q)$.
\end{lemma}
\begin{proof}
Since $|\quot{\treq}{L(q)}|$ is an integer between $1$ and $2^n$, the lower and upper limits of $\Delta(q)$ are acceptable values for $q$, hence the result.
\end{proof}

\begin{lemma}\label{lemma:few_acceptable_values_in_the_right_range}
Let $q \in \calQ^\mathsf{u}$. There are between $2$ and $2^{n+2}$ acceptable values for $q$ in $\Delta(q)$. It follows that $|\calH_q| \leq 2^{(n+2)|\calQ^\mathsf{u}|}$.
\end{lemma}
\begin{proof}
There is at least one two acceptable values for $q$ in $\Delta(q)$ because the limits of $\Delta(q)$ are acceptable. The $3 \cdot 2^{n}$ upper bound follows from the definition as every $\ell \in [2^n]$ can give three acceptable values: $\ell$, $(1-\varepsilon)\cdot \ell$ and $(1+\varepsilon)\cdot\ell$. Since there are less than $|Q^{\sf u}|$ ancestors for $q$, the bound on $|\calH_q|$ follows.
\end{proof}

We claim that, when $h \in \calH_q$, 
\begin{align*}
[H(q) = h] \text{ and } [\hat N(q) \in \Delta(q)]
\quad &\Rightarrow \quad 
[H(q) = h] \text{ and } [\max(N_{max}(q),\hat N(q)) \in \Delta(q)]
\\
&\Rightarrow \quad 
[H(q) = h] \text{ and } [N(q) \in \Delta(q)] 
\end{align*}  
The second implication follows from Lemma~\ref{lemma:rounding_of_good_is_still_good}. For the first implication, if $\hat{N}(q) \geq N_{max}(q)$ then $\hat N(q) \in \Delta(q)$ clearly implies $N(q) \in \Delta(q)$. Otherwise, if $\hat{N}(q) < N_{max}(q)$, then $N(q) = N_{max}(q)$. Thus, $\hat N(q) \in \Delta(q)$ directly implies the low bound $N(q) \geq (1-\varepsilon)|\quot{\treq}{L(q)}|$. For the upper bound, recall that $N_{max}(q) = \max(N(q') \mid q' \in \preds(q))$ and observe that, since $h \in \calH_q$, we have $N(q') \in \Delta(q')$ for all $q' \in \preds(q)$ and thus $N_{max}(q) \leq (1+\varepsilon)\max(|\quot{\treq}{L(q')}| \mid q' \in \preds(q))$. Since $|\quot{\treq}{L(q')}| \leq |\quot{\treq}{L(q)}|$ holds for all $q' \in \preds(q)$, we conclude that $N(q) = N_{max}(q) \leq (1+\varepsilon)|\quot{\treq}{L(q)}|$. This ends the proof of the implication. We can now write
\begin{align*}
P \leq \sum_{q \in \calQ^{> 0}} \sum_{h \in \calH_q} \Pr[H(q) = h \text{ and } \hat N(q) \not\in \Delta(q)] 
\end{align*}
We have
$$
\hat N(q) = N_{max}(q)\cdot \mom(\beta,|\hat S^1(q)|,\dots,|\hat S^\nsnt(q)|) = N_{max}(q)\cdot\median(M_1,\dots,M_{\nt})
$$ 
with $M_j = \frac{1}{\ns}(|\hat{S}^{\ns j + 1}(q)| + \dots + |\hat{S}^{\ns(j + 1)}(q)|)$. We now move to the $\mathfrak{S}$-process. Let $\vM{h}{j}{q} = \frac{1}{\ns}(|\vZ{\ns j + 1}{h}{q}| + \dots + |\vZ{\ns(j + 1)}{h}{q}|)$. Using Lemma~\ref{lemma:coupling} we obtain
\begin{align*}
\Pr[H(q) = h \text{ and } \hat N(q) \not\in \Delta(q)] \leq 
\Pr[m_{h}(q)\cdot \median(\vM{h}{1}{q},\dots,\vM{h}{\nt}{q}) \not\in \Delta(q)]
\end{align*}
Now we can work only in the $\mathfrak{S}$-process and use that the random sets $\{|\vZ{r}{h}{q}|\}_{r \in [\nsnt]}$ are i.i.d. to invoke Lemma~\ref{lem:medians_of_means}. Let us assume the following for now:

\begin{lemma}\label{lemma:variance} Let $\cwdth$ be the width of the concurrent alphabet $(\Sigma,\indep)$. If $h \in \calH_q$ then
$\Ex[|\vZ{r}{h}{q}|] = m_h(q)^{-1}|\quot{\treq}{L(q)}|$ and 
$
\Var[|\vZ{r}{h}{q}|] \leq  2\cdot\cwdth\cdot n^{\cwdth+1} \cdot (1+\varepsilon)\cdot \Ex[|\vZ{r}{h}{q}|]^2
$.
\end{lemma} 
 
\noindent Let $\mathfrak{M} = \median(\vM{h}{1}{q},\dots,\vM{h}{\nt}{q})$ and $\mu = m_h(q)^{-1}|\quot{\treq}{L(q)}|$ then 
$$
\Pr[m_{h}(q)\cdot \median(\vM{h}{1}{q},\dots,\vM{h}{\nt}{q}) \not\in \Delta(q)] = \Pr[|\mathfrak{M} - \mu| \geq \varepsilon\mu]
$$
We combine Lemma~\ref{lem:medians_of_means} with the variance bound from Lemma~\ref{lemma:variance}.
$$
\Pr[|\mathfrak{M} - \mu| \geq \varepsilon\mu] \leq \exp\left(-\nt\left(1 - \frac{4\cdot\cwdth\cdot n^{\cwdth+1} \cdot (1+\varepsilon)\cdot \mu^2}{\varepsilon^2 \mu^2 \ns}\right)\right) = \exp\left(-\nt\left(1 - \frac{4\cdot\cwdth\cdot n^{\cwdth+1} \cdot (1+\varepsilon)}{\varepsilon^2 \ns}\right)\right)
$$
Plugging in the values for $\ns$ and $\nt$, we find that $\Pr[|\mathfrak{M} - \mu| \geq \varepsilon\mu] \leq 1/(16\cdot|\calQ^{\mathsf{u}}|\cdot 2^{(n+2)\cdot|\calQ^{\mathsf{u}}|})$. Hence
$$
P \leq \sum_{q \in \calQ^{> 0}} \sum_{h \in \calH_q} \frac{1}{16\cdot|\calQ^{\mathsf{u}}|\cdot 2^{(n+2)\cdot|\calQ^{\mathsf{u}}|}} \leq \frac{|\calH_q|}{16\cdot 2^{(n+2)\cdot|\calQ^{\mathsf{u}}|}} \leq \frac{1}{16}
$$

\noindent Thus, the proof of Lemma~\ref{lemma:proba_p(q)} will be complete with the proof of Lemma~\ref{lemma:variance}.

\begin{proof}[Proof of Lemma~\ref{lemma:variance}]
$\Ex[|\vZ{r}{h}{q}|] = m_h(q)^{-1}|\quot{\treq}{L(q)}|$ follows directly from Lemma~\ref{lemma:proba_first_order}.
\begin{align*}
\Var[|\vZ{r}{h}{q}|] \leq \Ex[|\vZ{r}{h}{q}|^2] = \Ex[|\vZ{r}{h}{q}|] + 
\sum_{t \in \quot{\treq}{L(q)}}\sum_{\substack{t' \neq t \\ t' \in \quot{\treq}{L(q)} }} \Pr[t \in \vZ{r}{h}{q} \text{ and } t' \in \vZ{r}{h}{q}]
\end{align*}
Fix $t \in \quot{\treq}{L(q)}$ and let $q^t_0,\dots,q^t_k = q$ be the sequence of nodes in $\run(t,q)$. We partition $(\quot{\treq}{L(q)})\setminus \{t\}$ into $\depclass(t,1,q) \cup \dots \cup \depclass(t,k+1,q)$. If $t' \in \depclass(t,i,q)$ then, by Lemma~\ref{lemma:proba_second_order}, we have that
$$
 \Pr[t \in \vZ{r}{h}{q} \text{ and } t' \in \vZ{r}{h}{q}] \leq \frac{h(q^t_i)}{m_h(q)^2}
$$
We use this in the variance computation with the bound on $|\depclass(t,i,q)|$ from Lemma~\ref{lemma:combinatorial_lemma}.
\begin{align*}
\sum_{t}\sum_{t' \neq t} \Pr[t \in \vZ{r}{h}{q} \text{ and } t' \in \vZ{r}{h}{q}] 
= \sum_{t}\sum_{i = 1}^{k+1} \sum_{t' \in \depclass(t,i,q)}  \Pr[t \in \vZ{r}{h}{q} \text{ and } t' \in \vZ{r}{h}{q}] 
\\
\leq \sum_{t}\sum_{i = 1}^{k+1} \sum_{t' \in \depclass(t,i,q)}  \frac{h(q^t_i)}{m_h(q)^2} = \sum_{t}\sum_{i = 1}^{k+1} |\depclass(t,i,q)|  \frac{h(q^t_i)}{m_h(q)^2} \leq \sum_{t}\sum_{i = 1}^{k+1} \cwdth\cdot n^\cwdth \frac{|\quot{\treq}{L(q)}|}{|\quot{\treq}{L(q_i)}|} \frac{h(q^t_i)} {m_h(q)^2}
\end{align*}
Since $h \in \calH_q$, $h(q^t_i)$ is in $\Delta(q^t_i)$, we have $h(q^t_i)\cdot |\quot{\treq}{L(q_i)}|^{-1} \geq 1 - \varepsilon$, so this sum is bounded by
\begin{align*}
\sum_{t}\cwdth\cdot n^{\cwdth+1} \cdot (1+\varepsilon)\cdot \frac{|\quot{\treq}{L(q)}|}{m_h(q)^2}  
= \cwdth\cdot n^{\cwdth+1} \cdot (1+\varepsilon)\cdot \frac{|\quot{\treq}{L(q)}|^2}{m_h(q)^2} = \cwdth\cdot n^{\cwdth+1} \cdot (1+\varepsilon)\cdot \Ex[|\vZ{r}{h}{q}|]^2 
\end{align*}
So we have shown
$
\Var[|\vZ{r}{h}{q}|] \leq \Ex[|\vZ{r}{h}{q}|] +   \cwdth\cdot n^{\cwdth+1} \cdot (1+\varepsilon)\cdot   \Ex[|\vZ{r}{h}{q}|]^2
$.  
To end the proof, we show that $\Ex[|\vZ{r}{h}{q}|] \leq \cwdth\cdot n^{\cwdth+1} \cdot (1+\varepsilon)\cdot   \Ex[|\vZ{r}{h}{q}|]^2$. First, $\cwdth \cdot  n^{\cwdth+1} \geq 1$ comes from $\cwdth \geq 1$. Second, because $h \in \calH_q$, we have $m_h(q) \leq (1+\varepsilon)\max(|\quot{\treq}{L(q')}| \mid q' \in \preds(q)) \leq (1 + \varepsilon)|\quot{\treq}{L(q)}|$, so $\Ex[|\vZ{r}{h}{q}|] = m_h(q)^{-1}|\quot{\treq}{L(q)}| \geq \frac{1}{1+\varepsilon}$ and thus $(1+\varepsilon)\cdot   \Ex[|\vZ{r}{h}{q}|]^2 \geq \Ex[|\vZ{r}{h}{q}|]$. Hence
$$
\Var[|\vZ{r}{h}{q}|] \leq 2\cdot\cwdth\cdot n^{\cwdth+1} \cdot (1+\varepsilon)\cdot   \Ex[|\vZ{r}{h}{q}|]^2
$$ 
\end{proof}

\subsection{Few Samples are Needed -- Proof of Lemma~\ref{lemma:proba_S(q)}}\label{subsection:sample_sets_never_get_large}

We show that all $|S^r(q)|$ stay below $\threshold$ with probability at least $\frac{7}{8}$ when running $\TraceMCCore^*$.
\begin{align*}
&\Pr\left[\bigcup_{r \in [\nsnt]}\bigcup_{q \in \calQ^\mathsf{u}} |S^r(q)| \geq \threshold \right]
\\
&\leq \Pr\left[\bigcup_{q \in \calQ^\mathsf{u}} N(q) \in \Delta(q)\right] + \Pr\left[\bigcap_{q \in \calQ^\mathsf{u}} N(q) \not\in \Delta(q) \text{ and } \bigcup_{r \in [\nsnt]}\bigcup_{q \in \calQ^\mathsf{u}} |S^r(q)| \geq \threshold \right]
\\
&\leq \frac{1}{16} + \sum_{r \in [\nsnt]} \sum_{q \in \calQ^\mathsf{u}} \Pr\left[N(q) \in \Delta(q) \text{ and } |S^r(q)| \geq \threshold \right] \tag{Lemma~\ref{lemma:proba_p(q)}}
\end{align*}
Next,
\begin{align*}
\Pr&\left[N(q) \in \Delta(q) \text{ and } |S^r(q)| \geq \threshold \right] 
\leq \Pr\left[\sum_{t \in \quot{\treq}{L(q)}} \mathbbm{1}(t \in S^r(q)) \cdot \mathbbm{1}(N(q) \in \Delta(q)) \geq \theta \right]
\\ &\leq \frac{1}{\theta}\cdot \Ex\left[\sum_{t \in \quot{\treq}{L(q)}} \mathbbm{1}(t \in S^r(q)) \cdot \mathbbm{1}(N(q) \in \Delta(q))\right] \qquad\text{(Markov's inequality)}
\\
&\leq \frac{1}{\theta}\sum_{t \in \quot{\treq}{L(q)}} \Ex\left[\mathbbm{1}(t \in S^r(q)) \cdot \mathbbm{1}(N(q) \in \Delta(q))\right] = \frac{1}{\theta}\sum_{t \in \quot{\treq}{L(q)}} \Pr\left[t \in S^r(q)\text{ and } N(q) \in \Delta(q)\right]
\end{align*}
We show the following in appendix.
\begin{restatable}{lemma}{whiteMagicLemma}\label{lemma:white_magic_lemma}
For every $t \in \quot{\treq}{L(q)}$ and $r \in [\nsnt]$ we have $\Ex\left[\mathbbm{1}(t \in S^r(q))\cdot N(q)\right] = 1$.
\end{restatable}
We use Lemma~\ref{lemma:white_magic_lemma} plus the fact that the random quantity $\mathbbm{1}(N(q) \in \Delta(q)) \cdot N(q)^{-1}$ only takes value less than $((1-\varepsilon)\cdot |\quot{\treq}{L(q)}|)^{-1}$.
\begin{align*}
\Pr\left[t \in S^r(q)\text{ and } N(q) \in \Delta(q)\right] = \Ex\left[\mathbbm{1}(t \in S^r(q)) \cdot \mathbbm{1}(N(q) \in \Delta(q))\right]
\\
=  \Ex\left[\mathbbm{1}(t \in S^r(q)) \cdot \mathbbm{1}(N(q) \in \Delta(q)) \cdot N(q) \cdot N(q)^{-1}\right] 
\\
\leq  \Ex\left[\mathbbm{1}(t\in S^r(q)) \cdot N(q)\right]  \cdot \frac{1}{(1-\varepsilon)\cdot |\quot{\treq}{L(q)}|} 
\\
\leq  \frac{1}{(1-\varepsilon)\cdot |\quot{\treq}{L(q)}|}
\end{align*}
Using this bound and the value of $\theta$, we finish the proof of Lemma~\ref{lemma:proba_S(q)}.
\begin{align*}
\Pr\left[\bigcup_{r \in [\alpha]} \bigcup_{q \in \calQ^{\mathsf{u}}} |S^r(q)| \geq \theta\right] &\leq \frac{1}{16} + \frac{1}{\theta}\sum_{r \in [\alpha]}\sum_{q \in Q} \sum_{t \in \quot{\treq}{L(q)}} \frac{1}{(1-\varepsilon) \cdot |\quot{\treq}{L(q)}|}
\leq \frac{1}{16} + \frac{\alpha\cdot|\calQ^{\mathsf{u}}|}{(1-\varepsilon)\cdot \theta} \leq  \frac{1}{8}
\end{align*}


\section{Conclusion} 
\seclabel{conclusion}

We studied the problem of counting the number of trace equivalence classes
that intersect a  regular language, from an algorithmic and complexity-theoretic lens.
We show that the problem is \#P-complete even for DFAs, but admits a
fully polynomial randomized approximation scheme (FPRAS) even for the case of NFAs,
generalizing recent results on \#NFA~\cite{MCMNFA2024,MeeldCPODS2025,ACJR19,ACJR21}.
To the best of our knowledge, our work is the first to present a principled theoretical 
framework for designing counting and sampling modulo equivalences 
by decomposing the core problem into a few core sub-problems such 
as equivalence checking, membership modulo equivalence, 
and generation of normal forms. 
We believe these core insights will also generalize to other notions of equivalences and symmetries on combinatorial structures.
Building on our FPRAS, we also presented an almost-uniform sampling algorithm.

Several promising research directions remain. First, improving 
the runtime complexity of our FPRAS is an important challenge. 
Second, our current sampling approach requires $O(n)$ calls to 
the counting routine. Developing a more efficient sampler that 
requires fewer counter invocations while maintaining the 
distribution guarantees would be valuable. 
Finally, extending our techniques to more expressive language 
classes and notions of equivalences~\cite{FarzanMathur2024} would enable 
analysis of richer space of program behaviors. 

\bibliographystyle{ACM-Reference-Format}
\bibliography{references}

\clearpage

\appendix


\section{Proofs from \secref{prelim}}

The median-of-means estimator for random variables $X_1, \ldots, X_\alpha$
with identical mean is shown in \algoref{mom}.

\begin{algorithm}[H]
	\caption{	\algolabel{mom}
$\mom(\beta,X_1,\dots,X_\alpha)$}
	\begin{algorithmic}[1]
		\State $\gamma \gets \alpha/\beta$
		\State \textbf{for} $1 \leq i \leq \gamma$ \textbf{do} $Y_i \gets \sum_{j = 1}^\beta X_{j + (i-1)\cdot \gamma}$
		\State \Return $\median(Y_1,\dots,Y_\gamma)$
	\end{algorithmic}
\end{algorithm}

\MoM*
\begin{proof}
For all $i \in [\beta]$
$$
\Ex[Y_i] = \mu \quad \text{ and } \quad \Var[Y_i] = \frac{\sigma^2}{\beta}
$$
By Chebyshev inegatily:
$$
\Pr\left[|Y_i - \mu| \geq \epsilon\right] \leq \frac{\sigma^2}{\epsilon^{2}\beta}
$$
Let $W_i$ be the random variable taking value $1$ if $|Y_i - \mu| \geq \epsilon$ and $0$ otherwise. Let $W = \sum_{i \in [\gamma]} W_i$. The inequality above implies that $\Ex[W_i] \leq \sigma^2\epsilon^{-2}\beta^{-1}$ and $\Ex[W] = \gamma\cdot\Ex[W _1] \leq \gamma \sigma^2\epsilon^{-2}\beta^{-1}$. Since the $Y_i$'s are independent, so are the $W_i$'s, thus Hoeffding inequality gives 
\begin{align*}
\Pr\left[|Z - \mu| \geq \epsilon\right] \leq \Pr\left[W \geq \frac{\gamma}{2}\right] = \Pr\left[W - \Ex[W] \geq \frac{\gamma}{2} - \Ex[W]\right]
\\
\leq 
\Pr\left[W - \Ex[W] \geq \frac{\gamma}{2} - \frac{\gamma\sigma^2}{\epsilon^{2}\beta}\right]
\\
\leq 
\exp\left(-\gamma\left(1 - \frac{2\sigma^2}{\epsilon^2\beta}\right)\right)
\end{align*}
\end{proof}
\section{Proofs from \secref{sampling}}

\insertA*
\begin{proof}
$w$ may contain letters $a$, so we differentiate the new letter by calling it $a'$. Let $\pi$ be the permutation of $[l]$ and $j$ be the integer such that $\nf{w \cdot a'} = \ell_{\pi(1)}\dots \ell_{\pi(j)}\cdot a' \cdot \ell_{\pi(j+1)} \dots \ell_{\pi(l)} := w'_1 a w'_2$. Let $w_1 = \ell_{1}\dots \ell_{j}$ and $w_2 = \ell_{j+1} \dots \ell_l$. $\nf{w a'}$ is obtained from $\nf{w}\cdot a$ by a sequence of contiguous letter transpositions given by $\indep$. Applying the same transpositions, but without those transpositions acting on $a'$, turns $\nf{w} = w_1w_2$ into $w'_1 w'_2$. We then have $w_1 \preceq_{\sf lex} w'_1$ by $\lexord$ minimality of $\nf{w}$. On the other hand, since $\nf{w  a'} \lexord \nf{w} \cdot a' = w_1w_2\cdot a'$, we find that $w'_1 \preceq_{\sf lex} w_1$ by $\lexord$ minimality of $\nf{w a}$. Thus $w_1 = w'_1$ holds. It follows that the sequence of transpositions actually turns $w_1w_2$ into $w_1w'_2$ and thus $w_2 \treq  w'_2$. So, by $\lexord$ minimality of $\nf{w}$ we have $w_1w_2 \preceq_{\sf lex} w_1w'_2$ and therefore $w_2 \preceq_{\sf lex} w'_2$. But $w_2 \treq  w'_2$ also implies that $\nf{w a'} \treq w'_1  a'  w_2$ so we use that $\nf{w a'}$ is $\lexord$ minimal to derive $w'_2 \preceq_{\sf lex} w_2$ and therefore $w_2 = w'_2$.
\end{proof}

\uPrefixOne*
\begin{proof}
Let $\nf{w} = \ell_1\dots\ell_l$ and $0 \leq k \leq \ell$ such that $u' = \ell_1\dots \ell_k$. Let $i$ be the largest integer such that $\nf{w \cdot a} = \ell_1 \dots \ell_{i-1} a \ell_i\dots \ell_l$. Let $k+1 \leq j \leq l$ be the largest integer such that $\ell_j = c$. We have $\nf{w a} = \nf{\nf{w} \cdot a}$. Since $(a,c) \in \dep$, we have $i > j$ and thus $i > k+1$. Therefore, $u'$ and $u' \cdot \ell_{k+1}$ are prefixes of $\nf{w \cdot a}$. Since both are also prefixes of $\nf{w}$ and since $u'$ is the $u$-prefix of $\nf{w}$, we deduce that $u'$ is the $u$-prefix of $\nf{w a}$. We have that $v_1 = \ell_{k+1}\dots\ell_{j}$ and $v_2 = \ell_{j+1}\dots\ell_l$.
\end{proof}

\uPrefixTwo*
\begin{proof}
The lemma is immediate when $v = \lambda$, so we assume otherwise. Since $a$ is independent of every letter in $v$, we have $wa \treq u'va \treq u'av \treq \nf{u'a}\cdot v$. If $\nf{u' a}\neq u'a$ then $\nf{u' a} \lexord u'a$ and thus $\nf{u' a}\cdot v \lexord u'av$. Let $u' = \ell_1\dots\ell_k$. Using Lemma~\ref{lemma:insert_a}, consider the largest $i \leq k$ such that $\nf{u'a} = \ell_1\dots\ell_{i-1}a\ell_{i}\dots\ell_k$. 
$\nf{u'a} \lexord u'a$ implies $\ell_1\dots\ell_{i-1}a \lexord \ell_1\dots\ell_{i-1}\ell_i$. Let $u'_1 = \ell_1\dots\ell_{i-1}$ (potentially equals $\lambda$) and $u'_2 = \ell_i\dots\ell_k$ (different from $\lambda$).

By Lemma~\ref{lemma:insert_a} we either have $\nf{w a} = u' v_1 a v_2$ with $v = v_1v_2$ or $\nf{w a} = u'_3 a u'_4 v$ with $u' = u'_3 u'_4$. But every $u'v_1av_2$ starts with $\ell_1\dots\ell_{i-1}\ell_i$ so already know that $\nf{u'a}\cdot v \lexord u' v_1a v_2$. Therefore $\nf{w a}$ is of the form $u'_3 a u'_4 v$. We must then have $u'_3 au'_4 v\preceq_{\sf lex} u'_1 a u'_2 v$ and $u'_3 au'_4 v\treq u'_1 a u'_2 v$, and therefore $u'_3 a u'_4 \treq u'_1au'_2$. But $u'_1au'_2 = \nf{u'a}$ is minimal for $\lexord$, so $u'_3 a u'_4 = u'_1au'_2$. Therefore $\nf{wa} = \nf{u'a}\cdot v = u'_1 a u'_2 v$.

Let $b$ be the first letter of $v$. To finish the proof, it is sufficient to show that $u'_1au'_2b$ is not an DAG-prefix of $G(\trpo{u})$. We assume that $G(\trpo{u'_1au'_2}) = G(\trpo{u'a})$ is a DAG prefix   of $G(\trpo{u})$ otherwise we are done already. Since $G(\trpo{u'})$ is a DAG-prefix of $G(\trpo{u})$, there is $(a,i)$ in the border of $G(\trpo{u'})$. Now we show that $b$ is not a letter in the border of $G(\trpo{u'_1au'_2}) = G(\trpo{u'a})$. Since $u'b$ is a prefix of $w$ but the $u$-prefix of $w$ is $u'$, $b$ cannot be a letter of the border of $G(\trpo{u'})$. So if $b$ is a letter of the border of $G(\trpo{u'a})$ then there is $(b,j)$ adjacent to $(a,i)$ in $G(\trpo{u})$. But then $(a,b) \in \dep$, which contradicts the lemma's assumptions.
\end{proof}

\uPrefixThree*
\begin{proof}
The lemma is clear when $v = \lambda$, so we assume otherwise. We have that $\nf{v} = v$ for otherwise $\nf{w} \neq u'v$. If $b = a$ then $u'a = u'b$; since we know that $G(\trpo{u'b})$ is not a DAG-prefix of $G(\trpo{u})$ we indeed have that the $u$-prefix of $\nf{w a}$ is $u'$. For the rest of the proof we suppose $b \neq a$. Let $v = \ell_1\dots\ell_k$, with $b = \ell_1$. By Lemma~\ref{lemma:insert_a}, $\nf{av} = \nf{va} = \ell_1\dots \ell_{i-1}a\ell_i \dots \ell_k$ for some $i$. Since $(a,b) \in \indep$ for all letters $b$ in $v$, we have $i > 1$ if and only if $b \lexord a$. We consider the cases $b \lexord a$ and $a \lexord b$ separately. First we note that $\nf{wa}$ is either $u'\nf{av}$ or $\nf{u'a}v  = u'av$.

If $b \lexord a$ then $u'\nf{av} \lexord u'av$ and both $u'$ and $u'b$ are prefixes of $\nf{wa}$. We know that $G(\trpo{u'b})$ is not a DAG-prefix of $G(\trpo{u})$, so the $u$-prefix of $\nf{w a}$ is indeed $u'$ in that case. We then have $v_1 = \ell_1\dots\ell_{i-1}$ and $v_2= \ell_i\dots\ell_k$.

If $a \lexord b$ then $u'\nf{av} = u'av$, so $\nf{wa} = u'av$. It then is enough to show that, in this case, $G(\trpo{u'ab})$ is not a DAG-prefix of $G(\trpo{u})$. This is clear when $G(\trpo{u'a})$ is not a DAG-prefix of $G(\trpo{u})$. So we assume otherwise and it now suffices to show that $b$ is not a letter in the border of $G(\trpo{u'a})$. Since $u'b$ is a prefix of $w$ but the $u$-prefix of $w$ is $u'$, $b$ cannot be a letter of the border of $G(\trpo{u'})$. So if $b$ is a letter of the border of $G(\trpo{u'a})$ then there is $(b,j)$ adjacent to $(a,i)$ in $G(\trpo{u})$. But then $(a,b) \in \dep$, which contradicts the lemma's assumptions. 
\end{proof}

\prefixValidatorSoundness*
\begin{proof}
States in $\calQ_u$ are of the form $(u',b,L)$. There are $|\alphabet|+1$ choices for $b$ and $2^{|\alphabet|}$ choices for $L$. Since $u$ is fixed, the number of choices for $u'$ is bounded by the number of DAG-prefixes for $G(\trpo{u})$, which is at most $\cwdth \cdot |u|^{\cwdth}$. The construction of $\calA_u$ is mostly given in Definition~\ref{def:prefix-aut}. We maintain a set $\calQ'$ of states to process. We start with $\calQ' = \{q_{I,u}\} = \{(\lambda,\lambda,\emptyset)\}$ and $\calQ_u = \emptyset$ and $\calT_u = \emptyset$. While there exists a state $q = (u',b,L) \in \calQ'$, we generate for every $a \in \Sigma$ the state $q'$ such that $(q,a,q')$ is the transition defined in Definition~\ref{def:prefix-aut}, then we add $(q,a,q')$ to $\calT_u$, next we add $q'$ to $\calQ'$ unless it is already in $\calQ_u$, finally we add $q$ to $\calQ_u$. This construction takes time $|\calT_u| \cdot |\calQ_u|$ times the maximum time required to construct $q'$. In the worst case, constructing $q'$ means going through $L$ to check whether $(a,c) \in \dep$ for any $c \in L$ (this takes time $O(|\Sigma|)$), then computing $\nf{u'a}$ and its $u$-prefix, which (this takes time $O(|G(\trpo{u})|)$), then doing a length-$O(|u|)$ word comparison and concatenating two $O(|\Sigma|)$-size sets. Thus, the runtime for constructing $\calA_u$ is at most $O( |\calT_u| \cdot |\calQ_u| \cdot (|\Sigma| + |G(\trpo{u})|)) = O( |\calT_u| \cdot |\calQ_u| \cdot (|\Sigma| + |u|^2))$.

Next, to show that $L(\calA_u) = \{w \mid \exists v \in \Sigma^*,\, \nf{w} = uv\} = \{w \mid \text{the } u\text{-prefix of } \nf{w}  \text{ is } u\}$, we prove that, for every $w \in \Sigma^*$, the state $(u',b,L)$ reached by $w$ is such that
\begin{itemize}
\item[$i.$] $u'$ is the $u$-prefix of $\nf{w}$
\item[$ii.$] $b$ is the first letter of the $u$-residual of $\nf{w}$ (or $\lambda$ if the $u$-residual is $\lambda$)
\item[$iii.$] $L$ is the set of letters of the $u$-residual of $\nf{w}$
\end{itemize}
The proof is by induction on the length of $w$. The claim holds for the base case $w = \lambda$ since the initial state is $(\lambda,\lambda,\emptyset)$. For the inductive case, suppose the claim holds for all $w' \in \Sigma^{k-1}$ and consider $w \in \Sigma^k$. Let $w = w'a$ where $w' \in \Sigma^{k-1}$. By induction, $w'$ reaches the state $p' = (u'',b',L')$ such that $i.$, $ii.$ and $iii.$ are satisfied. There is a unique transition $(p',a,p)$ for some $p := (u',b,L)$ in the automaton so $w$ reaches $p$. We next show that $i.$, $ii.$ and $iii.$ hold for $p$. For that we consider the four different cases for defining $(p',a,p)$. Let $v$ be the $u$-residual of $w'$
\begin{itemize}
\item[•] If there is $c \in L'$ such that $(a,c) \in \indep$ then by Lemma~\ref{lemma:uPrefixOne} there is $v_1 \neq \lambda$ and $v_2$ such that $v =  v_1v_2$ and $\nf{w'a}  = u''v_1av_2$ and the $u$-prefix of $\nf{w'a}$ is $u''$. The automaton defines $(u',b,L)$ as $u' = u''$, $b = b'$ and $L = L' \cup \{a\}$. $i.$ is satisfied by Lemma~\ref{lemma:uPrefixOne}; $ii.$ is satisfied since $b' = \mathit{firstLetter}(v_1v_2) = \mathit{firstLetter}(v_1av_2)$: $iii.$ is satisfied  since $L' = \mathit{letters}(v_1v_2)$.
\item[•] If there is no $c \in L'$ such that $(a,c) \in \indep$ and $\nf{u''a} = u''a$ and $b' \neq \lambda$ and $b' \lexord a$ then by Lemma~\ref{lemma:uPrefixThree} there is $v_1 \neq \lambda$ and $v_2$ such that $v =  v_1v_2$ and $\nf{w'a}  = u''v_1av_2$ and the $u$-prefix of $\nf{w'a}$ is $u''$. This is
exactly like the previous case and the automaton defines $u'$, $b$ and $L$ in the same way as before, so the inductive claim holds here too.
\item[•] If there is no $c \in L'$ such that $(a,c) \in \indep$ and $\nf{u'''a} = u''a$ and $b' = \lambda$ or $a \lexord b'$ then by Lemma~\ref{lemma:uPrefixThree} we have $\nf{w'a} = u''av$ and the $u$-prefix of $\nf{w'a}$ is the $u$-prefix of $u''a$, which is either $u''$ or $u''a$. Let $u^*$ be the $u$-prefix of $u''a$ and $v^*$ be its $u$-residual; so $\nf{u''a} = u^*v^*$. Note that $v^*$ is either $a$ or $\lambda$. The automaton defines $(u',b,L)$ as $u'' = u^*$, $b = \mathit{firstLetter}(v^*b')$ and $L = L' \cup \mathit{letters}(v^*)$. $i$ is satisfied by Lemma~\ref{lemma:uPrefixThree}; $ii.$ is satisfied since the $u$-residual of $\nf{w'a}$ is $v^*v$ and since $\mathit{firstLetter}(v^*b') = \mathit{firstLetter}(v^*v)$; $iii.$ is satisfied since $L' = \mathit{letters}(v)$.
\item[•] If there is no $c \in L'$ such that $(a,c) \in \indep$ and $\nf{u''a} \neq u''a$ then by Lemma~\ref{lemma:uPrefixTwo}, there is $u''_1 \in \Sigma^*$ and $u''_2 \in \Sigma^* \setminus \{\lambda\}$ such that $u'' = u''_1u''_2$ and $\nf{w'a} = u''_1au''_2v$. Let $u^*$ be the $u$-prefix of $u''_1au''_2$ and $v^*$ be its $u$-residual; so $\nf{w'a} = u^*v^*$. The automaton defines $(u',b,L)$ as $u'' = u^*$, $b = \mathit{firstLetter}(v^*b')$ and $L = L' \cup \mathit{letters}(v^*)$. $i$ is satisfied by Lemma~\ref{lemma:uPrefixTwo}; $ii.$ is satisfied since the $u$-residual of $\nf{w'a}$ is $v^*v$ and since $\mathit{firstLetter}(v^*b') = \mathit{firstLetter}(v^*v)$; $iii.$ is satisfied since $L' = \mathit{letters}(v)$.
\end{itemize}
Thus, the claim hold for every $w$. Since the automaton accepts exactly the word reaching a state $(u,b,L)$ for any $b$ and $L$, it accepts exactly the words whose $u$-prefix is $u$.
\end{proof}

\mainReduction*
\begin{proof}Verifying the claimed runtime and the number of calls to the FPRAS is straightforward.
Let $C(u) = |\{t \in \quot{\treq}{L_n(\calA)} \mid \exists v,\, \nf{t} = u\cdot v\}|$ and $C = C(\lambda) = |\quot{\treq}{L_{n}(\aut)}|$. 
To analyze the distribution of the output of \textsc{TraceSample}, we work with a slightly modified algorithm whose output distribution has a larger total variation distance to $U$. Let \textsc{TraceSampleCore}$^*$ be the same algorithm as \textsc{TraceSampleCore} except that $\tilde{C}(w)$ is computed before starting the sampling procedure for \emph{all} possible $w$. Instead of running lines~\ref{line:prefix-aut} and~\ref{line:sample}, \textsc{TraceSampleCore}$^*$ retrieves $\tilde{C}(w)$ whenever it needs it. This is a very expensive preprocessing time-wise and space-wise but, since each pre-computed $\tilde{C}(w)$ is either not used or used only once, the output distribution of \textsc{TraceSampleCore}$^*$ is identical to that of \textsc{TraceSampleCore}. We analyze \textsc{TraceSample}$^*$: the algorithm identical to \textsc{TraceSample} except that it calls \textsc{TraceSampleCore}$^*$.

Let $E$ be the event that, in every call to \textsc{TraceSampleCore}$^*$, every $\tilde{C}(w)$ is $(1 \pm \varepsilon)|\quot{\treq}{L_{n-|w|}(\calA_w)}|$. There are $m$ calls \textsc{TraceSampleCore}$^*$, and at most $2^{n + 1}$ values for $w$, and each value has probability at most $\delta$ not to be in the correct range.
\begin{align*}
\Pr[E] \geq 1 - 2^{n-1}\cdot m \cdot \delta'  \geq 1-\delta
\end{align*}
Let $out$ be the output of \textsc{TraceSample} and $out'_i$ be the output of the $i^\text{th}$ call to \textsc{TraceSampleCore}$^*$ in \textsc{TraceSample}. We first show that, when $E$ occurs, all calls to \textsc{TraceSampleCore}$^*$ go through the rejection test, i.e., they execute Line~\ref{line:rejection}. Let $out'$ be the result of some call to \textsc{TraceSampleCore}$^*$.
\begin{claim}\label{claim:sampler_claim_1}
$\Pr[\phi < 1/(2\tilde{C}) \land E] = 0$ and $\Pr[out' \neq \bot | E] \geq \frac{1}{4} \cdot $.
\end{claim}
\begin{proof}
Consider a possible value $t$ for $u_n$. Let $t_n = t$ and $t_i$ be the length $i$ prefix of $t$ (with $t_0 = \lambda$). Let $\phi_j = \tilde{C}(t_j)/\sum_{w \in ext(t_{j-1})} \tilde{C}(w)$. At the end of the $j^\text{th}$ iteration of the while loop, we have $\phi = \phi_1\dots \phi_j$. For every $j$, let $l_j$ be a value such that $\Pr[(u_n = t) \land (\phi_1 = l_1) \land \dots \land (\phi_n = l_n) \land (\tilde{C} = k) \land E] \neq 0$. Since $\tilde{C}(w) \in (1 \pm \varepsilon')C(w)$ holds for every $w$ we have $\sum_{w \in ext(t_{j-1})} \tilde{C}(w) \in (1 \pm \varepsilon') \sum_{w \in ext(t_{j-1})} C(w) = (1 \pm \varepsilon') C(t_{j-1})$. Therefore
$$
\frac{1 - \varepsilon'}{1 + \varepsilon'}\cdot \frac{C(t_i)}{C(t_{i-1})} 
\leq 
l_i 
\leq 
\frac{1 + \varepsilon'}{1 - \varepsilon'}\cdot \frac{C(t_i)}{C(t_{i-1})}
\quad \text{ and }\quad
\left(\frac{1 - \varepsilon'}{1 + \varepsilon'}\right)^n\cdot \frac{1}{C(t_0)} 
\leq 
l_1\dots l_n 
\leq 
\left(\frac{1 + \varepsilon'}{1 - \varepsilon'}\right)^n\cdot \frac{1}{C(t_0)} 
$$ 
Since $\varepsilon' = 1/(16n)$ we have $\left(\frac{1 + \varepsilon'}{1 - \varepsilon'}\right)^n \leq 5/4$ and $\left(\frac{1 - \varepsilon'}{1 + \varepsilon'}\right)^n \geq 3/4$. Using $C(t_0) = C(\lambda) = C$ and $k \in (1 \pm \varepsilon')C$ we find
$$
\frac{3(1- \varepsilon')}{4k} 
\leq
\frac{3}{4C} 
\leq 
l_1\dots l_n 
\leq 
\frac{5}{4C} 
\leq 
\frac{5(1 + \varepsilon')}{4k}
$$ 
For every $n \geq 1$, $3(1 - \varepsilon')/4k$ is at least $1/(2k)$. This holds for every possible $t$, $l_1,\dots,l_n$ and $k$, so $E$ implies $\phi \geq 1/(2\tilde{C})$. It follows that, if $out' = \bot$, then this is because the sample is rejected at Line~\ref{line:rejection}. Thus
\begin{align*}
\Pr[out' \neq \bot \land E] = \sum_{k,l} \Pr[out' \neq \bot \mid \tilde{C} = k \land \phi = l \land E]\cdot\Pr[\tilde{C} = k \land \phi = l \land E]
\\
= \sum_{k,l} \frac{1}{2kl}\cdot\Pr[\tilde{C} = k \land \phi = l \land E]
\geq \sum_{k,l} \frac{2}{5(1+\varepsilon')}\cdot\Pr[\tilde{C} = k \land \phi = l \land E]
\geq \frac{2\Pr[E]}{5(1+\varepsilon')} \geq \frac{\Pr[E]}{3}
\end{align*}
\end{proof}

Let $\Pi(t)$ be the distribution of $out$. The variable following the uniform distribution $U(t)$ is independent of $E$ and $\neg E$. We have to show that $\sum_t |\Pi(t) - U(t)| \leq 2\delta$.
\begin{align*}
\sum_t |\Pi(t) - U(t)| = \sum_{t} \left|(\Pr[out = t \mid \neg E] - U(t))\cdot\Pr[\neg E] 
+ (\Pr[out = t \mid E] - U(t))\cdot\Pr[E]\strut\right|
\\
\leq \sum_{t}\left|\Pr[out = t \mid \neg E] - U(t)\strut\right|\Pr[\neg E] + \sum_{t} \left|\Pr[out = t \mid E] - U(t)\strut\right|\cdot\Pr[E]
\\
\leq \Pr[\neg E] + \sum_{t} \left|\Pr[out = t \mid E] - U(t)\strut\right|\cdot\Pr[E]
\\
\leq \delta + \sum_{t} \left|\Pr[out = t \mid E] - U(t)\strut\right|
\end{align*}
We now show that $\sum_{t} \left|\Pr[out = t \mid E] - U(t)\strut\right| \leq \delta$. To that end, we first isolate the probability that \textsc{TraceSample} returns $\bot$ assuming $E$. By Claim~\ref{claim:sampler_claim_1}, 
\begin{align*}
\Pr[out = \bot \mid E] = \prod_{i = 1}^m \Pr[out'_i = \bot \mid E] \leq \left(\frac{3}{4}\right)^m \leq \delta
\end{align*}
Thus
\begin{align*}
\sum_{t} \left|\Pr[out = t \mid E] - U(t)\strut\right| = \Pr[out = \bot \mid E] + \sum_{t \neq \bot} \left|\Pr[out = t \mid E] - U(t)\strut\right| 
\\
\leq \delta + \sum_{t \neq \bot} \left|\Pr[out = t \mid E] - U(t)\strut\right|
\end{align*}
It remains to show that the right-hand side sum is $0$. Let $out = out'_i$ be the event that the output of \textsc{TraceSample}$^*$ is the output of the $i^\text{th}$ call to \textsc{TraceSampleCore}$^*$.
\begin{align*}
\sum_{t \neq \bot} \left|\Pr[out = t \mid E] - U(t)\strut\right|
\leq \sum_{t \neq \bot} \left|\sum_{i = 1}^m \Pr[out=out'_i \mid E]\cdot(\Pr[out'_i = t \mid E \land (out=out'_i)] - U(t))\right|
\\
\leq \sum_{i = 1}^m \Pr[out=out'_i \mid E] \cdot \sum_{t \neq \bot}\left|\Pr[out'_i = t \mid E \land (out=out'_i)] - U(t)\strut\right|
\end{align*}
Note that $\Pr[out'_i = t \mid E \land (out=out'_i)] = \Pr[out'_i = t \mid E \land (out'_i \neq \bot)]
$. Proving the following claim ends the proof.
\begin{claim}\label{claim:sampler_claim_2}
For $t \neq \bot$, $\Pr[out'_i = t \mid E \land (out'_i \neq \bot)] = \frac{1}{C} = U(t)$
\end{claim}
\begin{proof}
We drop the $i$ subscript. It is enough to show that, for $t \neq \bot$, $\Pr[out' = t \mid E]$ equals a constant $C$ independent of $t$. Indeed, we then have
\begin{align*}
\Pr[out' = t \mid E \land (out' \neq \bot)] = \frac{\Pr[out' = t \mid E]}{\Pr[out' \neq \bot \mid E]} 
\leq \frac{C}{\Pr[out' \neq \bot \mid E]} := D
\end{align*}
If $C$ is a constant independent of $1$ then so is $D$ and therefore, since $1 = \sum_{t \neq \bot} \Pr[out' = t \mid E \land (out' \neq \bot)] = \sum_{t \neq \bot} D$, we find that $D = 1/C$. Now, to prove $\Pr[out' = t \mid E] = C$, the first observation is that $\tilde{C}$ can only take finitely many values and that its set $S$ of possible values is independent of $t$. We reuse the notations $t_j$, $\phi_j$ from Claim~\ref{claim:sampler_claim_1}. At the end of the $j^\text{th}$ iteration of the while loop, we have $\phi = \phi_1\dots \phi_j$. Recall that in \textsc{TraceSampleCore}$^*$ the $\phi_j$ are computed before the sampling procedure starts. Let $S(k)$ be the finite set of all $n$-tuples $(l_1,\dots,l_n)$ such that $\Pr[(\phi_1 = l_1) \land \dots \land (\phi_n = l_n) \land (\tilde{C} = k) \land E]$. When $E$ occurs, every $\phi_j$ is different from $0$ and thus $\Pr[(u_n = t) \land (\phi_1 = l_1) \land \dots \land (\phi_n = l_n) \land (\tilde{C} = k) \land E] \neq 0$. For $k \in S$ and $(l_1,\dots,l_n) \in S(k)$ we have the following.
$$
\Pr[out' = t \mid (u_n = t) \land (\phi_1 = l_1) \land \dots \land (\phi_n = l_n) \land (\tilde{C} = k) \land E] = \frac{1}{2k\prod_{i = 1}^n l_i}
$$
and
$$
\Pr[u_i = t_i \mid (u_{i-1} = t_{i-1}) \land (\phi_1 = l_1) \land \dots \land (\phi_n = l_n) \land (\tilde{C} = k) \land E] = l_i
$$
Thus
\begin{align*}
\Pr[out' = t\mid E] 
= \sum_{k \in S} \Pr[\tilde{C} = k \mid E] \sum_{(l_1,\dots,l_n) \in S(k)} 
\frac{\Pr[(\phi_1 = l_1) \land \dots \land (\phi_n = l_n) \mid (\tilde{C} = k) \land E] }{2k\prod_{i = 1}^n l_i}
\\ 
\cdot \prod_{i = 1}^n \Pr[u_i = t_i \mid  (u_i-1 = t_i-1) \land (\phi_1 = l_1) \land \dots \land (\phi_n = l_n) \land  (\tilde{C} = k) \land E]
\\
= \sum_{k \in S} \Pr[\tilde{C} = k \mid E] \sum_{(l_1,\dots,l_n) \in S(k)} 
\frac{\Pr[(\phi_1 = l_1) \land \dots \land (\phi_n = l_n) \mid (\tilde{C} = k) \land E] }{2k}
= \sum_{k \in S} \frac{\Pr[\tilde{C} = k \mid E]}{2k} 
\end{align*}
This shows that $\Pr[out' = t \mid E]$ is a constant $C$ independent of $t$.
\end{proof} 
\end{proof}

\section{Proofs from \secref{fpras}}

\probaFirstOrder*
\begin{proof}
By induction on $q$. \\

\noindent
\underline{ Base case.}
If $q = q_I$ then $\vY{r}{\emptyset}{1}{q_I} = \{\eqcl{\lambda}\} = \quot{\treq}{L(q)}$ so, $\Pr[\eqcl{\lambda} \in \vY{r}{\emptyset}{1}{q_I}] = 1$ and we are done. \\

\noindent
\underline{ Inductive case.}
Now suppose $q\neq q_I$ and that the lemma holds for $q$'s predecessors.  We use the same notations as in the definitions of $\vY{r}{h}{v}{q}$. The $\union$ procedure ensures that, for every $t \in \quot{\treq}{L(q)}$, there is a unique predecessor $q_i$ of $q$ and a unique $t'\in \quot{\treq}{L(q_i)}$ such that $t \in \vZ{r}{h}{q}$ if only if $t' \in \reduce(\vY{r}{h_i}{h(q_i)}{q_i},\textstyle\frac{h(q_i)}{m_h(q)})$. Let $v_i =h(q_i)$ for convenience 
\begin{align*}
\Pr\left[t \in \vZ{r}{h}{q}\right] = \Pr\left[t' \in \reduce\left(\vY{r}{h_i}{v_i)}{q_i},\frac{v_i}{m_h(q)}\right)\right] 
= \frac{v_i}{m
_h(q)}\Pr\left[t' \in \vY{r}{h_i}{v_i}{q_i}\right]
\end{align*} 
Using the induction hypothesis, we find that $\Pr\left[t \in \vZ{r}{h}{q}\right] = m_h(q)^{-1}$. Finally,
\begin{align*}
\Pr\left[t \in \vY{r}{h}{v}{q}\right] = \Pr\left[t \in \reduce(\vZ{r}{h}{q},\frac{m_h(q)}{v})\right] &
= \frac{m_h(q)}{v} \Pr\left[t \in \vZ{r}{h}{q}\right] = v^{-1}
\end{align*} 
\end{proof}

\probaSecondOrder*
\begin{proof}
$$
\Pr[t \in \vY{r}{h}{v}{q},\, t' \in \vY{r}{h'}{v'}{q'}] =
\Pr\left[t \in \reduce\left(\vZ{r}{h}{q},\frac{m_h(q)}{v}\right),\, t' \in \reduce\left(\vZ{r}{h'}{q'},\frac{m_{h'}(q')}{v'}\right)\right] 
$$
The events of surviving of different elements surviving $\reduce$ procedures are independent. Thus
$$
\Pr[t \in \vY{r}{h}{v}{q},\, t' \in \vY{r}{h'}{v'}{q'}] =
\frac{m_h(q)m
_{h'}(q')}{v'v}\Pr[t \in \vZ{r}{h}{q}\text{ and } t' \in \vZ{r}{h'}{q'}] 
$$
Therefore, it suffices to show the bound on $\Pr[t \in \vZ{r}{h}{q}\text{ and } t' \in \vZ{r}{h'}{q'}]$ to prove the lemma. There exists a unique $q_i \in \preds(q)$ (resp. $q'_j \in \preds(q')$) a unique trace $t_i \in \quot{\treq}{L(q_i)}$ (resp. $t'_j \in \quot{\treq}{L(q'_j)}$ such that $t \in \vZ{r}{h}{q}$ (resp. $t' \in \vZ{r}{h'}{q'}$) if and only if $t_i \in \reduce(\vY{r}{h_i}{v_i}{q_i},\frac{v_i}{m_h(q)})$ (resp. $t'_j \in \reduce(\vY{r}{h'_j}{v'_j}{q'_j},\frac{v'_j}{m_{h'}(q')})$), where $h_i$ (resp. $h'_j$) is the restriction of $h$ (resp. $h'$) to the ancestors of $q_i$ (resp. $q'_j$) and $v_i = h(q_i)$ (resp. $v'_j = h'(q'_j)$).
\begin{align*}
\Pr[t \in \vZ{r}{h}{q} \text{ and } t' \in \vZ{r}{h'}{q'}] \\
= \Pr\left[t_i \in \reduce\left(\vY{r}{h_i}{v_i}{q_i},\frac{v_i}{m_h(q)}\right) \text{ and }
t'_j \in \reduce\left(\vY{r}{h'_j}{v'_j}{q'_j},\frac{v'_j}{m_{h'}(q')}\right)\right]
\end{align*}
Since $\reduce$ are independent,
\begin{align}\label{proba:p}
\Pr[t \in \vZ{r}{h}{q} \text{ and } t' \in \vZ{r}{h'}{q'}] = \frac{v_iv'_j}{m_h(q)m_{h'}(q')}\Pr[t_i \in \vY{r}{h_i}{v_i}{q_i} \text{ and }
t'_j \in \vY{r}{h'_j}{v'_j}{q'_j}]
\end{align}
Now we proceed by induction on $k$ to prove the upper bound on $\Pr[t \in \vZ{r}{h}{q} \text{ and } t' \in \vZ{r}{h'}{q'}]$.\\

\noindent
\underline{Base case.} If $k = 0$, then $q_i = q'_j = q_I$ and thus $v_i = v'_j = 1$, $h_i = h'_j =\emptyset$ and $h(q_I) = 1$. $q_I$ is the divergence state of $\run(t,q)$ and $\run(t',q')$ and~(\ref{proba:p}) reduces to $(m_h(q)m_{h'}(q'))^{-1} = h(q_I)(m_h(q)m_{h'}(q'))^{-1}$. So the bound holds in this case. For the inductive, we consider two cases.
\\

\noindent
\underline{Inductive case when $(t_i,q_i) = (t'_j,q'_j)$.} \\
Since $h$ and $h'$ are compatible, we have $h_i = h'_j$ and $v_i = v'_j$. So, using Lemma~\ref{lemma:proba_first_order},  ~(\ref{proba:p}) boils down to 
\begin{align*}
\frac{v_iv'_j}{m_h(q)m_{h'}(q')}\Pr[t'_j  \in \vY{r}{h'_j}{v'_j}{q'_j}] = \frac{v_i}{m_h(q)m_{h'}(q')}
\end{align*}
Since $(t,q) \neq (t',q')$ and $(t_i,q_i) = (t'_j,q'_j)$, the divergence state of $\run(t,q)$ and $\run(t',q')$ is $q^* = q_i$. So the bound holds in this case.
\\

\noindent
\underline{Inductive case when $(t_i,q_i) \neq (t'_j,q'_j)$.} 
\\Since $q_i$ and $q'_j$ both belong to $\calQ^{k-1}$ and $k > 1$, we use the induction hypothesis to bound~(\ref{proba:p}) from above by
\begin{align*}
\frac{v_iv'_j}{m_h(q)m_{h'}(q')} \cdot\frac{h(q^*)}{v_iv'_j} = \frac{h(q^*)}{m_h(q)m_{h'}(q')}
\end{align*}
where $q^*$ is the divergence state of $\run(t_i,q_i)$ and $\run(t'_j,q'_j)$. It remains to argue that $q^*$ is also the divergence state of $\run(t,q)$ and $\run(t',q')$, which is clear since $\run(t,q)$ extend $\run(t_i,q_i)$ and $\run(t',q')$ extends $\run(t'_j,q'_j)$. So the bound holds in this case.
\end{proof}

\coupling*

\begin{proof}
For $A$ an algorithm and $X_1,\dots,X_k$ random variables in $A$, with $\Omega_i$ the universe of $X_i$, we denote by
$$
\Pr_A\left[\bigcap_{l = 1}^k X_l = \omega_l\right]
$$
the probability that, after running $A$ (which we assume terminates), we have $X_i = \omega_i$ for every $i$ (with $\omega_i \in \Omega_i$). We rename $\TraceMCCore^*$ $A_1$. We make a sequence of modifications to $A_1$ obtaining algorithms $A_2,A_3,\dots$. For $X_1,\dots,X_k$ random variables in $A_i$ and $Y_1,\dots,Y_k$ are random variables in $A_j$ such that $X_i$ and $Y_i$ have the same universe $\Omega_i$, we write
$$
(X_1,\dots,X_k)_{A_i} \sim (Y_1,\dots,Y_k)_{A_j}
$$
to say that, for every $(\omega_1,\dots,\omega_k) \in \Omega_1 \times \dots \times \Omega_k$, we have
$$
\Pr_{A_i}\left[\bigcap_{l = 1}^k X_l = \omega_l\right] =  \Pr_{A_j}\left[\bigcap_{l = 1}^k Y_l = \omega_l\right]
$$
We may have variables that have the same name in $A_i$ and $A_j$, say $X_1,\dots,X_k$. It should be clear that $
(X_1,\dots,X_k)_{A_i} \sim (X_1,\dots,X_k)_{A_j}
$ means that the lefthand side variables are considered in $A_i$ and the righthand side variables are considered in $A_j$.

The algorithm modifications are all on done in $\estimateAndSample$. We start from the following.

\begin{algorithm}[H]\caption*{$\estimateAndSample(q)$ in $A_1$}
\begin{algorithmic}[1]
\State compute $\hat{S}^r(q)$ for all $r$ using $\{ S^r(q') \mid q' \in \preds(q)\}$
\State compute $N(q)$ using $\{\hat{S}^r(q)\}_r$
\State compute $S^r(q)$ for all $r$ using $N(q)$ and $\hat{S}^r(q)$
\end{algorithmic}
\end{algorithm}

For every state $q$, we keep the history for the ancestors of $q$. Formally, we maintain a variable $H(q)$ for every state $q$. Initially $H(q)$ is empty for all $q$. After $N(q)$ is computed, $H(q')$ is updated for all descendants $q'$ of $q$ as follow: $H(q') = H(q') \cup (p \mapsto N(q))$. Thus, when we reach $\estimateAndSample(q')$, $H(q')$ is the history for $q'$. Clearly, the variables $H(q)$ are unused for the computation of the other random variables in $A_1$.
 
\begin{algorithm}[H]\caption*{$\estimateAndSample(q)$ in $A_1$} 
\begin{algorithmic}[1]
\State compute $\hat{S}^r(q)$ for all $r$ using $\{ S^r(q') \mid q' \in \preds(q)\}$
\State compute $N(q)$ using $\{\hat{S}^r(q)\}_r$ \textcolor{red}{and update $H$}
\State compute $S^r(q)$ for all $r$ using $N(q)$ and $\hat{S}^r(q)$
\end{algorithmic}
\end{algorithm}

In $A_2$, for every realizable history $h$ for $q$, and every $r \in [\nsnt]$ and for every $v$ that is a possible candidate for $N(q)$, we have sets $\hat{S}^r_h(q)$ and $S^r_{h,v}(q)$. Initially these new sets are empty. When the $A_2$ computes the value for $N(q)$, $H(q)$ has already been set. Once $\hat{S}^r(q)$ is computed, we copy its content in $\hat{S}^r_{H(q)}(q)$ and, once $S^r(q)$ is computed, we copy its content in $S^r_{H(q),N(q)}(q)$. $A_1$ and $A_2$ construct the random variables $N(q)$, $S^r(q)$, $\hat S^r(q)$ and $H(q)$ in the same way so
$$
(H(q),N(q),(S^r(q),\hat S^r(q))_{r \in [\alpha]})_{A_1} \sim 
(H(q),N(q),(S^r(q),\hat S^r(q))_{r \in [\alpha]})_{A_2}
$$
We also have that 
$$
(H(q),N(q),(S^r(q),\hat S^r(q))_{r \in [\alpha]})_{A_2}
\sim
(H(q),N(q),(S^r_{H(q),N(q)}(q),\hat S^r_{H(q)}(q))_{r \in [\alpha]})_{A_2}
$$
\begin{algorithm}[H]\caption*{$\estimateAndSample(q)$ in $A_2$}
\begin{algorithmic}[1]
\State compute $\hat{S}^r(q)$ for all $r$ using $\{ S^r(q') \mid q' \in \preds(q)\}$
\State \textcolor{red}{copy $\hat{S}^r(q)$ to $\hat{S}^r_{H(q)}(q)$ for all $r$}
\State compute $N(q)$ using $\{\hat{S}^r(q)\}_r$ and update $H$
\State compute $S^r(q)$ for all $r$ using $N(q)$ and $\hat{S}^r(q)$
\State \textcolor{red}{copy $S^r(q)$ to $S^r_{H(q),N(q)}(q)$ for all $r$}
\end{algorithmic}
\end{algorithm}

We are going to dedfine $A_3,A_4,A_5,A_6$ and $A_7$ such that, for every $3 \leq i \leq 7$,
\begin{equation}\label{eq:same_distrib}
(H(q),N(q),(S^r_{H(q),N(q)}(q),\hat S^r_{H(q)}(q))_{r \in [\alpha]})_{A_{i-1}}
\sim
(H(q),N(q),(S^r_{H(q),N(q)}(q),\hat S^r_{H(q)}(q))_{r \in [\alpha]})_{A_i}
\end{equation}

For a fixed $q$, to compute $N(q)$, $\{\hat S^r(q)\}_r$, $A_2$ uses $\{p(q')\}_{q' \in \preds(q)}$ and $\{\{S^r(q')\}_r\}_{q' \in \preds(q)}$. But at that point, $S^r_{H(q'),N(q')}(q')$ is indentical to $S^r(q')$ so we can swap them. Similarly, to compute $\{S^r(q)\}_r$, the algorithm uses $N(q)$ and $\{\hat S^r(q)\}_r$. But at that point, $\hat S^r(q)$ and $\hat S^r_{H(q)}(q)$ are indentical.$A_3$ instead does the computation with the $S^r_{H(q'),N(q')}(q')$ and the $\hat S^r_{H(q)}(q)$
and then copies the result in $\hat S^r(q)$ and $S^r(q)$. Since the swapped sets are identical, Eq~(\ref{eq:same_distrib}) holds for $i = 3$.

\begin{algorithm}[H]
\caption*{$\estimateAndSample(q)$ in $A_3$}
\begin{algorithmic}[1]
\State compute $\hat{S}^r(q)$ for all $r$ using $\{\textcolor{red}{S^r_{H(q'),p(q')}(q')} \mid q' \in \preds(q)\}$
\State copy $\hat{S}^r(q)$ to $\hat{S}^r_{H(q)}(q)$ for all $r$
\State compute $N(q)$ using $\{\textcolor{red}{\hat{S}^r_{H(q)}(q)}\}_r$ and update $H$
\State compute $S^r(q)$ for all $r$ using $N(q)$ and \textcolor{red}{$\hat{S}_{H(q)}^r(q)$}
\State copy $S^r(q)$ to $S^r_{H(q),N(q)}(q)$ for all $r$
\end{algorithmic}
\end{algorithm}

\noindent In $A_4$ we swap $\hat{S}^r_{H(q)}(q)$ with $\hat S^r(q)$ and $S^r_{H(q),N(q)}(q)$ and $S^r(q)$. That is, first compute $\hat{S}^r_{H(q)}(q)$ (resp. $S^r_{H(q),N(q)}(q)$) and then copy its content to $\hat{S}^r(q)$ (resp. $S^r(q)$). Again, compared to $A_3$ we are just swapping the roles of sets that are anyway indentical, so Eq~\ref{eq:same_distrib} holds for $i = 4$.

\begin{algorithm}[H]\caption*{$\estimateAndSample(q)$ in $A_4$}
\begin{algorithmic}[1]
\State compute \textcolor{red}{$\hat{S}^r_{H(q)}(q)$} for all $r$ using $\{S^r_{H(q'),p(q')}(q') \mid q' \in \preds(q)\}$
\State  compute $N(q)$ using $\{\hat{S}^r_{H(q)}(q)\}_r$ and update $H$
\State  compute \textcolor{red}{$S^r_{H(q),N(q)}(q)$} for all $r$ using $\hat{S}^r_{H(q)}(q)$
\State  copy \textcolor{red}{$S^r_{H(q),N(q)}(q)$ to $S^r(q)$} for all $r$
\State  copy \textcolor{red}{$\hat{S}^r_{H(q)}(q)$ to $\hat{S}^r(q)$} for all $r$
\end{algorithmic}
\end{algorithm}

$A_4$ does not touch any sets $\hat{S}^r_{h}(q)$ or $S^r_{v,h}(q)$ when $h \neq H(q)$ and $v \neq N(q)$. $A_5$ computes $\hat{S}^r_{H(q)}(q)$ and $S^r_{N(q),H(q)}(q)$ like $A_4$ -- which ensures Eq~(\ref{eq:same_distrib}) holds for $i = 5$ -- but also computes all other $\hat{S}^r_{h}(q)$ and $S^r_{h,t}(q)$. Say $\preds(q) = (q_1,\dots,q_k)$ then $A_5$ sets $m_h(q) = \max(h(q_1),\dots,h(q_k)$ and $\hat{S}^r_{h}(q) = \union(q,\reduce(S^r_{h_1,h(q_1)}(q_1),h(q_1)/m_h(q)),\dots,\reduce(S^r_{h_k,h(q_k)}(q_k),h(q_k)/m_h(q)))$ and $S^r_{h,v}(q) = \reduce(\hat{S}^r_{h}(q),m_h(q)/v)$, where $h_i$ denote the restriction of $h$ to the ancestors of $q_i$.

\begin{algorithm}[H]\caption*{$\estimateAndSample(q)$ in $A_5$}
\begin{algorithmic}[1]
\State  compute $\hat{S}^r_{H(q)}(q)$ for all $r$ using $\{ S^r_{H(q'),p(q')}(q') \mid q' \in \preds(q)\}$
\State \textcolor{red}{compute $\hat{S}^r_{h}(q)$ for all $r$ and $h$ using $\{ \{S^r_{h,v}(q') \}_{r,h,v} \mid q' \in \preds(q)\}$}
\State  compute $N(q)$ using $\{\hat{S}^r_{H(q)}(q)\}_r$ and update $H$
\State compute $S^r_{H(q),N(q)}(q)$ for all $r$ using $\hat{S}^r_{H(q)}(q)$
\State \textcolor{red}{compute $S^r_{h,v}(q)$ for all $r$ and all $(h,v) \neq (H(q),N(q))$ using $\{\hat{S}^r_{h}(q)\}_h$}
\State copy $S^r_{H(q),N(q)}(q)$ to $S^r(q)$ for all $r$
\State copy $\hat{S}^r_{H(q)}(q)$ to $\hat{S}^r(q)$ for all $r$
\end{algorithmic}
\end{algorithm}

\noindent But then, in $A_5$, the variables $\hat{S}^r_{H(q)}(q)$ and $S^{r}_{H(q),N(q)}(q)$ are computed like any other $\hat S^r_h$ and $S^r_{h,v}$. So we define $A_6$ that directly computes all variables uniformly and have that Eq~(\ref{eq:same_distrib}) holds for $i = 6$

\begin{algorithm}[H]\caption*{$\estimateAndSample(q)$ in $A_6$}
\begin{algorithmic}[1]
\State compute $\hat{S}^r_{h}(q)$ for all $r$ \textcolor{red}{and all $h$} using $\{ S^r_{h',p(q')}(q')\}_{h',q' \in \preds(q)}$
\State compute $N(q)$ using $\{\hat{S}^r_{H(q)}(q)\}_r$ and update $H$
\State compute $S^r_{h,v}(q)$ for all $r$, \textcolor{red}{and all $h$ and $v$} using $\{\hat{S}^r_{h}(q)\}_h$
\State copy $S^r_{H(q),N(q)}(q)$ to $S^r(q)$ for all $r$
\State copy $\hat{S}^r_{H(q)}(q)$ to $\hat{S}^r(q)$ for all $r$
\end{algorithmic}
\end{algorithm}

\noindent In $A_6$, line 3 does not depend on line 2 so we can swap them, thus obtaining $A_7$. Clearly, Eq~(\ref{eq:same_distrib}) holds for $i = 7$.

\begin{algorithm}[H]\caption*{$\estimateAndSample(q)$ in $A_7$}
\begin{algorithmic}[1]
\State compute $\hat{S}^r_{h}(q)$ for all $r$ and all $h$ using $\{ S^r_{h',p(q')}(q')\}_{h',q' \in \preds(q)}$
\State compute $S^r_{h,v}(q)$ for all $r$, $h$ and $v$ using $\{\hat{S}^r_{h}(q)\}_h$
\State compute $N(q)$ using $\{\hat{S}^r_{H(q)}(q)\}_r$ and update $H$
\State copy $S^r_{H(q),N(q)}(q)$ to $S^r(q)$ for all $r$
\State copy $\hat{S}^r_{H(q)}(q)$ to $\hat{S}^r(q)$ for all $r$
\end{algorithmic}
\end{algorithm}

Finally we let $A_8$ be the same as $A_7$ without line 3, 4 and 5. Since lines 1 and 2 (for $q$) do not depend on the execution of lines 3,4,5 (for ancestors of $q$) we have that
$$
((S^r(q),\hat S^r(q))_{r \in [\alpha]})_{A_7}
\sim
((S^r_{H(q),N(q)}(q),\hat S^r_{H(q)}(q))_{r \in [\alpha]})_{A_8}
$$
\begin{algorithm}[H]\caption*{$\estimateAndSample(q)$ in $A_8$}
\begin{algorithmic}[1]
\State compute $\hat{S}^r_{h}(q)$ for all $r$ and all $h$ using $\{ S^r_{h',p(q')}(q')\}_{h',q' \in \preds(q)}$
\State compute $S^r_{h,v}(q)$ for all $r$, $h$ and $v$ using $\{\hat{S}^r_{h}(q)\}_h$
\end{algorithmic}
\end{algorithm}

We see that $A_8$ is exactly the random process, so
$$
\Pr_{A_8}[e(\hat S^1_h(q),\dots,\hat S^\nsnt_h(q))] = \Pr[e(\vZ{1}{h}{q},\dots,\vZ{\nsnt}{h}{q})]
$$
and
$$
\Pr_{A_8}[e(S^1_{h,v}(q),\dots,S^\nsnt_{h,v}(q))] = \Pr[e(\vY{1}{h}{v}{q},\dots,\vZ{\nsnt}{h}{v}{q})]
$$
Now, we have shown that 
$$
(H(q),N(q),(S^r(q),\hat S^r(q))_{r \in [\alpha]})_{A_1}
\sim
(H(q),N(q),(S^r_{H(q),N(q)}(q),\hat S^r_{H(q)}(q))_{r \in [\alpha]})_{A_7}
$$
It follows that 
$$
(H(q),N(q),(S^r(q))_{r \in [\alpha]})_{A_1}
\sim
(H(q),N(q),(S^r_{H(q),N(q)}(q))_{r \in [\alpha]})_{A_7}
$$
and 
$$
(H(q),(\hat S^r(q))_{r \in [\alpha]})_{A_1}
\sim
(H(q),(\hat S^r_{H(q)}(q))_{r \in [\alpha]})_{A_7}
$$
Therefore 
\begin{align*}
\Pr_{A_1}[H(q) = h \text{ and } N(q) = v \text{ and } e(\hat S^1(q),\dots,\hat S^\nsnt(q))] 
\\
= 
\Pr_{A_7}[H(q) = h \text{ and } N(q) = v \text{ and } e(S^1_{h,v}(q),\dots,S^\nsnt_{h,v}(q))] 
\\
\leq \Pr_{A_7}[e(S^1_{h,v}(q),\dots,S^\nsnt_{h,v}(q))] = \Pr_{A_8}[e(S^1_{h,v}(q),\dots,S^\nsnt_{h,v}(q))] 
\\
= \Pr[e(\vY{1}{h}{v}{q},\dots,\vY{\nsnt}{h}{v}{q})]
\end{align*}
and
\begin{align*}
\Pr_{A_1}[H(q) = h \text{ and } e(\hat S^1(q),\dots,\hat S^\nsnt(q))] = 
\Pr_{A_7}[H(q) = h \text{ and } e(\hat S^1_h(q),\dots,\hat S^\nsnt_h(q))] 
\\
\leq \Pr_{A_7}[e(\hat S^1_h(q),\dots,\hat S^\nsnt_h(q))] = \Pr_{A_8}[e(\hat S^1_h(q),\dots,\hat S^\nsnt_h(q))] 
\\
= \Pr[e(\vZ{1}{h}{q},\dots,\vZ{\nsnt}{h}{q})]
\end{align*}
\end{proof}

\whiteMagicLemma*
\begin{proof}
The proof is by induction on the depth of $q$. \\

\noindent \underline{Base case.} For $q = q_I$ we have that $\quot{\treq}{L(q_I)} = \{\lambda\}$, $N(q_I) = 1$ and $S^r(q_I) = \{\lambda\}$ so $\Pr[\lambda \in S^r(q)] = 1 = \Ex[\mathbbm{1}(\lambda \in S^r(q))] = \Ex[\mathbbm{1}(\lambda \in S^r(q))\cdot N(q)]$.\\

\noindent \underline{Inductive case.} Fix $q \neq q_I$. Suppose the lemma holds for the predecessors $q_1,\dots,q_k$ of $q$. Let $w = can(t,q)$. Let $V$ be all possible values for $N(q)$ and Let $V'$ be all possible values for $N_{max}(q)$. These sets are finite. Whenever a conditional probability or expectation $\Pr[A \mid B]$ or $\Ex[A \mid B]$ is undefined because $\Pr[A \mid B] = 0$, we let $\Pr[A \mid B] = 0$ and $\Ex[A \mid B] = 0$.
\begin{align*}
\Ex&[\mathbbm{1}(w \in S^r(q)) \cdot N(q)] = \sum_{v \in V} \Ex[\mathbbm{1}(w \in S^r(q)) \cdot N(q) \mid N(q) = v]\cdot\Pr[N(q) = v] 
\\
&= \sum_{v\in V} v\cdot \Ex[\mathbbm{1}(w \in S^r(q)) \mid N(q) = v]\cdot \Pr[N(q) = v]
\\
&= \sum_{v\in V} v\cdot \Pr[w \in S^r(q) \mid N(q) = v]\cdot\Pr[N(q) = v]
\\
&= \sum_{v\in V}v\cdot \Pr[w \in \reduce(\hat S^r(q), N_{max}(q)/N(q)) \mid N(q) = v]\cdot \Pr[N(q) = v]
\\
&= \sum_{v\in V}v\cdot \Pr[w \in \reduce(\hat S^r(q), N_{max}(q)/N(q)) \mid w \in \hat S^r(q),\,N(q) = v]\cdot\Pr[w \in \hat S^r(q), \,N(q) = v]
\\
&= \sum_{v\in V} \sum_{v' \in V'} v\cdot \Pr[w \in \reduce(\hat S^r(q), N_{max}(q)/N(q)) \mid w \in \hat S^r(q),\,N_{max}(q) = v',\,N(q) = v]
\\
&\qquad\cdot \Pr[w \in \hat S^r(q), \, N_{max}(q) = v',\, N(q) = v]
\\
&= \sum_{v\in V} \sum_{v' \in V'} v'\cdot\Pr[\hat S^r(q), \, N_{max}(q) = v'\, N(q) = v]
= \sum_{v' \in V'} v'\cdot \Pr[w \in \hat S^r(q), \, N_{max}(q) = v'] 
\\
&= \Ex[\mathbbm{1}(w \in \hat S^r(q)) \cdot N_{max}(q)]
\end{align*}
Let $q_j$ be the unique predecessor of $q_i$ such that $w = w'a$ and $(q_i,a,q) \in \calT^{\mathsf{u}}$ and $w' \in \quot{\treq}{L(q')}$ and $w \in \hat S^r(q)$ only if $w' \in \bar S^r(q_i,q)$. Let $V_i$ be all possible values for $N(q_i)$. This set is finite.
\begin{align*}
\Ex&[\mathbbm{1}(w \in \hat S^r(q)) \cdot N_{max}(q)]
= \sum_{v' \in V'} v'\cdot \Pr[w \in \hat S^r(q) \mid N_{max}(q) = v']\cdot \Pr[N_{max}(q) = v'] 
\\
&= \sum_{v' \in V'} v'\cdot\Pr[w' \in \bar S^r(q_i,q) \mid N_{max}(q) = v']\cdot \Pr[N_{max}(q) = v']
\\
&= \sum_{v' \in V'} v'\cdot \Pr[w' \in \reduce(S^r(q_i),N(q_i)/N_{max}(q)) \mid w' \in S^r(q_i),\, N_{max}(q) = v']
\\
&\qquad\cdot \Pr[w' \in S^r(q_i),\,N_{max}(q) = v']  
\\
&= \sum_{v' \in V'} \sum_{v_i \in V_i} v'\cdot \Pr[w' \in \reduce(S^r(q_i),N(q_i)/N_{max}(q)) \mid w' \in S^r(q_i),\,N(q_i) = v_i,\,N_{max}(q) = v']
\\
&\qquad \cdot \Pr[w' \in S^r(q_i),\,N(q_i) = v_i,\,N_{max}(q) = v']
\\
&= \sum_{v' \in V'} \sum_{v_i \in V_i} v_i \cdot \Pr[w' \in S^r(q_i),\,N(q_i) = v_i,\,N_{max}(q) = v']  
\\
&= \sum_{v_i \in V_i} v_i \cdot \Pr[w' \in S^r(q_i),\,N(q_i) = v_i,]
\\
&= \Ex[\mathbbm{1}(w' \in \hat S^r(q_i)) \cdot N(q_i)]
\end{align*}
By induction, $\Ex[\mathbbm{1}(w' \in \hat S^r(q_i)) \cdot N(q_i)] = 1$, so we are done.
\end{proof}
\end{document}